\documentclass[usenames,dvipsnames,12pt]{article}
\usepackage{graphicx}
\usepackage{fancybox,amssymb, amsthm,
graphicx,graphics,color,epsfig,subfigure,latexsym,tabularx,times,rotating,amsmath,multirow,amsfonts, amssymb, amsthm, amsbsy, amscd, bm,
  bbm,enumitem,picture,wrapfig,
verbatim,setspace, array,
lscape}
\usepackage[ruled]{algorithm2e}
\usepackage{subfigure}
\usepackage{siunitx}
\PassOptionsToPackage{hyphens}{url}\usepackage{hyperref}
\usepackage[square, comma, sort&compress, numbers]{natbib}
\usepackage{tikz}
\input{tikzlibrarybayesnet.code}

\usepackage{lscape,ctable}

\newcommand{\beginsupplement}{%
        \setcounter{table}{0}
        \renewcommand{\thetable}{S\arabic{table}}%
        \setcounter{figure}{0}
        \renewcommand{\thefigure}{S\arabic{figure}}%
     }

\usepackage[T1]{fontenc}
\usepackage[sc]{mathpazo}

\renewcommand{\algorithmcfname}{TreeBH}

\definecolor{Cerulean}{rgb}{0.0, 0.48, 0.65}

\theoremstyle{plain}
\newtheorem{thm}{Theorem}[]
\newtheorem{prop}{Proposition}[]

\newtheorem{definition}{Definition}[]
\newtheorem{lem}{Lemma}[]

\definecolor{ggreen}{rgb}{0.01, 0.75, 0.24}

\begin{document}

\begin{center}

{\bf  \large Testing hypotheses on a tree: new error rates and controlling strategies}

\vspace*{.5cm}
 Marina Bogomolov$^*$, Christine B. Peterson$^*$, Yoav Benjamini, Chiara Sabatti

\vspace*{.2cm}

 * These authors contributed equally to this work.

 \vspace*{1cm}

\end{center}

\begin{abstract}
We introduce a multiple testing procedure (TreeBH) which addresses the challenge of controlling error rates at multiple levels of resolution. Conceptually, we frame this problem as the selection of hypotheses which are organized hierarchically in a tree structure. We describe a fast algorithm for the proposed sequential procedure, and prove that it controls relevant error rates given certain assumptions on the dependence among the $p$-values.  Through simulations, we demonstrate that TreeBH offers the desired guarantees under a range of dependency structures---including  one similar to that  encountered in genome-wide association studies---and that it has the potential of gaining  power over alternative methods.
We also introduce a  modified version of TreeBH which we prove to control  the relevant error rates   under any dependency structure.
 We conclude with two case studies: we first analyze data collected as part of the Genotype-Tissue Expression (GTEx) project, which aims to characterize the genetic regulation of gene expression across multiple tissues in the human body, and secondly, data examining the relationship between the gut microbiome and colorectal cancer.
\end{abstract}

\section*{Introduction}

A common feature of current research is that measurements on a vast number of variables are gathered prior to the specification of  hypotheses of scientific interest. Only in a second stage is this bounty of data explored to identify a set of promising hypotheses among all the possible relations. This {\em modus operandi} is largely driven by a preference for comprehensive data acquisition. For example, the most commonly adopted  technology to measure gene expression does not provide us with quantifications for a few selected genes that we suspect are involved in the biological mechanism under study, but all the genes in a cell;  the routine way to evaluate
genetic variation does not target candidate loci in DNA, but assesses hundreds of thousands of sites; resequencing of the 16S rRNA gene in microbiome studies quantifies the abundance of all bacteria in a community, rather than pre-identified interesting types. In this context, statistical hypothesis testing is not an instrument for confirmatory investigation, but rather serves as an exploratory tool. It  acts as a sieve to sort out the findings that have higher chances of being replicable. This goal is formalized in the notion of false discovery rate (FDR) control. The procedure introduced in \cite{BH95} does not provide guarantees on each specific finding, but assures scientists that on average only a small portion of the discoveries they put forward will not be replicated. This paper deals with contexts where guarantees of this type are appropriate.

Another characteristic of studies that start with comprehensive measurements of a vast collection of variables is that the interpretation of the resulting discoveries is  challenging. It often relies on scientists recognizing some structure among the many tested hypotheses. For example, eQTL studies start with testing the association between millions of genetic variants and tens of thousands of genes across multiple tissues. However,  rather than reporting all of the hypotheses that have been rejected, researchers focus on aggregate discoveries,  identifying genes whose expression is subject to genetic regulation (eGenes) and SNPs that affect the gene expression of one or more  genes 
(eSNPs) across tissues \cite{GTEx17}. Similarly, in microbiome studies, rather than listing all single bacteria whose abundance has been associated with a health outcome, findings are summarized at the levels of species, families, classes, etc.\ \cite{SH14}. This practice of reporting discoveries that are at a coarser resolution than the individual hypotheses tested has important consequences.

Firstly, because FDR-controlling strategies
are crucially based on the way in which discoveries are counted,  a guarantee on the average proportion of unreplicable findings for the hypotheses at the highest resolution does not translate to a similar guarantee for discoveries derived by {\em post hoc} aggregation.
This problem has been addressed in specific contexts when the resolution of scientific interest is coarser, including
 MRI imaging \cite{PetW04, BH07, SetS15}, copy number variant detection \cite{SiegmundFDR}, and genome wide association studies \cite{BetS16}.
In other settings, 
scientists are interested in discoveries at multiple resolutions, and error control guarantees should be provided for all of these. This viewpoint motivated our previous work \cite{BB14,PetS16},  the p-filter method recently introduced by \cite{FBR15,RetJ17}, and its generalization to the context of model selection \cite{KS17}. 

Acknowledging the existing structure among the hypotheses at the testing stage rather than {\em post hoc} not only allows proper FDR control,
but also  presents an opportunity to increase power.
The authors in \cite{YetL06,Y08} described a setting where hypotheses can be
arranged along a tree, with resolution increasing with distance from
the root node.
By carrying out testing along the tree, one
might avoid having to test high resolution hypotheses when the
corresponding coarser ones are not rejected:
this could reduce the total number of tests,  
as well as allow an adaptively
higher threshold for discoveries along
branches that contain more signal.
The notion that testing hypotheses in a way that leverages their position in a directed acyclic graph can increase power is also at the basis of the work in \cite{Meinshausen08,R08,GM08,MG15,MG15b}, which, however, attempts to control the familywise error rate--a global error measure that is overly strict for the hypothesis-generating studies we are interested in.
Among the contributions aiming to control FDR,  \cite{Y08} relies on very strong independence assumptions, while \cite{LG16,LetF17,LetF17b,RetJ17b} all study strategies that lead to the control of FDR on the total discoveries, without affording interpretation of results at multiple layers of resolution.
Finally, in a selective inference framework, \cite{BB14} and \cite{HetS16}  consider the question of how to test groups of fine
resolution hypotheses when their $p$-values have been used in
selecting promising groups. This is a special version of the problem we are interested in, when there are only two possible levels of resolution.

Here, we introduce a novel framework to test families of hypotheses along
a tree with arbitrarily many levels, where ``family'' refers to
a set of scientifically related hypotheses. By extending and generalizing the results in \cite{BB14}, the procedure outlined in this paper, for the first time, both {\em allows  FDR control at multiple levels of resolution}  and {\em capitalizes on the structure of the hypotheses to increase power}. Interestingly, it is applicable  also to situations where data is gathered sequentially, allowing the researchers to test the hypotheses only in resolution order, rather than simultaneously.

The rest of the paper is organized as follows. In Section \ref{sec:treedef}, we
describe precisely the types of trees of hypotheses we are concerned with.  Section \ref{sec:errorrates} introduces
measures of global error that reflect  a selective strategy of
testing, and Section \ref{sec:procedure} contains theoretical results  on  control of these
error rates.  Section \ref{sec:sim}
presents the results of  simulations contrasting the
approach we propose to alternative methods as well as
exploring the robustness of the controlling strategy to dependence. Finally,
Section \ref{sec:casestudies} contains the results from our
procedure on data from the GTEx project and from a microbiome study, where the hypotheses are naturally organized in a multi-layer tree.

\section{A tree of families of hypotheses} \label{sec:treedef}

We consider a collection of hypotheses ${\cal F}=\{H_1,\ldots,H_m\}$ that are organized in a tree structure, reflecting logical relations among them. Each hypothesis $H_i$ has only one parent but can have multiple children, so that when one null hypothesis is true, all of its descendants are true, and when one null hypothesis is false, all of its ancestors are  false. These logical relations hold, for example, when each parent hypothesis is intersection of all its children. The level of a hypothesis corresponds to its distance from the root node: we are interested in situations where hypotheses on the same level correspond to scientific statements made at the same  resolution.  Because of the important role of levels, we explicitly include these in our notation, indicating with $\mathcal{F}_{i}^{\ell + 1}$ the family of hypotheses at level $\ell +1$ that has $H_i$ as a parent hypothesis at level $\ell$, and with $P^1(i)$ the index of the parent hypothesis of $H_i$ at level $\ell-1$.   Similarly, $P^2(i)$ is  the index of the parent of $H_{P^1(i)},$ $H_{P^{k}(i)}$ is a parent of $H_{P^{k-1}(i)},$  and $P^{\ell}(i)=0$: we refer to the hypotheses with indices in the set $\{P^1(i), \ldots, P^{\ell}(i)\}$ as the \textit{ancestor} hypotheses of $H_i.$ Figure 1 depicts the organization of hypotheses within a three-level tree.

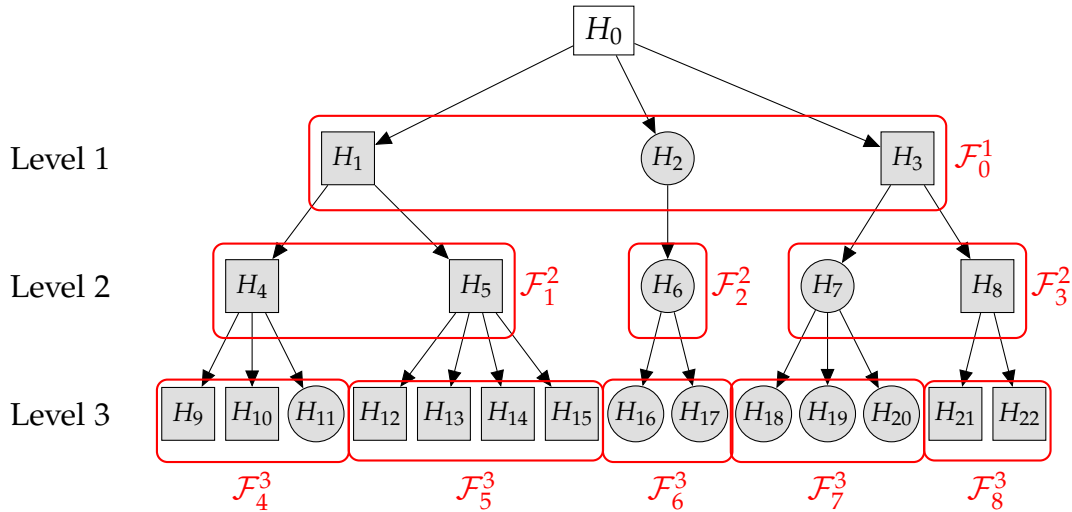
\begin{figure}[h!]
  \begin{center}
  \begin{tikzpicture}[scale=.85]
  \node[obs, rectangle, text opacity=1]  (H1)  at(0,0)   {$H_{9}$} ;
   \node[obs,  rectangle, text opacity=1]  (H2)  at(1,0)   {$H_{10}$} ;
     \node[obs,  text opacity=1]  (H3)  at(2,0)   {$H_{11}$} ;
   \node[obs, rectangle, text opacity=1]  (H4)  at(3,0)   {$H_{12}$} ;
     \node[obs, rectangle,  text opacity=1]  (H5)  at(4,0)   {$H_{13}$} ;
   \node[obs, rectangle, text opacity=1]  (H6)  at(5,0)   {$H_{14}$} ;
   \node[obs,  rectangle,text opacity=1]  (H7)  at(6,0)   {$H_{15}$} ;
   \node[obs,  text opacity=1]  (H8)  at(7,0)   {$H_{16}$} ;
     \node[obs,  text opacity=1]  (H9)  at(8,0)   {$H_{17}$} ;
   \node[obs,   text opacity=1]  (H10)  at(9,0)   {$H_{18}$} ;
     \node[obs,  text opacity=1]  (H11)  at(10,0)   {$H_{19}$} ;
   \node[obs,  text opacity=1]  (H12)  at(11,0)   {$H_{20}$} ;
      \node[obs,  rectangle,text opacity=1]  (H13)  at(12,0)   {$H_{21}$} ;
   \node[obs,  rectangle,text opacity=1]  (H14)  at(13,0)   {$H_{22}$} ;

   \node[obs, rectangle, text opacity=1]  (-1H1)  at(1,2)   {$H_{4}$} ;
     \node[obs, rectangle,  text opacity=1]  (-1H2)  at(4.5,2)   {$H_{5}$} ;
   \node[obs,  text opacity=1]  (-1H3)  at(7.5,2)   {$H_{6}$} ;
      \node[obs,  text opacity=1]  (-1H4)  at(10,2)   {$H_{7}$} ;
   \node[obs, rectangle, text opacity=1]  (-1H5)  at(12.5,2)   {$H_{8}$} ;

   \node[obs, rectangle, text opacity=1]  (-2H1)  at(2.5,4)   {$H_{1}$} ;
     \node[obs,  text opacity=1]  (-2H2)  at(7.5,4)   {$H_{2}$} ;
   \node[obs,  rectangle,text opacity=1]  (-2H3)  at(11.25,4)   {$H_{3}$} ;

 \node[inner sep=4pt,draw,rectangle]  (-3H1)  at(6.5,6)   {$H_{0}$} ;

   \edge {-1H1} {H1}; %
    \edge {-1H1} {H2}; %
   \edge {-1H1} {H3}; %

    \edge {-1H2} {H4}; %
    \edge {-1H2} {H5}; %
   \edge {-1H2} {H6}; %
   \edge {-1H2} {H7}; %

  \edge {-1H3} {H8}; %
    \edge {-1H3} {H9}; %

 \edge {-1H4} {H10}; %
    \edge {-1H4} {H11}; %
   \edge {-1H4} {H12}; %

    \edge {-1H5} {H13}; %
    \edge {-1H5} {H14}; %

 \edge {-2H1} {-1H1}; %
    \edge {-2H1} {-1H2}; %

 \edge {-2H2} {-1H3}; %

    \edge {-2H3} {-1H4}; %
 \edge {-2H3} {-1H5}; %

 \edge {-3H1} {-2H1}; %
 \edge {-3H1} {-2H2}; %
 \edge {-3H1} {-2H3}; %

 \plate [inner sep=0.15cm, yshift=.05cm,color=red,thick] {plateX} {(-2H1) (-2H2) (-2H3)} {};
  \node[text=red] at(12.3,4) {${\cal F}_0^1$};
 \plate[inner sep=0.15cm, yshift=.05cm, color=red,thick] {plateX} {(-1H1) (-1H2) }{};
 \node[text=red] at(5.5,2) {${\cal F}_1^2$};
  \plate[inner sep=0.15cm, yshift=.05cm, color=red,thick] {plateX} { (-1H3)  }{};
  \node[text=red] at(8.5,2) {${\cal F}_2^2$};
   \plate[inner sep=0.15cm, yshift=.05cm, color=red,thick] {plateX} { (-1H4) (-1H5) }{};
   \node[text=red] at(13.5,2) {${\cal F}_3^2$};
     \plate[inner sep=0.05cm, xshift=-0.0cm, yshift=0.02cm, color=red,thick] {plateX} { (H1) (H2)(H3) }{};
     \node[text=red] at(1,-1.2) {${\cal F}_4^3$};
\plate[inner sep=0.05cm, xshift=-0.0cm, yshift=0.02cm,color=red,thick] {plateX} { (H4) (H5)(H6)(H7) }{};
\node[text=red] at(4.5,-1.2) {${\cal F}_5^3$};

   \plate[inner sep=0.05cm, xshift=-0.0cm, yshift=0.02cm, color=red,thick] {plateX} { (H8) (H9)}{};
     \node[text=red] at(7.5,-1.2) {${\cal F}_6^3$};
\plate[inner sep=0.05cm, xshift=-0.0cm, yshift=0.02cm,color=red,thick] {plateX} {  (H10)(H11)(H12) }{};
\node[text=red] at(10,-1.2) {${\cal F}_7^3$};
\plate[inner sep=0.05cm, xshift=-0.0cm, yshift=0.02cm,color=red,thick] {plateX} { (H13) (H14)}{};
\node[text=red] at(12.5,-1.2) {${\cal F}_8^3$};

\node[] at(-2,0) {Level 3};
\node[] at(-2,2) {Level 2};
\node[] at(-2,4) {Level 1};

\end{tikzpicture}
\end{center}
\caption{\em Hierarchical structure of hypotheses in a three-level tree. Circles denote true null hypotheses, while squares denote false nulls. Children of the same parent constitute a family of hypotheses, identified with ${\cal F}_i^j$.}\label{testtree}
\end{figure}

In a microbiome study, for example, the root node might correspond to the hypothesis that there is no relation at all between bacterial communities in the oral cavity and incidence of caries. The hypotheses at level 1 might specialize this to each of the phyla of bacteria, those in level 2 to each of the classes, and so on, following the branching of the taxonomic tree.

We assume that valid $p$-values for all the hypotheses in the tree can be calculated on the basis of data ${\cal D}$. The testing strategy that we propose is hierarchical, starting from the coarser hypotheses and increasing in resolution (similarly to the procedure in \cite{YetL06,Y08}), so that it is applicable also to situations when the data is acquired in a sequential manner. In many applications, however, the same data is used to test all hypotheses, and the $p$-values across the tree are going to have complex dependencies. Scientists might use a variety of  strategies to calculate $p$-values for hypotheses at different levels to maximize power, or they might rely on the $p$-values for the finer-scale hypotheses and obtain the remainder with combination rules. The latter situation, for example, is typical of eQTL studies, where $p$-values for the hypotheses of no association between individual variants and the expression of specific genes are calculated and then combined to test more global hypotheses--leading for example to the discovery of eGenes (genes whose expression is regulated by DNA variants). It is not the goal of this work to investigate the most powerful tests, rather we focus on the development of multiplicity adjustment strategies that accept any collection of valid $p$-values for the hypotheses in the tree.

In our theory and simulations, aiming at a general purpose rule, we emphasize the case where each parent hypothesis is the intersection of all its children, and the $p$-values for the Level $\ell-1$ hypotheses are derived using
Simes method \cite{Simes1986} on the $p$-values of the family of
hypotheses they index at Level $\ell$, starting from the available valid Level-$L$ $p$-values and working up from the bottom level of the tree. The $p$-value for hypothesis $H_{i}$ at Level $\ell-1,$ indexing family $\mathcal{F}_i^\ell$, is obtained as
\begin{equation}\label{simes}
p = \min_{j} p_{(j)}\times \frac{k}{j},
\end{equation} where the set $\big\{p_j: j = 1, \ldots, k\big\}$ corresponds to the $p$-values for the hypotheses belonging to $\mathcal{F}_i^\ell$, and
 the index in parentheses signifies that the $p$-values have been sorted in increasing order. Simes rule  is relatively robust to dependence and has nice properties when used in conjunction with the Benjamini-Hochberg (BH) procedure \citep{BH95}, see Section \ref{sec:procedure} for more details.

\section{Error rate} \label{sec:errorrates}
Once the hypotheses have been organized in this hierarchical
structure, it is clear that discoveries can be made at multiple
levels, leaving to the researcher the choice of how to interpret this logically connected set of findings.
 By definition, FDR is a measure relative to
the total number of discoveries: depending on what we consider
the relevant total, we have a different error rate.  Yekutieli
\cite{Y08}, for example, discusses three natural choices in the case of a tree of hypotheses: the full-tree FDR
(based on counting all discoveries at each level), the
level-restricted FDR (in which one focuses on findings at a specific
level), and the outer-node FDR (which considers only  the finest
scale results reached, that is all discoveries that are not parents
to other discoveries). While the outer node
discoveries (and hence the outer-node FDR) might be the most
interesting to an individual researcher, level-specific error rates might be easier to interpret,
and the results of a procedure that controls these might be easier to
share broadly with the scientific community; for example, the results of an eQTL study are often described with a list of eGenes, and a list of gene specific variants, corresponding to discoveries at a specific level of the tree. On the other hand, separately controlling the FDR at
each level can be an inadequate way of handling multiplicity
(especially when the hierarchy includes many layers) and, more importantly, may lead to results
that lack coherence across resolutions. When adapted to our setting, the p-filter methodology,
introduced in \cite{FBR15} and not tied to hierarchical testing,
 suggests controlling the level-specific
FDRs, after imposing consistency across different levels:
coarser discoveries happen if and only if at least one of the finest
level hypotheses that they index is rejected.

We introduce here a  notion of level-specific FDR for a general multi-level tree that reflects a hierarchical order of testing where families of hypotheses at higher resolution are tested only when their parent hypotheses have been rejected. This resolves issues of consistency across levels, while preserving the appeal of marginal error measures.
Building on the work in  \cite{BB14}, our error rate assigns a specific role to the families ${\cal F}^{\ell}_{s}$ tested at level $\ell$. We suggest controlling the expectation of a weighted average false discovery proportion across the families tested at each level, with the number and identity of families tested random outcomes of the testing that took place at previous levels. Figure \ref{testtree2} gives a graphical representation of the FDPs averaged for one specific pattern of rejections: the error rate we propose is the expected value of such a random average.
\begin{figure}[h!]
  \begin{center}
  \begin{tikzpicture}[scale=.85]
  \node[obs,  fill=Cerulean!30,rectangle, text opacity=1,very thick,draw=red]  (H1)  at(0,0)   {$H_{9}$} ;
   \node[obs,  fill=Cerulean!30, rectangle, text opacity=1,very thick,draw=red]  (H2)  at(1,0)   {$H_{10}$} ;
     \node[obs,  fill=Cerulean!30, text opacity=1,very thick,draw=red]  (H3)  at(2,0)   {$H_{11}$} ;
   \node[obs,  fill=Cerulean!30,rectangle, text opacity=1,very thick,draw=red]  (H4)  at(3,0)   {$H_{12}$} ;
     \node[obs,  fill=Cerulean!30,rectangle,  text opacity=1,very thick,draw=red]  (H5)  at(4,0)   {$H_{13}$} ;
   \node[obs,  fill=Cerulean!30,rectangle, text opacity=1,very thick,draw=red]  (H6)  at(5,0)   {$H_{14}$} ;
   \node[obs,   fill=Cerulean!30,rectangle,text opacity=1,very thick,draw=red]  (H7)  at(6,0)   {$H_{15}$} ;
   \node[obs,  text opacity=1]  (H8)  at(7,0)   {$H_{16}$} ;
     \node[obs,  text opacity=1]  (H9)  at(8,0)   {$H_{17}$} ;
   \node[obs,   fill=Cerulean!30, text opacity=1]  (H10)  at(9,0)   {$H_{18}$} ;
     \node[obs,   fill=Cerulean!30,text opacity=1]  (H11)  at(10,0)   {$H_{19}$} ;
   \node[obs,   fill=Cerulean!30,text opacity=1]  (H12)  at(11,0)   {$H_{20}$} ;
      \node[obs,  fill=Cerulean!30, rectangle,text opacity=1,very thick,draw=red]  (H13)  at(12,0)   {$H_{21}$} ;
   \node[obs,   fill=Cerulean!30,rectangle,text opacity=1]  (H14)  at(13,0)   {$H_{22}$} ;

   \node[obs,  fill=Cerulean!30,rectangle, text opacity=1,very thick,draw=red]  (-1H1)  at(1,2)   {$H_{4}$} ;
     \node[obs, fill=Cerulean!30, rectangle,  text opacity=1,very thick,draw=red]  (-1H2)  at(4.5,2)   {$H_{5}$} ;
   \node[obs,  text opacity=1]  (-1H3)  at(7.5,2)   {$H_{6}$} ;
      \node[obs,  fill=Cerulean!30, text opacity=1,very thick,draw=red]  (-1H4)  at(10,2)   {$H_{7}$} ;
   \node[obs,  fill=Cerulean!30, rectangle, text opacity=1,very thick,draw=red]  (-1H5)  at(12.5,2)   {$H_{8}$} ;

   \node[obs, fill=Cerulean!30, rectangle, text opacity=1,very thick,draw=red]  (-2H1)  at(2.5,4)   {$H_{1}$} ;
     \node[obs,   fill=Cerulean!30,text opacity=1]  (-2H2)  at(7.5,4)   {$H_{2}$} ;
   \node[obs,   fill=Cerulean!30,rectangle,text opacity=1,very thick,draw=red]  (-2H3)  at(11.25,4)   {$H_{3}$} ;

 \node[inner sep=4pt,draw,rectangle]  (-3H1)  at(6.5,6)   {$H_{0}$} ;

   \edge {-1H1} {H1}; %
    \edge {-1H1} {H2}; %
   \edge {-1H1} {H3}; %

    \edge {-1H2} {H4}; %
    \edge {-1H2} {H5}; %
   \edge {-1H2} {H6}; %
   \edge {-1H2} {H7}; %

  \edge {-1H3} {H8}; %
    \edge {-1H3} {H9}; %

 \edge {-1H4} {H10}; %
    \edge {-1H4} {H11}; %
   \edge {-1H4} {H12}; %

    \edge {-1H5} {H13}; %
    \edge {-1H5} {H14}; %

 \edge {-2H1} {-1H1}; %
    \edge {-2H1} {-1H2}; %

 \edge {-2H2} {-1H3}; %

    \edge {-2H3} {-1H4}; %
 \edge {-2H3} {-1H5}; %

 \edge {-3H1} {-2H1}; %
 \edge {-3H1} {-2H2}; %
 \edge {-3H1} {-2H3}; %

     \plate[inner sep=0.05cm, xshift=-0.0cm, yshift=0.02cm, color=ggreen,thick] {plateX} { (H1) (H2)(H3) }{\textcolor{ggreen}{FDP}};

\plate[inner sep=0.05cm, xshift=-0.0cm, yshift=0.02cm,color=ggreen,thick] {plateX} { (H4) (H5)(H6)(H7) }{\textcolor{ggreen}{FDP}};

  \plate[inner sep=0.05cm, xshift=-0.0cm, yshift=0.02cm,color=ggreen,thick] {plateX} { (H13) (H14)}{\textcolor{ggreen}{FDP}};

  \plate[inner sep=0.05cm, xshift=-0.0cm, yshift=0.02cm,color=ggreen,thick] {plateX} { (H10) (H11) (H12)}{\textcolor{ggreen}{FDP}};
 \plate[inner sep=0.3cm, xshift=-0.08cm, yshift=-0.03cm,draw=ggreen,thick,dashed] {plateX} {(H1) (H2) (H3) (H4) (H5) (H6) (H7) (H8) (H9) (H10) (H11) (H12) (H13) (H14)}

\end{tikzpicture}

\vspace*{.5cm}

\begin{tikzpicture}[scale=.85]
  \node[obs, fill=Cerulean!30, rectangle, text opacity=1,very thick,draw=red]  (H1)  at(0,0)   {$0$} ;
   \node[obs,  fill=Cerulean!30, rectangle, text opacity=1,very thick,draw=red]  (H2)  at(1,0)   {$0$} ;
     \node[obs, fill=Cerulean!30,  text opacity=1,very thick,draw=red]  (H3)  at(2,0)   {$1$} ;
   \node[obs,fill=Cerulean!30,  rectangle, text opacity=1,very thick,draw=red]  (H4)  at(3,0)   {$0$} ;
     \node[obs,fill=Cerulean!30,  rectangle,  text opacity=1,very thick,draw=red]  (H5)  at(4,0)   {$0$} ;
   \node[obs, fill=Cerulean!30, rectangle, text opacity=1,very thick,draw=red]  (H6)  at(5,0)   {$0$} ;
   \node[obs, fill=Cerulean!30,  rectangle,text opacity=1,very thick,draw=red]  (H7)  at(6,0)   {$0$} ;
   \node[obs,  text opacity=1]  (H8)  at(7,0)   {$H_{8}$} ;
     \node[obs,  text opacity=1]  (H9)  at(8,0)   {$H_{9}$} ;
   \node[obs, fill=Cerulean!30,  text opacity=1 ]  (H10)  at(9,0)   {$H_{10}$} ;
     \node[obs,  fill=Cerulean!30, text opacity=1]  (H11)  at(10,0)   {$H_{11}$} ;
   \node[obs, fill=Cerulean!30,  text opacity=1]  (H12)  at(11,0)   {$H_{12}$} ;
      \node[obs,  fill=Cerulean!30, rectangle,text opacity=1,very thick,draw=red]  (H13)  at(12,0)   {$0$} ;
   \node[obs,  fill=Cerulean!30, rectangle,text opacity=1]  (H14)  at(13,0)   {$H_{14}$} ;

   \node[obs, fill=Cerulean!30, rectangle, text opacity=1,very thick,draw=red]  (-1H1)  at(1,2)   {$1/3$} ;
     \node[obs, fill=Cerulean!30, rectangle,  text opacity=1,very thick,draw=red]  (-1H2)  at(4.5,2)   {$0$} ;
   \node[obs,  text opacity=1]  (-1H3)  at(7.5,2)   {$H_{17}$} ;
      \node[obs, fill=Cerulean!30,  text opacity=1,very thick,draw=red]  (-1H4)  at(10,2)   {$0$} ;
   \node[obs, fill=Cerulean!30, rectangle, text opacity=1,very thick,draw=red]  (-1H5)  at(12.5,2)   {$0$} ;

   \node[obs, rectangle, fill=Cerulean!30, text opacity=1,very thick,draw=red]  (-2H1)  at(2.5,4)   {$1/6$} ;
     \node[obs,   fill=Cerulean!30,text opacity=1]  (-2H2)  at(7.5,4)   {$H_{21}$} ;
   \node[obs,  rectangle, fill=Cerulean!30,text opacity=1,very thick,draw=red]  (-2H3)  at(11.25,4)   {$0$} ;

 \node[inner sep=.8pt,draw,circle]  (-3H1)  at(6.5,6)   {\small $1/12$} ;

   \edge {-1H1} {H1}; %
    \edge {-1H1} {H2}; %
   \edge {-1H1} {H3}; %

    \edge {-1H2} {H4}; %
    \edge {-1H2} {H5}; %
   \edge {-1H2} {H6}; %
   \edge {-1H2} {H7}; %

  \edge {-1H3} {H8}; %
    \edge {-1H3} {H9}; %

 \edge {-1H4} {H10}; %
    \edge {-1H4} {H11}; %
   \edge {-1H4} {H12}; %

    \edge {-1H5} {H13}; %
    \edge {-1H5} {H14}; %

 \edge {-2H1} {-1H1}; %
    \edge {-2H1} {-1H2}; %

 \edge {-2H2} {-1H3}; %

    \edge {-2H3} {-1H4}; %
 \edge {-2H3} {-1H5}; %

 \edge {-3H1} {-2H1}; %
 \edge {-3H1} {-2H2}; %
 \edge {-3H1} {-2H3}; %



\end{tikzpicture}

\end{center}
\caption{\em Illustration of the proposed error rate $\text{sFDR}^{3}$ for the level 3 hypotheses, relying on the same tree of hypotheses as in Figure \ref{testtree}.
 Blue hypotheses are tested, while gray nodes indicate  not tested hypotheses. A red border distinguishes rejected hypotheses. The top display shows how $\text{sFDR}^{3}$ is the expected value of the average of the FDPs for the families of tested hypotheses in level 3, which, for   this particular realization, are indicated with green boxes. The bottom display exemplifies the  computation of the proposed error measure:  $\overline{\cal C}_j$ for each selected hypothesis is given inside the node and the realized $\text{sFDP}^{3}$ is given in the root node.}\label{testtree2}
\end{figure}
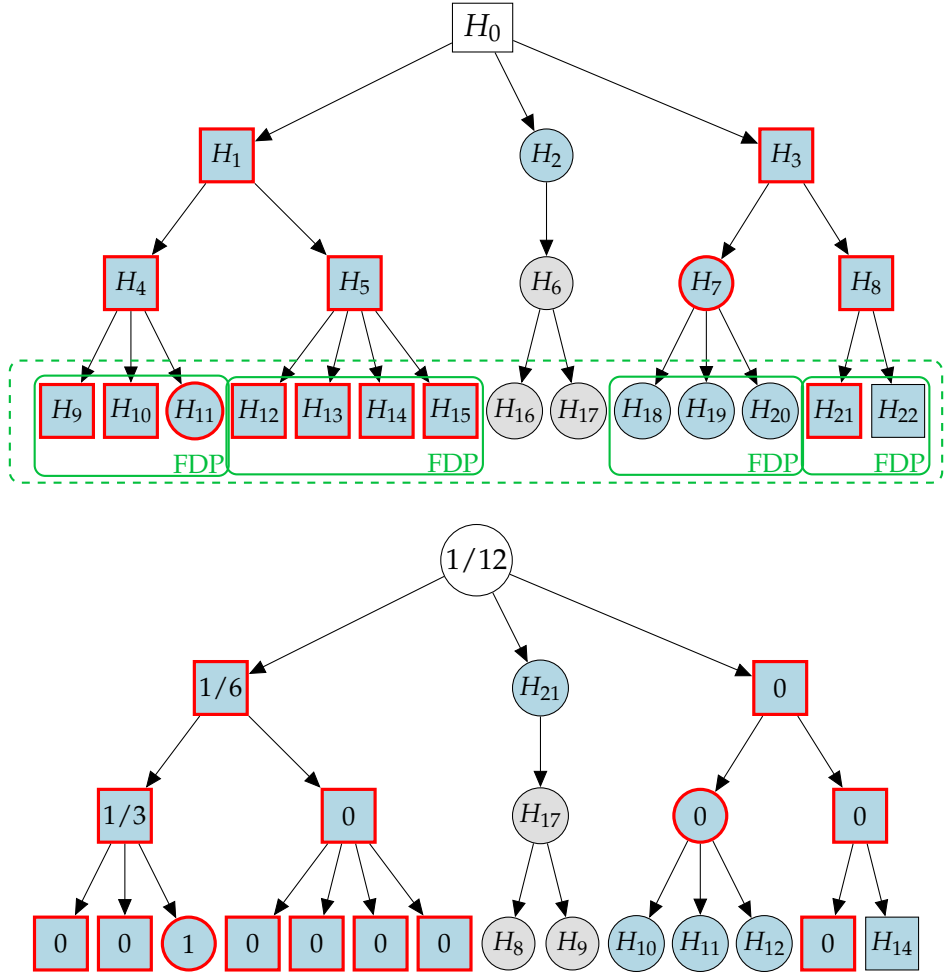

The FDP of each tested family of hypotheses is weighted depending on the amount of rejections along the path that lead to testing the family. Specifically, $\overline{\text{sFDP}}^{1}=\text{FDP}^1,$ and letting ${\cal S}^{\ell}_{i}$ denote the set of indexes of rejected hypotheses in ${\cal F}^{\ell}_i$, for $\ell>1$ we set
\begin{eqnarray}\overline{\text{sFDP}}^{\ell}&\!\!\!\!=\!\!\!\!&
\sum_{i_1\in
\mathcal{S}^1_{0}}\frac{1}{\max\{|\mathcal{S}^1_{0}|,1\}}\sum_{i_2\in
\mathcal{S}^2_{i_1}}\frac{1}{\max\{|\mathcal{S}^2_{i_1}|,1\}}\cdots\!\!\!\!\sum_{i_{\ell-1}\in
\mathcal{S}^{\ell-1}_{i_{\ell-2}}}\frac{1}{\max\{|\mathcal{S}^{\ell-1}_{i_{\ell-2}}|,1\}}\text{FDP}_{\mathcal{F}_{i_{\ell-1}}^{\ell}}\label{selectiveFDP}\\
\text{sFDR}^{\ell}&\!\!\!\!=\!\!\!\!&\mathbb{E}(\overline{\text{sFDP}}^{\ell}),\label{selectiveT}
\end{eqnarray}
where $\text{FDP}_{\mathcal{F}_{i}^{j}}$ denotes the $\text{FDP}$ within the family $\mathcal{F}_{i}^{j},$ and each of the sums in (\ref{selectiveFDP}) refers to a different level of the tree, from 1 to $\ell-1$ (the root node $H_0$ is introduced for convenience only and we consider it as always rejected, identifying $\mathcal{S}_{i_0}^1$ with $\mathcal{S}_0^1$).
We explore the properties and meaning of this error rate, considering equivalent expressions for the sum in (\ref{selectiveFDP}).

Recall that the indexes corresponding to  the ancestors of a hypothesis $H_i$ at level $\ell$ are given by  $\{P^1(i),\ldots, P^{\ell}(i)\}$, with  $H_{P^1(i)}$ a hypothesis at level $\ell -1$, $H_{P^2(i)}$ a hypothesis at level $\ell -2$, all the way till $H_{P^{\ell}(i)}=H_0$, the root node hypothesis. Then, $\overline{\text{sFDP}}^{\ell} $ (\ref{selectiveFDP}) can be expressed as
\begin{equation} \overline{\text{sFDP}}^{\ell} = \sum_{{\cal F}^{\ell}_i \text{ is tested}}w_i^{\ell} \text{FDP}_{{\cal F}^{\ell}_i} \;\;\; \text{ with } w_i^{\ell}  = \frac{1}{\prod_{k=1}^{\ell-1}|\mathcal{S}_{P^k(i)}^{\ell-k}|}, \label{weights}\end{equation}
which makes explicit how the random weights for the FDP of  the family ${\cal F}^{\ell}_i$  are  adaptive to  the number of rejections along the branch
that resulted in the selection of the family. The weight of the Level 1 family $\mathcal{F}^1$ is always 1, and therefore $\text{sFDR}^{1}$ is
equal to $\text{FDR}^{1}$, the Level 1 restricted FDR.
Similarly, $\text{sFDR}^{2}$ is equal to
the error rate introduced in \cite{BB14} of $\mathbb{E}(\sum_{i\in
\mathcal{S}^1_{0}}\text{FDP}_{\mathcal{F}_{i}^{2}}/{\max\{|{\cal S}^1_{0}|,1\}}),$ the expected average false discovery proportion
across all the selected Level 2 families. Alternative weighting schemes, such as equal weights, could also be used;
we choose to focus on adaptive weights, however, as they account for heterogeneity in signal across different branches in the tree: as the number of rejections in the path leading to family ${\cal F}^{\ell}_i$ increases, the cost of a false discovery in ${\cal F}^{\ell}_i$ decreases.

\noindent{\bf Remark 1} \quad
Note that the definition \eqref{selectiveT} can be extended to consider error measures other than FDP.
 Specifically, for
each selected family $\mathcal{F}_i^\ell,$ one may replace
$\text{FDP}_{\mathcal{F}_i^\ell}$ by an error measure $\mathcal{C}_i^{\ell}$
such that $\mathbb{E}(\mathcal{C}_i^{\ell})$ is a known error rate, such as
weighted FDR, FWER or
$\mathbb{P}(\text{FDP}>\gamma)$ \cite{BB14}.
The procedure and theory for this case is given in Section \ref{sec:procedure} and in the Supplementary Material.

\noindent{\bf Remark 2} \quad The usual FDR for one family of hypotheses can be expressed as  the expected
average across all the type I error indicators for the selected
hypotheses. Interestingly,  $\text{sFDR}^{\ell}$ admits a similar interpretation, and therefore  can be viewed as a
generalization of the simple FDR to a tree of hypotheses. To make this apparent,  we
note that we can  define $\overline{\text{sFDR}}^\ell$ \eqref{selectiveFDP} proceeding from level $\ell$ to level 0.
Assign to
each rejected hypothesis at level $\ell$ an error measure which is
the type I error indicator for this hypothesis. Now, starting from
level $\ell-1$ and continuing to coarser resolution levels, each
selected parent hypothesis is assigned an error measure which is the
average of the error measures of all the selected hypotheses in
the family it indexes, until one reaches the root of the tree,
$H_{0}.$ 
Then $\text{sFDR}^\ell$ is the
expectation of the error measure assigned to the root hypothesis.
Formally,
\begin{align*}
&\text{For each selected hypothesis $H_i$ at level $\ell,$}\,\,\,
\overline{\mathcal{C}}_{i}=\textbf{I}_{\{\text{$H_i$ is
null}\}}.\\&\text{For each selected hypothesis $H_j$ at level
$k\in\{0,\ldots,\ell-1\},$}\,\,\,
\overline{\mathcal{C}}_{j}=\frac{\sum_{r\in
\mathcal{S}_{j}^{k+1}}\overline{\mathcal{C}}_{r}}{\max\{|\mathcal{S}_{j}^{k+1}|, 1\}}.
\\&\text{sFDR}^{\ell}=\mathbb{E}(\overline{\mathcal{C}}_{0}).
\end{align*}
We illustrate this recursive computation in  the bottom portion of Figure \ref{testtree2}, using the same tree structure as in Figure \ref{testtree}.
This formulation makes it clear that for $\ell=1,$ $\text{sFDR}^{\ell}$ is the
simple FDR for $\mathcal{F}^1.$ 

\noindent{\bf Remark 3} \quad The error rate $\text{sFDR}^{\ell}$ assigns an important role to {\em families}:
hypotheses in the same family are considered as addressing the same
scientific question, and guarantees on the family-specific FDRs are
achieved. Similar to the FDR, which is adaptive to the amount of signal within
a family, becoming less stringent when there are many effects, the
$\text{sFDR}^{\ell}$ is adaptive to the amount of signal in the
selected families at level $\ell$ and their ancestor families. 

\section{Testing strategy} \label{sec:procedure}
\subsection{The TreeBH procedure}\label{sec:TreeBH}
TreeBH is
 a hierarchical testing strategy, which  targets the control of
$\text{sFDR}^{\ell}$ to bound $q^{(\ell)}$ for all
$\ell\in\{1,\ldots, L\}.$
The name  reflects the fact that it extends the Benjamini-Hochberg procedure  \citep{BH95} to the context of tree-structured hypotheses.
Testing is conducted in a step-wise fashion, starting from level 1. 
At each level of the tree, a multiplicity adjustment procedure is applied separately to each selected family of hypotheses, with a target bound that is more stringent, the fewer rejections were made in previous steps. Precisely,
 {\em
\begin{description}
\item[Step 1.] Apply the BH procedure with the bound $q^{(1)}$ to the family $\mathcal{F}^1$ $p$-values. If $\mathcal{S}^{1}_{0}=\phi,$ i.e.\ no hypothesis is rejected in $\mathcal{F}^1$, then stop. Otherwise, proceed to Step 2.
\item[Step \boldmath$\ell>1.$] For each family $\mathcal{F}_i^{\ell}$ for which all the ancestors are rejected, i.e.\ $i\in \mathcal{S}_{P^1(i)}^{\ell-1},$ $P^k(i)\in \mathcal{S}_{P^{k+1}(i)}^{\ell-k-1}$ for $k=1,\ldots, \ell-2,$ apply the BH procedure to the family $\mathcal{F}_i^{\ell}$ $p$-values with the bound
    \begin{align}q_i= \Bigg\{\prod_{k=1}^{\ell-1} \frac{|\mathcal{S}_{P^k(i)}^{\ell-k}|}{|\mathcal{F}_{P^k(i)}^{\ell-k}|}\Bigg\}q^{(\ell)}=q_{P^1(i)}\frac{|\mathcal{S}_{P^1(i)}^{\ell-1}|}{|\mathcal{F}_{P^1(i)}^{\ell-1}|}\frac{q^{(\ell)}}{q^{(\ell-1)}},\label{qitreebh}\end{align}
    i.e.\ the targeted bound $q^{(\ell)}$ multiplied by the product of proportions of rejected hypotheses within each of the families that contain ancestors of the family $\mathcal{F}_i^{\ell}.$
 \\ If  no rejections are made or $\ell=L$, stop. Otherwise, proceed to step $\ell+1$.
\end{description}}
\noindent
Note that the TreeBH procedure allows the data to be gathered sequentially, i.e.\ the data for the child hypotheses can be gathered only after their parent hypothesis has been rejected. 
 We have implemented the TreeBH procedure (using Simes' or Fisher's methods for obtaining Level $\ell-1$ $p$-values from Level $\ell$ $p$-values) in the R package TreeBH, available online at \url{http://odin.mdacc.tmc.edu/~cbpeterson/software}, and also included in the Supplementary Material.

Interestingly, if the targeted $\text{sFDR}^{\ell}$ levels satisfy $q^{(1)}\leq q^{(2)}\leq \ldots\leq q^{(L)}$ (in particular if they are all equal, i.e.\ $q^{(\ell)}=q$ for $\ell=1,\ldots, L$), and the $p$-values for the hypotheses
at Level $\ell-1$ are obtained using Simes method (\ref{simes}) on the $p$-values of the family of hypotheses they index at Level $\ell,$ working up from the bottom level of the tree,
then the TreeBH procedure satisfies the following consonance property: in each tested family at level $\ell>1$ at least one
rejection is made. This follows from the fact that for each $\ell>1$ and $H_i$ at Level $\ell-1,$ $\mathcal{F}_i^{\ell}$ is tested only if its parent hypothesis $H_i$ is rejected, i.e.\ if the Simes' $p$-value for the intersection of the hypotheses in  $\mathcal{F}_i^{\ell}$ is below the rejection threshold for the family $\mathcal{F}_{P^1(i)}^{\ell-1}.$ It is easy to see that if $q^{(\ell-1)}\leq q^{(\ell)},$ then the rejection threshold for the family $\mathcal{F}_{P^1(i)}^{\ell-1}$ is smaller or equal to the BH bound $q_i$ for family $\mathcal{F}_i^{\ell}$ (the equality holds if $q^{(\ell-1)}=q^{(\ell)}$). Since Simes' $p$-value for the intersection of the hypotheses in  $\mathcal{F}_i^{\ell}$ is the minimal BH-adjusted $p$-value in this family, we obtain that if $\mathcal{F}_i^{\ell}$ is tested, its minimal BH-adjusted $p$-value is bounded above by the BH bound $q_i$ for this family, therefore at least one rejection is made in $\mathcal{F}_i^{\ell}.$

%
%

Note that if one is interested in discoveries only up to some
specific level of resolution, say Level $k<L$, then one should apply
the TreeBH procedure to the subtree consisting of the first $k$
levels of the original tree.
Moreover, if one is interested only in discoveries
within families at some specific pre-specified levels, one may apply almost any selection
rules for selecting hypotheses within the preceding levels,
leading to
selection of families within the levels of interest (similarly to the procedure for a two-level tree in \cite{BB14}). Only the
families at those levels should be tested at the adjusted level
corresponding to the number of selected hypotheses in the preceding
levels.

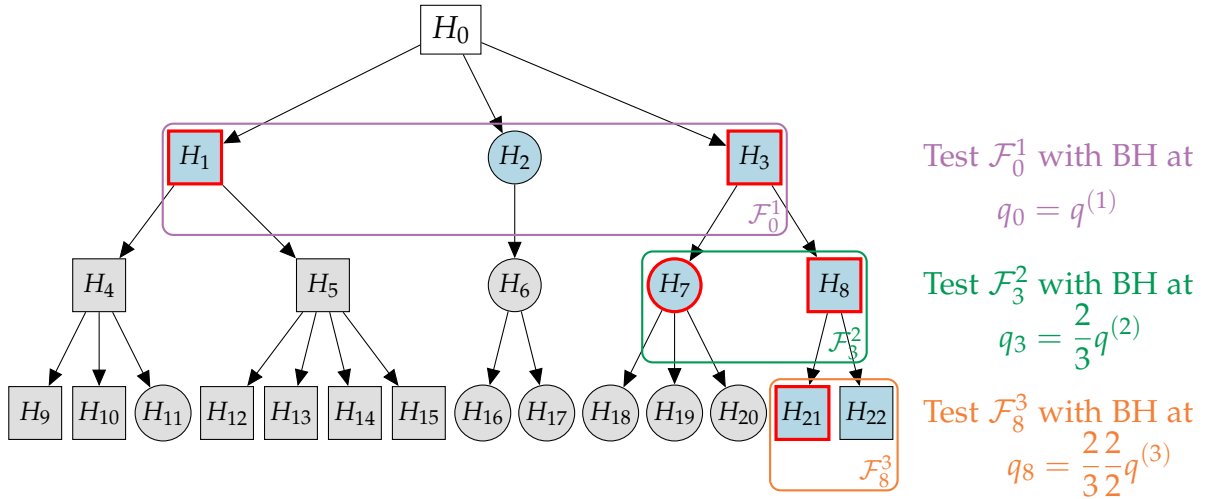
\begin{figure}[h!]
  \begin{center}
  \begin{tikzpicture}[scale=.85]
  \node[obs,  rectangle, text opacity=1]  (H1)  at(0,0)   {$H_{9}$} ;
   \node[obs,  rectangle, text opacity=1]  (H2)  at(1,0)   {$H_{10}$} ;
     \node[obs,  text opacity=1]  (H3)  at(2,0)   {$H_{11}$} ;
   \node[obs, rectangle, text opacity=1]  (H4)  at(3,0)   {$H_{12}$} ;
     \node[obs,  rectangle,  text opacity=1]  (H5)  at(4,0)   {$H_{13}$} ;
   \node[obs,  rectangle, text opacity=1]  (H6)  at(5,0)   {$H_{14}$} ;
   \node[obs,   rectangle,text opacity=1]  (H7)  at(6,0)   {$H_{15}$} ;
   \node[obs,  text opacity=1]  (H8)  at(7,0)   {$H_{16}$} ;
     \node[obs,  text opacity=1]  (H9)  at(8,0)   {$H_{17}$} ;
   \node[obs,    text opacity=1]  (H10)  at(9,0)   {$H_{18}$} ;
     \node[obs,   text opacity=1]  (H11)  at(10,0)   {$H_{19}$} ;
   \node[obs,   text opacity=1]  (H12)  at(11,0)   {$H_{20}$} ;
      \node[obs,  fill=Cerulean!30, rectangle,text opacity=1,very thick,draw=red]  (H13)  at(12,0)   {$H_{21}$} ;
   \node[obs,   fill=Cerulean!30,rectangle,text opacity=1]  (H14)  at(13,0)   {$H_{22}$} ;

   \node[obs, rectangle, text opacity=1]  (-1H1)  at(1,2)   {$H_{4}$} ;
     \node[obs,  rectangle,  text opacity=1]  (-1H2)  at(4.5,2)   {$H_{5}$} ;
   \node[obs,  text opacity=1]  (-1H3)  at(7.5,2)   {$H_{6}$} ;
      \node[obs,  fill=Cerulean!30, text opacity=1,very thick,draw=red]  (-1H4)  at(10,2)   {$H_{7}$} ;
   \node[obs,  fill=Cerulean!30, rectangle, text opacity=1,very thick,draw=red]  (-1H5)  at(12.5,2)   {$H_{8}$} ;

   \node[obs, fill=Cerulean!30, rectangle, text opacity=1,very thick,draw=red]  (-2H1)  at(2.5,4)   {$H_{1}$} ;
     \node[obs,   fill=Cerulean!30,text opacity=1]  (-2H2)  at(7.5,4)   {$H_{2}$} ;
   \node[obs,   fill=Cerulean!30,rectangle,text opacity=1,very thick,draw=red]  (-2H3)  at(11.25,4)   {$H_{3}$} ;

 \node[inner sep=4pt,draw,rectangle]  (-3H1)  at(6.5,6)   {$H_{0}$} ;

   \edge {-1H1} {H1}; %
    \edge {-1H1} {H2}; %
   \edge {-1H1} {H3}; %

    \edge {-1H2} {H4}; %
    \edge {-1H2} {H5}; %
   \edge {-1H2} {H6}; %
   \edge {-1H2} {H7}; %

  \edge {-1H3} {H8}; %
    \edge {-1H3} {H9}; %

 \edge {-1H4} {H10}; %
    \edge {-1H4} {H11}; %
   \edge {-1H4} {H12}; %

    \edge {-1H5} {H13}; %
    \edge {-1H5} {H14}; %

 \edge {-2H1} {-1H1}; %
    \edge {-2H1} {-1H2}; %

 \edge {-2H2} {-1H3}; %

    \edge {-2H3} {-1H4}; %
 \edge {-2H3} {-1H5}; %

 \edge {-3H1} {-2H1}; %
 \edge {-3H1} {-2H2}; %
 \edge {-3H1} {-2H3}; %

     \plate[inner sep=0.05cm, xshift=-0.0cm, yshift=0.02cm, color=Orchid,thick] {plateX} { (-2H1) (-2H2)(-2H3) }{\textcolor{Orchid}{${\cal F}_0^1$}};

\plate[inner sep=0.05cm, xshift=-0.0cm, yshift=0.02cm,color=ForestGreen,thick] {plateX} { (-1H4) (-1H5) }{\textcolor{ForestGreen}{${\cal F}_3^2$}};
\plate[inner sep=0.05cm, xshift=-0.0cm, yshift=0.02cm,color=Orange,thick] {plateX} { (H13) (H14) }{\textcolor{Orange}{${\cal F}_8^3$}};
  \node[text=Orchid] at(16,4) {Test ${\cal F}_0^1$ with BH at };
 \node[text=Orchid] at(16,3.2) {$q_0=q^{(1)}$};

 \node[text=ForestGreen] at(16,2) {Test ${\cal F}_3^2$ with BH at };
 \node[text=ForestGreen] at(16.2,1.2) {$q_3=\displaystyle\frac{2}{3}q^{(2)}$};

 \node[text=Orange] at(16,0) {Test ${\cal F}_8^3$ with BH at};
 \node[text=Orange] at(16.5,-.8) {$q_8=\displaystyle\frac{2}{3}\frac{2}{2}q^{(3)}$};

\end{tikzpicture}
\end{center}
\caption{Illustration of the testing strategy for a sequence of families in the tree from Figure \ref{testtree2}. We follow the same conventions, but focus only on the testing of three  families hierarchically connected.} \label{testing}
\end{figure}
Finally, as we discussed above, the error rate
$\text{sFDR}^{\ell}$ can be extended to address error measures other than FDP within the selected families. 
 If $\mathcal{C}_i^\ell$
is the error measure assigned to $\mathcal{F}_i^{\ell},$ then one
should apply an $\mathbb{E}(\mathcal{C}_i^{\ell})$-controlling  procedure
(instead of the BH procedure) with bound $q_i$ on the $p$-values of
$\mathcal{F}_i^{\ell}$ in each step of the algorithm. Thus for control of $\text{sFDR}^{\ell},$ one may improve the power of TreeBH by applying different variants of the BH procedure
within the selected families, possibly incorporating prior knowledge weights on the hypotheses (\cite{GRW06}), and/or the estimator of the proportion of nulls within a family (\cite{SetS04}, \cite{RetJ17}).
Below we give
the theoretical results for the TreeBH procedure and some of the generalizations discussed above. The remaining generalizations
and the corresponding theoretical results are given in
the Supplementary Material.

\subsection{Error rate control}\label{sec:errorrate}
The fact that the $\text{sFDR}^{\ell}$ incorporates information on the hierarchical order of testing leads to coherence between the findings at different levels: a discovery at a finer scale can happen only if the corresponding coarser hypothesis has been rejected. If we impose an additional consonance condition, requiring that whenever a parent hypothesis is rejected at least one of the hypotheses in the family it indexes is rejected, we can obtain the following result that derives control of $\text{sFDR}^{\ell-1}$ from the control of $\text{sFDR}^{\ell}$ at the higher level for a general hierarchical testing procedure.
\begin{prop}\label{prop-gen} Assume that for a given
$\ell\in\{2,\ldots,L\}$ and a given hierarchical testing procedure, it is guaranteed
that if a parent hypothesis at level $\ell-1$ is rejected, at least
one of the hypotheses within the family it indexes is rejected,
i.e.\ for each selected $H_i$ at level $\ell-1,$
$\mathcal{S}_i^{\ell}$ is not empty. Then $\text{sFDR}^{\ell-1}\leq
\text{sFDR}^{\ell},$ i.e.\ control of $\text{sFDR}^\ell$ guarantees
control of $\text{sFDR}^{\ell-1}$.
\end{prop}
The proof of this proposition follows from the recursive definition
of $\text{sFDR}^{\ell}$ above, and is deferred to the Supplementary Material. This
proposition implies that if a certain hierarchical testing procedure
guarantees for all the levels $\ell=1,\ldots, L-1,$ that each
rejected parent hypothesis has at least one rejected child
hypothesis, then control of $\text{sFDR}^{\ell}$ implies control of
$\text{sFDR}^k$ for all $k=1,\ldots, \ell-1.$ In particular,
controlling $\text{sFDR}^L$ guarantees control of
$\text{sFDR}^{\ell}$ for all the levels of the tree, so that one has
confidence regarding the errors within the selected families in each
level of the tree. Recall that the TreeBH procedure satisfies the above consonance property when $q^{(1)}\leq q^{(2)}\ldots\leq q^{(L)}$ and the $p$-value for each parent hypothesis is obtained using Simes' combination rule applied on the $p$-values of its children.


To state the conditions under which we prove that TreeBH controls the target error rate, recall the following definitions.
A set $D$ is called increasing if $x\in D$ and $y\geq x$ implies that $y\in D$ as well.
\begin{definition}[\cite{BY01}]
The vector $X$ is PRDS on $I_0$ if for any increasing set $D$ and for each $i\in I_0,$ $\mathbb{P}\left(X\in D\mid X_i=x \right)$ is nondecreasing in $x.$ The vector $X$ is PRDS if it is PRDS on any subset.
\end{definition}

\begin{thm}\label{thm-prds}
 Assume that we have valid $p$-values for the hypotheses in the tree. The TreeBH procedure with input parameters $(q^{(1)}, \ldots, q^{(L)})$ guarantees for each $\ell\in \{1,\ldots, L\}$ that
$\text{sFDR}^{\ell}\leq q^{(\ell)}$ if the following two conditions are satisfied:
\begin{enumerate}
\item The $p$-values for the hypotheses in $\mathcal{F}^1$ satisfy the PRDS property on the subset of $p$-values corresponding to true null hypotheses;
\item  For each $\ell\in \{2, \ldots, L\}$ and $H_i$ at Level $\ell-1,$ the $p$-values corresponding to the hypotheses in the set $\mathcal{F}_{i}^{\ell}\cup\left[\cup_{k=1}^{\ell-1}\left(\mathcal{F}_{P^k(i)}^{\ell-k}\setminus\{H_{P^{k-1}(i)}\}\right)\right],$ where $P^0(i)=i$, satisfy the PRDS property on the subset of $p$-values corresponding to true null hypotheses in $\mathcal{F}_{i}^{\ell}.$
\end{enumerate}
\end{thm}
The proof of this theorem is given in the Supplementary Material. Essentially, the dependency structure above requires positive dependency between the $p$-values in each family in the tree, and positive dependency between each $p$-value and the $p$-values of hypotheses at higher levels of the tree, which are not its ancestors, but belong to the same families as its ancestors. If we replace the latter positive dependency assumption by the stronger independence assumption, requiring that the $p$-value in each family is independent of any $p$-value belonging to the same family as one of its ancestors, excluding the ancestors themselves, we obtain a more general result below.

  Consider a selective error rate which addresses an error measure $\mathcal{C}_i^{\ell}$ within each selected family $\mathcal{F}_i^{\ell},$ such that $\mathbb{E}(\mathcal{C}_i^{\ell})$ is a known error rate, such as FWER or the weighted FDR (see Remark 1). When $\mathcal{C}_i^{\ell}$ is the FDP, the selective error rate is $\text{sFDR}^{\ell}$, as defined in \ref{selectiveT}, otherwise we consider its extension, where $\text{FDP}_{\mathcal{F}_i^{\ell}}$ in (\ref{selectiveFDP}) is replaced by $\mathcal{C}_i^{\ell}.$ As noted in Section \ref{sec:TreeBH}, the TreeBH procedure can be easily generalized to target this selective error rate: in each step of the algorithm, each family $\mathcal{F}_i^{\ell}$ should be tested using an $\mathbb{E}(\mathcal{C}_i^{\ell})$-controlling procedure at level $q_i$. In order to prove that the generalized procedure above controls the targeted selective error rate, we make the following assumptions:
 \begin{enumerate}
\item [A1] For each family $\mathcal{F}_i^{\ell},$  the error rate $\mathbb{E}(\mathcal{C}_i^{\ell})$ is such that $\mathcal{C}_i^{\ell}$ takes values in a countable set, and the procedure used for testing this family can control  $\mathbb{E}(\mathcal{C}_i^{\ell})$  at any desired level under the dependence among the $p$-values in $\mathcal{F}_i^{\ell}.$
\item [A2] For each $\ell\in \{2, \ldots, L\}$ and a family $\mathcal{F}_i^{\ell},$  the $p$-values for the hypotheses in $\mathcal{F}_i^{\ell}$ are independent of the $p$-values for the hypotheses in $\cup_{k=1}^{\ell-1}\left(\mathcal{F}_{P^k(i)}^{\ell-k}\setminus \{H_{P^{k-1}(i)}\}\right),$ where $P^0(i)=i.$
\item [A3] For each family $\mathcal{F}_i^{\ell},$ the procedure used for testing this family defines a simple selection rule, i.e.
    for each rejected hypothesis $H_j\in \mathcal{F}_i^\ell,$ when the $p$-values belonging to $\mathcal{F}_i^\ell \setminus \{H_j\}$ are fixed,
and the $p$-value for $H_j$ can change as long as $H_j$ is rejected,
the number of rejected hypotheses in $\mathcal{F}_i^\ell$ remains unchanged.
\end{enumerate}
 The requirement in A1 and in A3 on the error rates and the testing procedures are very lenient requirements: for all practical purposes $\mathcal{C}_i^{\ell}$ is a count or a ratio of counts, thus it has a countable support, and any non-adaptive step-up or step-down multiple testing procedure defines a simple selection rule, see \cite{BB14}. It is easy to see that the adaptive BH procedure incorporating Storey's estimator for the proportion of true nulls \cite{SetS04} is also simple. A general definition of a simple selection rule, following \cite{BB14}, is given in the Supplementary Material.
 \begin{thm}\label{thm-general}
 If assumptions A1--A3 hold, the generalized procedure above with input parameters $(q^{(1)}, \ldots, q^{(L)})$ guarantees for each $\ell\in \{1,\ldots, L\}$ that
$\mathbb{E}(s\mathcal{C}^{\ell})\leq q^{(\ell)},$
where $s\mathcal{C}^{\ell}$ is the selective error measure obtained by replacing  $\text{FDP}_{\mathcal{F}_i^{\ell}}$  by $\mathcal{C}_i^{\ell}$ in (\ref{selectiveFDP}).
 \end{thm}

 When the $p$-values for the hypotheses in each family are PRDS on the subset of $p$-values corresponding to true null hypotheses, and the dependency structure across the $p$-values of hypotheses at different levels is such that assumption A2 is satisfied,
 the TreeBH procedure satisfies the assumptions of Theorem \ref{thm-general}, as well of Theorem \ref{thm-prds}. However, the result of Theorem \ref{thm-prds} for the TreeBH procedure is stronger, since it allows to replace the independence assumption of A2 by the corresponding positive dependence assumption. Both theorems above do not require any special dependency structure between the $p$-values of hypotheses in the same branch of the tree, i.e.\ ancestor hypotheses
of each $H_i$ residing at Level $L$: these can be arbitrarily dependent.

The result of Theorem \ref{thm-general} follows from a more general theorem given in the Supplementary Material, which considers a further generalization of TreeBH, where different selection rules, not necessarily defined by multiple testing procedures, may be applied for selecting hypotheses within the families.

Let us now consider the setting we focus on, where each parent hypothesis is the intersection
of all the hypotheses in the family it indexes (i.e.\ the global null), and the $p$-values for the hypotheses
at Level $\ell-1$ are obtained using Simes' method (\ref{simes}) on the $p$-values of the family of hypotheses they index at Level $\ell,$ starting from valid Level-$L$ $p$-values and working up from the bottom level of the tree. Consider the case where the $p$-values within each family at level $L$ are PRDS, and they are independent
of the $p$-values in any other family at Level $L.$ Let us first show that in this case the $p$-values for all the hypotheses in the tree are valid. It follows from Theorem 1.2 in \cite{BY01} that when Simes' combination rule is applied on the set of $p$-values which are PRDS, the combined $p$-value is a valid $p$-value for the intersection of the corresponding hypotheses. This result also implies that Simes' $p$-value is valid when applied on independent $p$-values (which was proven in \cite{S86}). The above dependency structure implies that the $p$-values within each Level-$L$ family are PRDS, and the $p$-values within each family at a higher level are independent. Therefore, indeed Simes' $p$-value for each hypothesis in the tree is valid, yielding that the $p$-values for all the hypotheses in the tree are valid. Using this fact we obtain that in this case the assumptions of Theorem \ref{thm-general} (as well as those of Theorem \ref{thm-prds}) hold, therefore the TreeBH procedure guarantees control of $\text{sFDR}^{\ell}$ for each $\ell\in\{1, \ldots, L\}.$

Let us now consider the same setting as before, but a more general dependency structure at Level $L,$ where the entire set of Level-$L$ $p$-values satisfies the PRDS property, i.e.\ the $p$-values across different families are not necessarily independent: they may be positively dependent. In this case we do not have the result that Simes' $p$-value is valid for each hypothesis in the tree, and that the dependency structure assumed in Theorem \ref{thm-prds} or in Theorem \ref{thm-general} holds. However, we obtain that in this case the TreeBH procedure controls $\text{sFDR}^{L},$ i.e.\ the selective error rate for the finest Level-$L$ families, and it also controls $\text{sFDR}^{\ell}$ for each $\ell<L$ if the target input parameters are all equal.
\begin{thm}
Assume that each parent hypothesis is the intersection of the hypotheses within the family it indexes. Assume that the $p$-values for the hypotheses at the finest Level $L$ satisfy the PRDS property, and each of the $p$-values for the hypotheses at Level $\ell-1$ is calculated using a certain combination rule on the $p$-values of the family of hypotheses it indexes at Level $\ell,$ for $\ell=2, \ldots, L.$
\begin{enumerate}
\item If each of the combination rules used for computing $p$-values at levels $1, \ldots, L-1$ is a monotone non-decreasing function of each of its input $p$-values, then the TreeBH procedure guarantees $\text{sFDR}^L\leq q^{(L)}$ where $q^{(L)}$ is the targeted $\text{sFDR}^{L}$ input parameter.
\item If each of the combination rules used for computing $p$-values at levels $1, \ldots, L-1$ is Simes' combination rule (\ref{simes}), and the targeted $\text{sFDR}^{\ell}$ levels are equal, i.e. they satisfy $q^{(1)}= q^{(2)}= \ldots= q^{(L)}=q$ for some $q\in[0,1],$
then the TreeBH procedure guarantees for each $\ell\in \{1,\ldots, L\}$ that
$\text{sFDR}^{\ell}\leq q.$
\end{enumerate}
\end{thm}
This result follows from a more general theorem given in the Supplementary Material, which allows for different selection rules within the families at levels $1,\ldots, L-1.$

For general dependence among the $p$-values in the tree, we provide a more conservative variant of the TreeBH procedure, which reduces to the Benjamini-Yekutieli procedure \cite{BY01} when one has a single family of hypotheses, i.e.\ when $L=1.$
 As a first step, we should have valid $p$-values for the hypotheses in the tree. Considering the case where each parent hypothesis is intersection of its children, and valid $p$-values for Level-$L$ hypotheses are available, we can obtain the $p$-values for the intersection hypotheses at coarser resolution levels by combining the $p$-values within the families they index using
Bonferroni method, or using Simes method with the adjustment for general dependence (\cite{BY01}). Specifically, the $p$-value for $H_{i}$ indexing family $\mathcal{F}_i^\ell$  can be obtained as follows:
\begin{align}
&p_B=\min_{j} p_{(j)}\times k\label{bonf}\\
&p_S=\min_{j} p_{(j)}\times \frac{k}{j}\times g(k)\label{simesgen}
\end{align}
where
\begin{align*}
g(m)=\sum_{i=1}^{m}\frac{1}{i}
\end{align*}
for each natural number $m,$ the set $\big\{p_j: j = 1, \ldots, k\big\}$ corresponds to the $p$-values for the hypotheses belonging to $\mathcal{F}_i^\ell$, and
 the index in parentheses signifies that the $p$-values have been sorted in increasing order. The $p$-value in (\ref{bonf}) is based on Bonferroni method, while the $p$-value in (\ref{simesgen}) is based on Simes' method with the adjustment for arbitrary dependence.
It follows from Theorem 1.3 in \cite{BY01} that the above adjustment of Simes' $p$-value for the intersection hypothesis guarantees its validity under arbitrary dependence among the $p$-values. As discussed in \cite{BY01}, $g(m)\approx \log(m)+1/2,$ and using the above adjustment of Simes' $p$-value may be more powerful than using Bonferroni if the proportion of false null hypotheses within a family is not lower than $\log(k)/k.$ Other methods for obtaining valid $p$-values for the intersection hypothesis under arbitrary dependence are available, e.g. Sidak's method for combining absolute values of test statistics which have multivariate normal distribution with zero mean under the null \cite{Royen}.  

 \paragraph{The TreeBH procedure for arbitrary dependence}
 {\em
\begin{description}
\item[Step 1.] Apply the BH procedure with the bound $\tilde{q}_0=q^{(1)}/g\left(|\mathcal{F}^1|\right)$ to the family $\mathcal{F}^1$ $p$-values. If $\mathcal{S}^{1}_{0}=\phi,$ i.e.\ no hypothesis is rejected in $\mathcal{F}^1$, then stop. Otherwise, proceed to Step 2.
\item[Step \boldmath$\ell>1.$] For each family $\mathcal{F}_i^{\ell}$ for which all the ancestors are rejected, i.e.\ $i\in \mathcal{S}_{P^1(i)}^{\ell-1},$ $P^k(i)\in \mathcal{S}_{P^{k+1}(i)}^{\ell-k-1}$ for $k=1,\ldots, \ell-2,$ apply the BH procedure to the family $\mathcal{F}_i^{\ell}$ $p$-values with the bound
    $$\tilde{q}_i= \Bigg\{\prod_{k=1}^{\ell-1} \frac{|\mathcal{S}_{P^k(i)}^{\ell-k}|}{|\mathcal{F}_{P^k(i)}^{\ell-k}|}\Bigg\}\frac{q^{(\ell)}}{g\left(\prod_{k=0}^{\ell-1} |\mathcal{F}_{P^k(i)}^{\ell-k}|\right)},$$
    where $P^0(i)=i.$ 
   \\If  no rejections are made or $\ell=L$, stop. Otherwise, proceed to step $\ell+1$.
\end{description}}
\noindent
It was shown by \cite{blanchard2008two} that the Benjamini-Yekutieli procedure belongs to a more general class of procedures which guarantee FDR control under arbitrary dependence. We leave a similar generalization of the above procedure for future work.

The variant above is more conservative than the TreeBH procedure, because for each tested family $\mathcal{F}_i^{\ell}$,
the BH procedure is applied to the family $\mathcal{F}_i^{\ell}$ $p$-values with the bound $\tilde{q}_i=q_i/g\left(\prod_{k=0}^{\ell-1} |\mathcal{F}_{P^k(i)}^{\ell-k}|\right),$ where
$q_i$ is the testing level for family $\mathcal{F}_i^{\ell}$ if it is tested when the TreeBH procedure (without the adjustment) is applied.
If at each level $\ell$ of the tree, the families $\mathcal{F}_i^{\ell}$ are of the same size, i.e.\ $|\mathcal{F}_i^{\ell}|=m_{\ell},$ then the adjustment is the same for each family $\mathcal{F}_i^{\ell}$ at level $\ell$:  $\tilde{q}_i=q_i/g\left(\prod_{k=0}^{\ell-1} m_{\ell-k}\right).$
\begin{thm}\label{thm-gen-dep}
Assume that we have valid $p$-values for the hypotheses in the tree, which can be arbitrarily dependent. The TreeBH procedure for arbitrary dependence with input parameters $(q^{(1)},\ldots,
q^{(L)})$ guarantees for each $\ell\in \{1,\ldots, L\}$ that
$\text{sFDR}^{\ell}\leq q^{(\ell)}.$
\end{thm}

Based on our simulation results and the fact that the TreeBH may be viewed as a generalization of the BH procedure for a tree of hypotheses, which is known to be robust to different types of dependency, we believe that the modification of the TreeBH procedure is not required for many types of dependencies encountered in applications.

\section{Examples and simulations} \label{sec:sim}
In this section we consider examples and simulations with the intention of (1) illustrating the differences between the error rates and procedures we introduced with
other methods in the literature; (2) studying the effect of a type of dependence between the test statistics at Level 3 similar to that we expect in genetic studies.

In the interest of simplicity, we limit ourselves to $L=3$, where all hypotheses within a given tree level have the same number of children.  In these three-level trees we can  use an index system for the hypotheses that explicitly indicates logical relations:
hypotheses in level 3 are $H_{ijt}$, those in level two  $H_{ij\bullet}=\cap_tH_{ijt}$, and those in level 1  $H_{i\bullet\bullet}=\cap_jH_{ij\bullet}$. We can then describe the configurations of true and false nulls using a matrix containing all the Level 3 hypotheses, as in Table \ref{matrix}.
Each row includes all the hypotheses $H_{ijt}$ corresponding to one value of $i$ (so that the presence of a non-null hypothesis in row $i$ signifies that $H_{i\bullet\bullet}$ is false).
Each column corresponds to one pair $(j,t)$, with all the columns with the same value of $j$ adjacent, so that, within a row, blocks of columns correspond to the families in Level 3, and the
presence of a non-null hypothesis in row $i$, block $j$ signifies that $H_{ij\bullet}$ is false.
\begin{table}
 \begin{center}
 \begin{tabular}{rc|c|c|c}
 \multicolumn{5}{l}

$H_{i1\bullet}\;\;\;\;\;\;\;\;\;\;\;\;\;\;\;\;\;\;\; H_{i2\bullet} \;\;\;\;\;\;\;\;\; \;\;\;\;\;\;\;\;\;\;\;\;\;\;\;\;\; \;\;\;\;\;\;\;\;\;\;\;\;\;\;\;\;\; H_{im\bullet}\;\;\;\;\;\;\;\;\; $\\
 $H_{1\bullet\bullet} \rightarrow$ & $\overbrace{H_{111}\cdots H_{11k_{1}}}$ & $\overbrace{H_{121} \cdots H_{12k_{2}}}$& $\;\;\;\;\;\;\;\; \cdots\;\;\;\; \;\;\;\; $ & $\overbrace{H_{1m1} \cdots H_{1mk_{m}}}$\\
$H_{2\bullet\bullet} \rightarrow$  & $H_{211}\cdots H_{21k_{1}}$ & $H_{221} \cdots H_{22k_{2}}$& $\;\;\;\;\;\;\;\; \cdots\;\;\;\; \;\;\;\;$&$H_{2m1} \cdots H_{2mk_{m}}$\\
& $\vdots$ & $\vdots$ & $\vdots$ & $\vdots$ \\
$H_{n\bullet\bullet} \rightarrow$& $H_{n11}\cdots H_{n1k_{1}}$ & $H_{n21} \cdots H_{n2k_{2}}$& $\;\;\;\;\;\;\;\; \cdots\;\;\;\; \;\;\;\;$&$H_{nm1} \cdots H_{nmk_{m}}$\\
 \end{tabular}
\end{center}
\caption{\em Organizing the families of hypotheses in a matrix form}\label{matrix}
\end{table}

To compare the performance of different approaches we rely on
level-specific error rates and power. Specifically, we calculate the
FDR for discoveries at Levels 1, 2, and 3, which we
denote by $\text{FDR}^{\ell}$ for $\ell=1,2,3$ respectively,
as well as the selective $\text{sFDR}^{\ell}$ for
levels $\ell=2,3$. Note that $\text{sFDR}^{1}$ is omitted since
FDR$^{1}=\text{sFDR}^{1}$. For the hypotheses at each level, we also
calculate power. Across the different simulations, we compare our
procedure with BH (as a standard of FDR control without reference to
sub-groups of hypotheses) and p-filter which, while having a
 different target, shares many commonalities with our
procedure (attempting to control the FDR for ``group level
discoveries'' and relying on Simes' combination $p$-value).
Specifically, we use the following approaches and labels.

\begin{description}
\item \textbf{BH} = Benjamini-Hochberg method \cite{BH95} applied across the pooled set of $p$-values for the entire matrix of hypotheses. This guarantees control of FDR$^{3}$.
\item \textbf{BB} = Benjamini-Bogomolov method \cite{BB14} applied with hypotheses grouped into a two-level hierarchy with $H_{i\bullet\bullet}$ in Level 1, each indexing a family ${\cal F}^{2}_{i}=\{H_{ijt}, j=1,\ldots m, t=1,\ldots, k_i\}$. The selection in Level 1 is done using BH on Simes' $p$-values for $H_{i\bullet\bullet}$. This guarantees control of FDR$^{1}$, as well as of a selective error rate on these second layer hypotheses (which we do not calculate here).
\item \textbf{p-filter -- non-hier} = p-filter applied to the matrix of hypotheses in Table \ref{matrix}, and groups defined by the pooled set of all hypotheses,   rows, and column. This guarantees control of FDR$^1$ and FDR$^3$ and is considered for ease of comparison with \cite{FBR15}.
\item \textbf{p-filter -- hier} =  p-filter applied with groups defined by the pooled set of all hypotheses,  rows, and  sets of columns (i.e.\ a nested setup mimicking our hierarchical procedures). This guarantees control of FDR$^{\ell}$ for $\ell=1,2,3$.
 \item \textbf{TreeBH} = the 3-level TreeBH,  guaranteeing control of  FDR$^1$, sFDR$^2$ and  sFDR$^{3}$.
 \end{description}

\subsection{Example: the differences across error rates and procedures}\label{E1}
We start with a small example that allows us to explicitly work out
the differences between the various error rates and procedures. We
consider six hypotheses at Level 1, each indexing a family of 6
hypotheses, parents to families at Level 3 that contain 2,2,2,2,2,
and 90 hypotheses respectively, with truth assignment as described
in Table \ref{TE1}.

\begin{table}[h!]
\centering
\resizebox{\textwidth}{!}{
\begin{tabular}{ll|ll|ll|ll|ll|llll}
\color{red}${H_{1,1,1}}$ & $H_{1,1,2}$ & \color{red}${H_{1,2,1}}$ & $H_{1,2,2}$ & \color{red}${H_{1,3,1}}$ & $H_{1,3,2}$ & \color{red}${H_{1,4,1}}$ & $H_{1,4,2}$ & \color{red}${H_{1,5,1}}$ & $H_{1,5,2}$ & $\color{red} {H_{1,6,1}}$ & \color{red} ${H_{1,6,2}}$ & \color{red}${\ldots}$ & \color{red}${H_{1,6,90}}$ \\
$H_{2,1,1}$ & $H_{2,1,2}$ & $H_{2,2,1}$ & $H_{2,2,2}$ & $H_{2,3,1}$ & $H_{2,3,2}$ & $H_{2,4,1}$ & $H_{2,4,2}$ & $H_{2,5,1}$ & $H_{2,5,2}$ & \color{red}${H_{2,6,1}}$ & $H_{2,6,2}$ & $\ldots$ & $H_{2,6,90}$\\
\color{red}${H_{3,1,1}}$ & \color{red}${H_{3,1,2}}$ & \color{red}${H_{3,2,1}}$ & \color{red}${H_{3,2,2}}$ & \color{red}${H_{3,3,1}}$ & \color{red}${H_{3,3,2}}$ & \color{red}${H_{3,4,1}}$ &\color{red} ${H_{3,4,2}}$ & \color{red}${H_{3,5,1}}$ &\color{red} ${H_{3,5,2}}$ & $ H_{3,6,1}$ &  $H_{3,6,2}$ & $\ldots$ & $H_{3,6,90}$ \\
\color{red}${H_{4,1,1}}$ & \color{red}${H_{4,1,2}}$ & \color{red}${H_{4,2,1}}$ & \color{red}${H_{4,2,2}}$ & \color{red}${H_{4,3,1}}$ & \color{red}${H_{4,3,2}}$ & \color{red}${H_{4,4,1}}$ &\color{red} ${H_{4,4,2}}$ & \color{red}${H_{4,5,1}}$ &\color{red} ${H_{4,5,2}}$ & $ H_{4,6,1}$ &  $H_{4,6,2}$ & $\ldots$ & $H_{4,6,90}$ \\

\color{red}${H_{5,1,1}}$ & \color{red}${H_{5,1,2}}$ & \color{red}${H_{5,2,1}}$ & \color{red}${H_{5,2,2}}$ & \color{red}${H_{5,3,1}}$ & \color{red}${H_{5,3,2}}$ & \color{red}${H_{5,4,1}}$ &\color{red} ${H_{5,4,2}}$ & \color{red}${H_{5,5,1}}$ &\color{red} ${H_{5,5,2}}$ & $ H_{5,6,1}$ &  $H_{5,6,2}$ & $\ldots$ & $H_{5,6,90}$ \\

$H_{6,1,1}$ & $H_{6,1,2}$ & $H_{6,2,1}$ & $H_{6,2,2}$ & $H_{6,3,1}$ & $H_{6,3,2}$ & $H_{6,4,1}$ & $H_{6,4,2}$ & $H_{6,5,1}$ & $H_{6,5,2}$ & $H_{6,6,1}$ & $H_{6,6,2}$ & $\ldots$ & $H_{6,6,90}$

\end{tabular}
}
\caption{\em The setting for Example \ref{E1}. The non-null hypotheses are marked in red.}
\label{TE1}
\end{table}

First we discuss the implications of the configuration in Table
\ref{TE1} on error rates. We note that 5 out of 6 Level 1 hypotheses
are false, so we can expect FDR$^1$ will be contained for any
method. Consider now the error control for Level 3 discoveries:
methods such as BH and p-filter, which do not have consideration for families, will weight any false discovery against the many possible
true discoveries in family ${\cal F}^3_{1,6}$ (the family which
contains 90 non-null hypotheses). For the selective methods we
propose, instead, any false discovery in families ${\cal F}^3_{i,
6}$ for $i=3,\ldots, 6$ would result in FDP$^3_{i, 6}=1$ (and
FDP$^3_{2,6}\geq 0.5$)   and this would contribute with a
substantial weight to the average in sFDR$^3$.

We now consider how the power of  the different procedures is influenced by the configuration in Table \ref{TE1}. A large number of the Level 3 families are homogeneous: this presents an advantage for testing procedures that recognize the families, thereby allowing the BH threshold for significance to adapt to the different proportions of non-null hypotheses.
Figure \ref{FE1} reports the results of a simulation where all the error rates are targeted using a bound of 0.1 and, for each realization, the $p$-values for each of the hypotheses are generated independently as follows:

\begin{align*}
X &\sim \mu + \mathcal{N}(0,1)\\
\text{$p$-value} &= 1 - \mathbf{\Phi}(X),
\end{align*} where $\mathbf{\Phi}$ denotes the standard normal cdf, $\mu$ = 0 for null hypotheses, and $\mu > 0$ for non-null hypotheses, where larger values of $\mu$ correspond to
greater signal strength.
\begin{figure}[h!]
\centering
\includegraphics[width = .88\linewidth, trim={3cm 0 0 0}, clip]{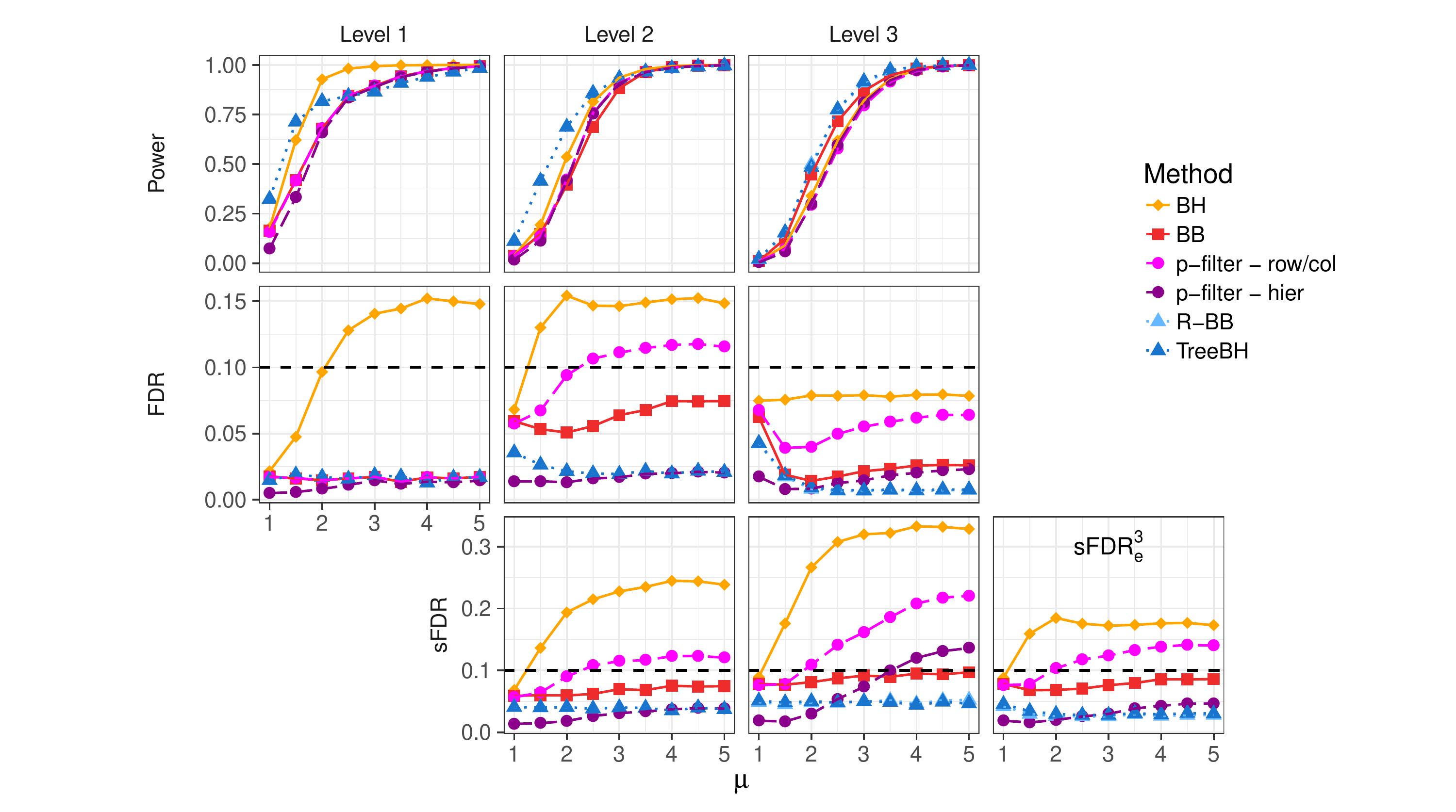}
\caption{\em Results for Example \ref{E1}. Each point corresponds to the average of 1000 realizations. Dashed horizontal lines indicate the target values for the error rates.}
\label{FE1}
\end{figure}
Figure \ref{FE1} underscores how each of the methods controls its target error rates, but not others. BH does not control any
Level 1 or Level 2 error rates, nor sFDR$^3$, and the p-filter methods do
not control all of the  sFDR$^{\ell}$. In this set-up it appears
that BB controls the sFDR$^{\ell}$, but we will see in other
examples that this is not always the case. Figure \ref{FE1} also
shows how the TreeBH 
procedure, despite showing the most
stringent error control in this example, has the highest power
across levels (beaten substantially only by BH in Level 1, where this procedure
has no FDR control). Interestingly, its power in Level 2 and 3 is
higher than that of BB, which shares some
of its hierarchical features and happens to control the selective
error rates in this case. The increased power at higher levels is
due to the fact that testing is carried out in more homogeneous
families. Higher power at Level 1 is due to the fact
that in calculating the Simes' $p$-values for $H_{i\bullet\bullet}$
one uses 6 $p$-values (rather than 100), and, at least for $i=3,4,5$,
five of those are going to be very small, due to the non-null status
of the hypotheses they represent: the effect of these $p$-values would
be more washed out in the entire pool of 100 hypotheses.

 \subsection{Example: the p-filter set-up}\label{E2}
For ease of comparison with the literature, we consider here the same example described in Section 6.2 of \cite{FBR15}. There are a total of 100$\times$100 hypotheses, and,
as illustrated in the upper left panel of Figure \ref{sim1_setup}, the non-null hypotheses (true signals) are arranged into 2 blocks of size $15 \times 15$ with $15$ additional
non-null hypotheses along the diagonal. In the simulations, $p$-values are generated using the same rule described in section \ref{E1}.
\begin{figure}[h!]
\centering
\includegraphics[width = 1.05\linewidth]{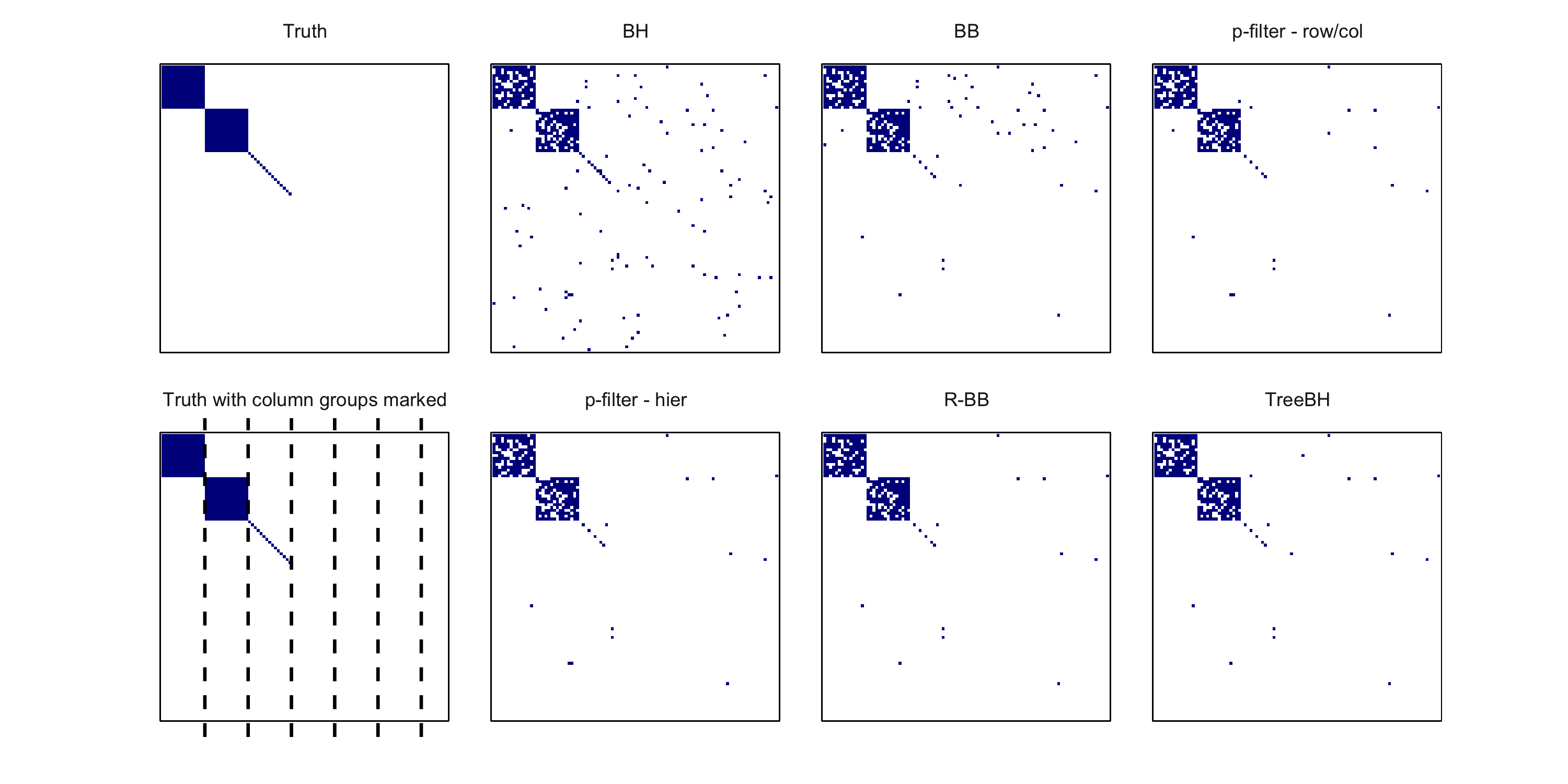}
\caption{\em True signals (marked in blue in panel labelled "truth") alongside results for an example run for all methods compared. The methods in the top row either do not group the hypotheses at all (BH), use a hierarchical grouping by row (BB), or group both by row and column (p-filter -- row/col). The methods in the bottom row (p-filter -- hier and TreeBH) use the column groupings demarcated with dashed lines in the lower left panel to group the hypotheses within each row.}
\label{sim1_setup}
\end{figure}

A single illustrative result using $\mu = 3$ is given in Figure \ref{sim1_setup}, and a full performance comparison across a range of $\mu$ values i.e.\ for varying signal strength is given in Figure \ref{sim1_results}. All methods are applied with target error rate 0.2.
 These results demonstrate that BH drastically fails to control the Level 1 and Level 2 FDR and the selective error rates, while BB and p-filter -- row/col,
 although they control the Level 1 FDR, do not control the Level 2 FDR or the Level 2 and Level 3 selective error rates. P-filter -- hier is able to control all error rates considered,
 but achieves lower Level 3 power vs.\ TreeBH. 
 The TreeBH 
 method, on the other hand, offers the best overall performance in terms of error control and power, even if many of the differences are small.
%
%

\begin{figure}[h!]
\centering
\includegraphics[width = 0.88\linewidth, trim={3cm 0 0 0}, clip]{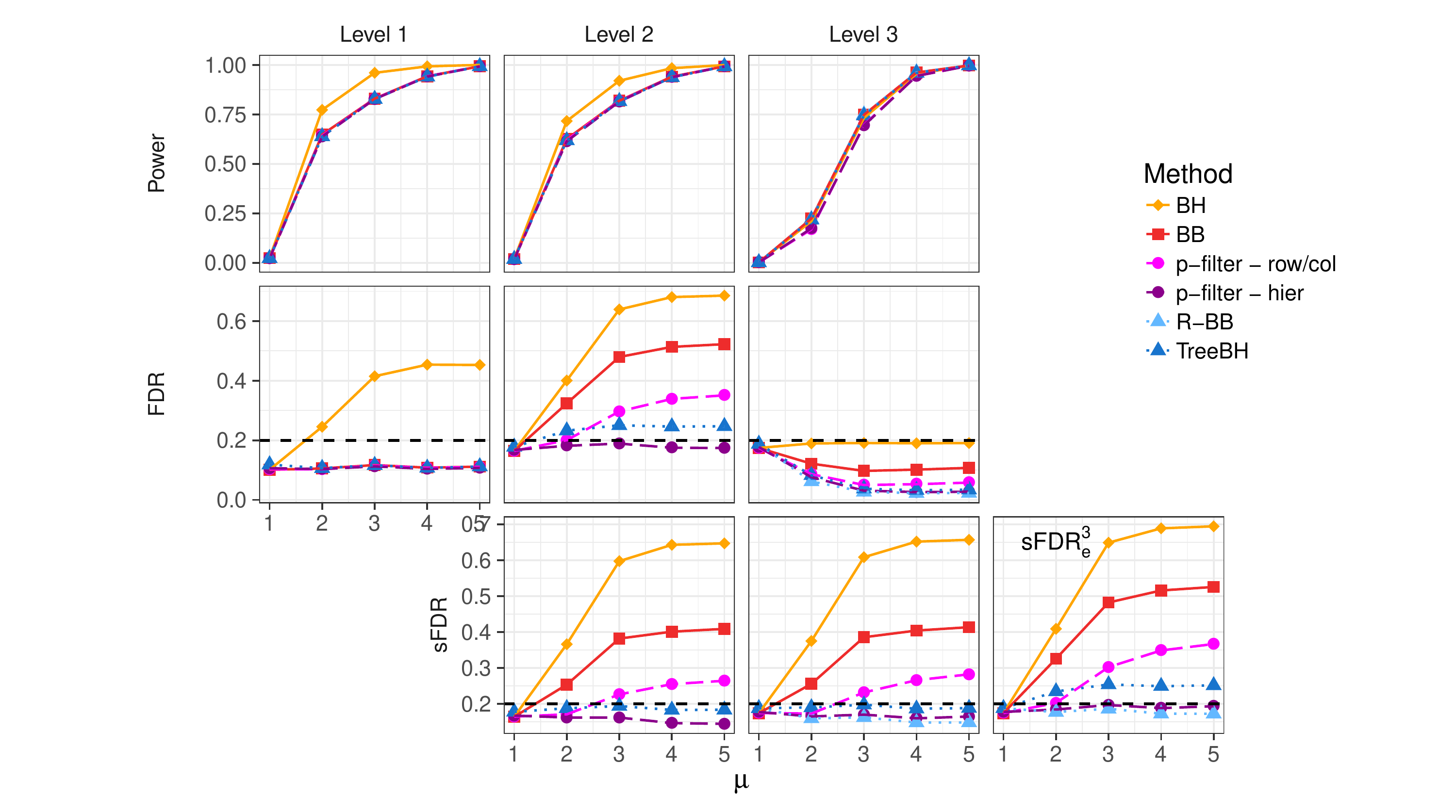}
\caption{\em Results for Example \ref{E2}, the "p-filter set-up".  Each point corresponds to the average of 100 realizations. Dashed horizontal lines indicate the target values for the error rates.}
\label{sim1_results}
\end{figure}

\subsection{Multi-trait GWAS}\label{sim_real_geno}
Having clarified the conceptual differences between our proposal and
existing methods, we now explore the properties of TreeBH 
in a set-up that resembles multi-trait genome-wide association
studies (GWAS): our goal is to investigate the effects of signal sparsity and of the dependence that exists among different tests.
 We adapt the simulation study proposed in Lewin et al. (2015) \cite{MTHESS} for eQTL studies, increasing the dimensionality to bring it closer to the GTEx set-up. While we defer a detailed description to the Supplementary Material, we  summarize some relevant features here. We use a matrix  $\mathbf{X}_{n\times p}$ of real genotyped data  from the North Finland Birth Cohort (NFBC) data set ($n$ = 5402 and $p$ = 8713) and simulate the expression of $q$ = 250 genes in each of $L$ = 5 tissues following a mixed effects model that incorporates across tissues correlation in errors. We assume that there are 50 causal SNPs which are each associated to the expression of 5, 10, or 20 genes.

 All $p$ = 8713 SNPs are tested for association to the simulated gene
expression traits. Association $p$-values are obtained using Matrix
eQTL \cite{MatrixEQTL}.
We compare the performance of BH, BB, 
and TreeBH,
where BB is applied with SNPs in Level 1 and all traits (i.e.\
expression for all genes across all tissues) in Level 2, and 
TreeBH is applied with SNPs in Level 1, genes in Level 2, and
tissues in Level 3. The p-filter is not included here as it was not
able to run in a timely manner (i.e.\ < 24 hrs) given similar
groupings to those in the  simulation with independent hypotheses.
All methods are applied with targeted error rate bound of 0.05. A performance
comparison for the 4 methods compared across 25 simulated data sets
is given in Figure S2: the TreeBH 
procedure appears to at least approximately control the target error
rates, unlike BH, and has a similar FDR, sFDR, and power to BB. In
particular, we see that BB is close to controlling the target error
rates, unlike in the simulations from Section \ref{E2}. This is
likely due to the fact that the adjustment for step 1 selection is
substantially more stringent here because of the sparsity of signal
across SNPs (specifically, the fact that only 50 of 8713 Level 1
hypotheses are non-null).

%

\section{Case studies} \label{sec:casestudies}
\subsection{Genetic regulation of gene expression across multiple tissues}

The goal of expression quantitative trait loci (eQTL) analysis is to identify DNA variants (typically SNPs) that influence the expression of genes. Since gene expression levels differ across tissues, eQTL analysis may reveal both shared and tissue-specific patterns of regulation. In addition to providing insight into the architecture of genetic regulation across tissues, the findings of multi-tissue eQTL analysis can help pinpoint disease mechanisms by linking risk variants from GWAS to the genes they regulate in specific tissues. The problem of identifying relevant associations is challenging, however, given the large number of hypotheses under consideration. 

In the eQTL setting, there are typically assumed to be two classes of regulation: local (in which variants affect the expression of nearby genes), and distal (in which variants affect the expression of genes located far away on the genome, perhaps even on different chromosomes). Local association
is much more common than distal regulation, and is believed to often be shared across multiple tissues. Given the very large number of hypotheses under consideration,
distal associations have proven to be difficult to detect. It is therefore common to focus on the more highly powered setting of local regulation, and we follow here this approach.

Typically, the $H_{ijt}$ hypotheses are tested separately for each
tissue using a simple linear model relating gene expression to
genetic variation. The collection of hypotheses $\{H_{ijt}\}$ can be
organized in different ways. The most common approach  has been to perform error control in each tissue
separately \cite{Nica2011, Grundberg2012}. Results across tissues
are then compared and conclusions are drawn on the tissue-specific
nature of the detected gene-SNP associations. It has been noted that
this approach is error prone and that joint analysis of multiple
tissues is likely to result in lower false positives and false
negatives. Methodology based on meta-analysis \cite{SetE13} and on
Bayesian model selection \cite{LetN13,FetS13} has been proposed to
address this shortcoming.  The testing procedure we propose here
provides some of the advantages of these methods while maintaining
the computational benefits of the simpler approach.

The genotype-tissue expression (GTEx) project is an NIH-funded project that aims to characterize variation in gene expression and genetic regulation across tissues. The Phase 1 data release includes 450 subjects and 44 tissues with at least 60 samples. In this release, gene expression was measured for around 21,000--34,000 genes per tissue, and genotypes  estimated for around 11 million SNPs. As in the simulation using real genotypes (Section \ref{sim_real_geno}), we focus on the set of reasonably independent SNPs filtered to have local $R^2$ < 0.5.  After tissue-specific QC, this set includes between 250,000 and 300,000 SNPs per tissue, with a total of 305,820 SNPs passing QC in at least one tissue.

For each tissue, the $H_{ijt}$ hypotheses are tested using a simple linear model with normalized expression for gene $j$
as the response and the estimated number of copies of the minor allele for SNP $i$ as the predictor. Gender, array platform, PEER factors (to adjust for confounding factors affecting global levels of gene expression), and genotype principal components (to adjust for population structure) are included as covariates. The corresponding $p$-values $p_{ijt}$ are generated using the software \texttt{FastQTL} \cite{FastQTL}.

One finding of interest in multi-tissue eQTL analysis is a set of eSNPs i.e.\ SNPs which we believe to play a functional regulatory role in at least one tissue.
Given this goal, we could naturally group the hypotheses into a hierarchical structure with SNPs in Level 1, genes in Level 2, and tissues in Level 3.
The Level 1 hypothesis ${H}_{i\bullet\bullet}$ addresses the question: does SNP $i$ have effects on expression in any tissue?
We consider SNP $i$ to be an eSNP if we reject $H_{i\bullet\bullet}$, and a SNP-gene pair to be discovered if we reject $H_{ij\bullet}$.
$P$-values are defined starting from the leaf hypotheses, which receive the $p$-value calculated by FastQTL. $P$-values for the Level 2 and Level 1 hypotheses are then defined using the Simes method.


Given this organization of the hypotheses, the TreeBH procedure controls
the following error rates: the FDR for eSNPs,  the expected average
proportion of false SNP-gene associations across the selected SNPs,
and a weighted average of proportion of false tissue discoveries for
the selected SNP-gene pairs.
We focus here on the tree described above as it
resembles the hierarchical structure that is
meaningful in multi-trait association studies and allows us to
annotate SNPs. Other ways of organizing
hypotheses in a tree are also meaningful; for
example, in our participation in the GTEx consortium we have used
other hierarchical structures when eGenes are the main discovery of interest.

To provide a benchmark for our results, in addition to TreeBH we
analyzed the data with two other approaches: the BH procedure
applied to the SNP-gene association $p$-values separately by tissue
(BH sep),  and the BH procedure applied to the pooled set of
$p$-values from all tissues (BH pooled). The total number of
discoveries for each of the adopted procedures is reported in Table
\ref{gtex_tab1}. The eSNP percentages show that a considerable
proportion of the SNPs tested for local association are selected
across all methods. The number selected for BH pooled is only
slightly less than the number selected for BH sep. This is due to
the fact that the BH threshold is adaptive, and in this setting, the
proportion of small $p$-values is quite substantial.
The TreeBH procedure is much more conservative at the SNP
level (because it provides control of the eSNP FDR), but relatively
less stringent in selecting the genes and tissues for these eSNPs,
resulting in a similar number of genes associated to each eSNP as
for the BH methods, an increased number of selected tissues for each
SNP-gene pair discovered, and a lower percentage of SNP-gene pairs
that were discovered in only 1 tissue. Given the hypothesis that
local regulatory relationships are likely to be shared across
tissues, the TreeBH results seem more biologically plausible vs.\
the results from the BH methods.

\begin{table}
\centering
\renewcommand{\arraystretch}{1.3}
\begin{tabular}{l| l| l|  l| l}
& & BH sep & BH pooled & TreeBH \\
\hline
\multirow{2}*{Level 1} & \# eSNPs & \num{9.1e4} & \num{8.6e4} & \num{4.5e4}
\\
\cline{2-5}
& \% eSNPs & 30\% & 28\% & 15\%\\
\hline
 \multirow{2}*{Level 2} & \# SNP-gene pairs & \num{1.9e5}  & \num{1.8e5}
& \num{9.3e4}  \\
\cline{2-5}
& \# genes per eSNP & 2.1 & 2.0 & 2.1\\
\hline
 \multirow{3}*{Level 3}  & \# SNP-gene-tissue triplets & \num{6.4e5} & \num{6.2e5}
& \num{5.1e6}\\
\cline{2-5}
& \# tissues per SNP-gene pair & 3.3 & 3.5 & 5.4\\
\cline{2-5}
& \% SNP-gene pairs 1 tissue only & 61\% & 61\% & 48\%
\end{tabular}
\caption[]{\em Numerical comparison of selection results for the BH
procedure applied separately by tissue (BH sep), BH procedure
applied to the pooled set of $p$-values from all tissues (BH
pooled), and from the 3-level TreeBH procedure (TreeBH).
"\%eSNPs" is the percentage of SNPs selected out of the 305,820 SNPs
tested for association. "\# genes per eSNP" and "\# tissues per
SNP-gene pair" are averages across all discovered eSNPs and all
discovered SNP-gene pairs, respectively. "\% SNP-gene pairs 1 tissue
only" is the percentage of discovered SNP-gene pairs that were
associated in only one tissue.} \label{gtex_tab1}
\end{table}


In Figure \ref{gtex_n_tissues}, we provide a more detailed look at the sharing of SNP-gene pairs across all 44 tissues, the 10 brain tissues, and the remaining 34 non-brain tissues.
We see that the proportion of SNP-gene pairs found in 1 tissue only had a larger decline when comparing the results for BH sep vs.\ TreeBH for the set of 10 brain tissues rather than all tissues. Specifically, the proportion of SNP-gene pairs discovered in 1 tissue only in all tissues was 0.61 for BH sep vs.\ 0.48 for TreeBH (21\% fewer), while in brain tissue the proportion was 0.58 for BH sep vs.\ 0.41 for TreeBH (29\% fewer). This suggests that the hierarchical method is in fact capitalizing on shared signal across closely related tissues.

The converse of this result is illustrated in Figure S3, which reports the number of SNP-gene pairs discovered in one tissue only (i.e.\ tissue-specific discoveries). While all tissues have fewer tissue-specific discoveries under  TreeBH vs.\ separate BH, some tissues (in particular, Testis) retain a large number of tissue-specific associations.

\begin{figure}[h!]
\centering
\includegraphics[width = \linewidth]{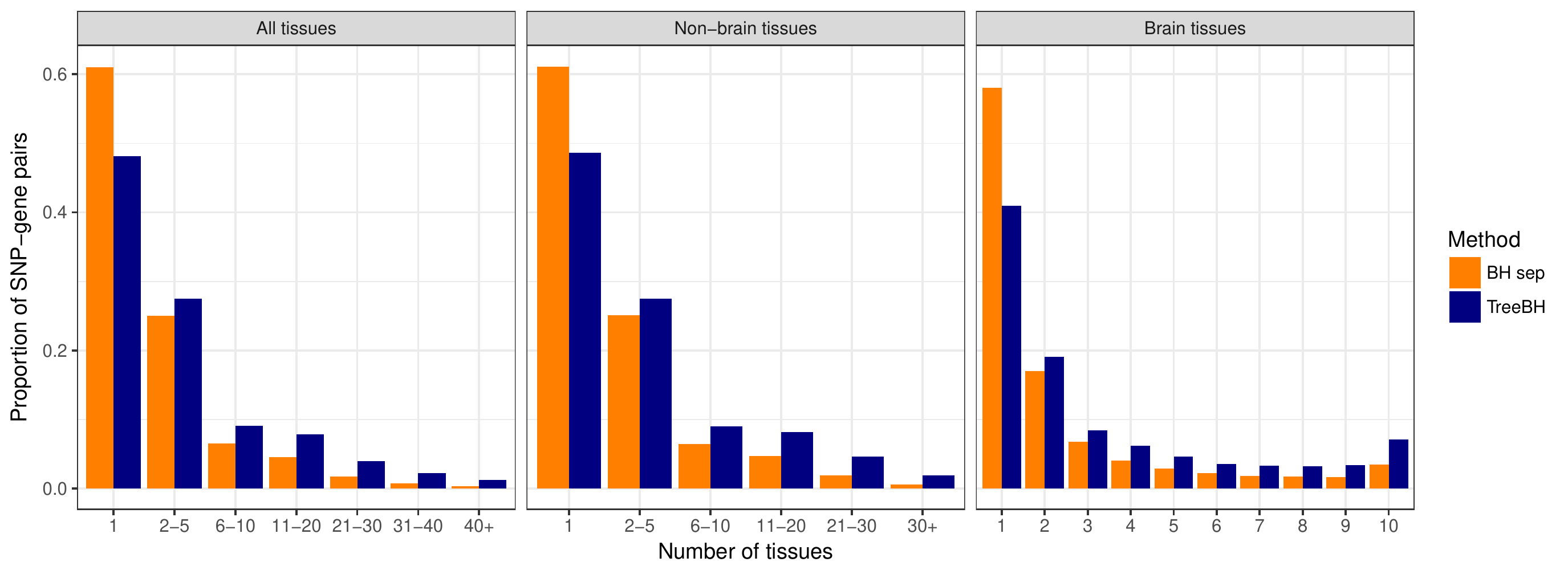}
\caption[]{\em Comparison of selection results for the BH procedure
applied separately by tissue ("BH sep") and the 3-level
TreeBH procedure in terms of sharing of SNP-gene pair
discoveries across tissues for all 44 tissues, the 34 non-brain
tissues, and the 10 brain tissues. } \label{gtex_n_tissues}
\end{figure}

%

\subsection{Association of the gut microbiome to colorectal cancer}

The data considered here were originally collected as part of study examining the association of gut microorganisms with colorectal cancer \cite{kostic}, and are provided in the R package \texttt{phyloseq} \cite{phyloseq}. For full details on the sample collection and processing, see \cite{kostic}.  Briefly, to quantify the relative abundance of microbial species, 16S rDNA (a variable region of the bacterial genome used for taxonomic classification) was amplified using PCR and then read via pyrosequencing. The resulting data are summarized as a table of counts for each operational taxonomic unit (OTU, a grouping of microorganisms similar to species but based on DNA sequence similarity) in each sample.

Here we consider the $n = 177$ samples with sufficient sample quantity (at least 500 total reads) and available diagnosis information (86 tumor, 91 normal). We focus on the $p = 496$ bacterial OTUs that are present in at least 10\% of samples. We compute $p$-values of difference between tumor and normal samples for each OTU using the negative binomial model for differential expression implemented in DESeq2 \cite{deseq2}. These OTUs are related within a taxonomic tree with the following levels: Kingdom - Phylum - Class - Order - Family - Genus - Species.
We are interested in controlling error rates at multiple levels in the taxonomic tree since the results of microbiome analysis are often discussed in terms of the taxonomic groupings discovered (e.g.\ phyla or genera) after controlling the FDR across OTUs. There is also an understanding that microbes which are closely related within a taxonomic tree are likely to have similar functions and roles in human disease. For a more complete discussion of the advantages of hierarchical testing in the analysis of microbiome data see \cite{SH14,SusanWorkflow}.

For the purpose of applying the TreeBH procedure, OTUs with missing
taxonomic information (specifically, a missing value for order,
family, genus, or species) are grouped together as the category
"Unknown" (nested within the appropriate hierarchy). The selections
obtained from applying the TreeBH procedure with 8 levels (Kingdom
- Phylum - Class - Order - Family - Genus - Species - OTU) using $q
= 0.05$ in each level of the tree are shown in Figure
\ref{fig:microbiome1}. Specifically, the TreeBH procedure results
in the selection of 6 phyla, 8 classes, 8 orders, 13 families, 15
genera, 17 species, and 19 OTUs. The number of selected OTUs is larger than the number of selected species as
 2 "Unknown" species are each represented by 2 OTUs.
Applying BH across the
$p$-values for all 496 OTUs with $q = 0.05$ results in the selection
of 33 OTUs (corresponding to 24 unique species), so the hierarchical procedure is more conservative in
terms of the raw number of OTUs (and species) discovered.
 The hierarchical
procedure does, however, include 4 additional species not discovered
using BH: in particular, 2 additional species of Fusobacteria are
discovered, which are key species of interest given the focus on the genus
{\textit{Fusobacterium}} in \cite{kostic}.
Another of the
additional species discovered ({\textit{Akkermansia muciniphila}})
has been previously associated with colorectal cancer specifically,
and was reported as one of the primary findings of \cite{Weir}. The
final additional species discovery under the hierarchical method
({\textit{Gemella haemolysans}}) has been previously associated to
oral cancer in a study of the oral microbiome \cite{Pushalkar}. In contrast, 10
of the 11 species discovered under BH which were not discovered
under TreeBH were categorized as "Unknown".
Thus, at a high level, the results of the TreeBH method, while more
conservative overall, were enhanced in terms of biologically
interpretable findings.

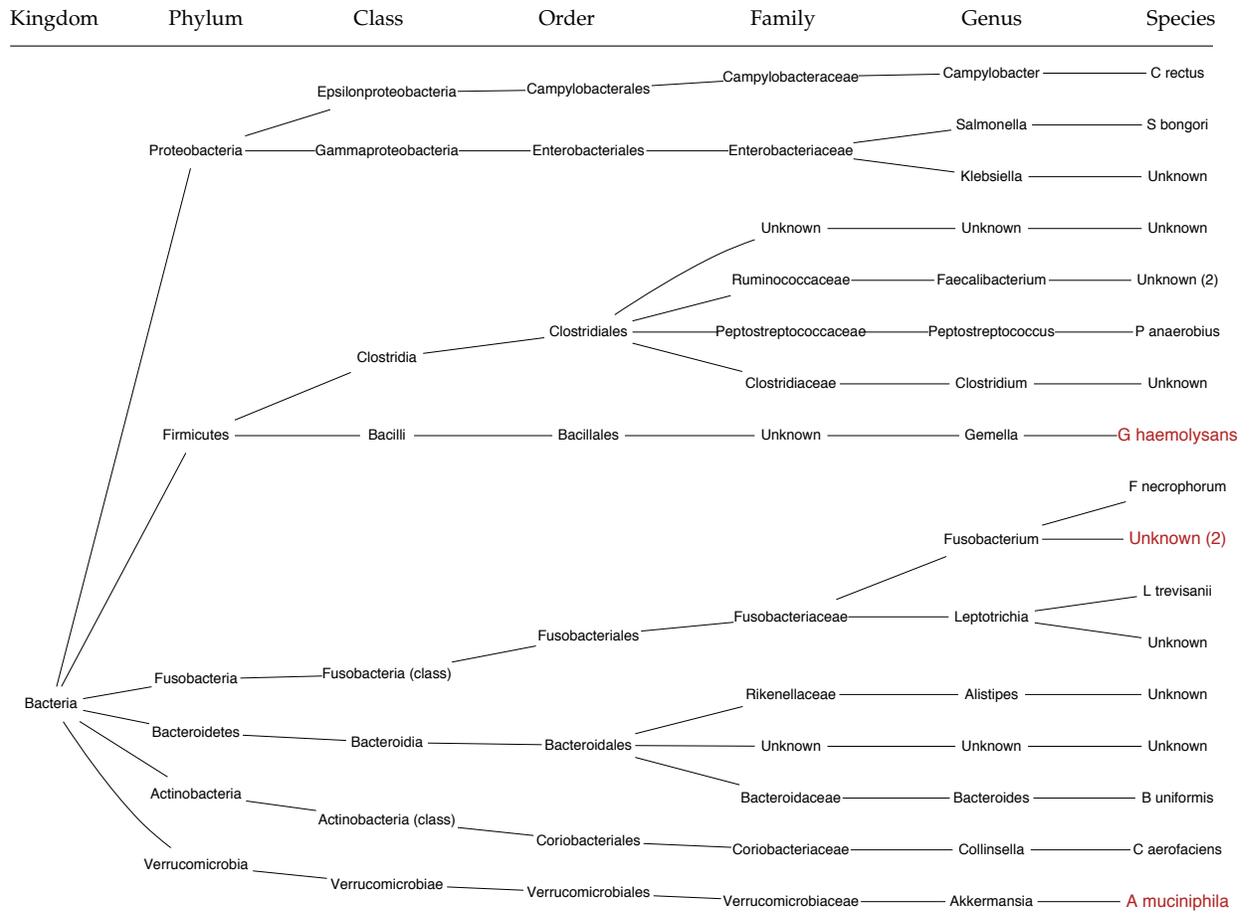
\begin{figure}
\vskip1cm
\centering
\begin{picture}(450,300)
\put(0,0){\includegraphics[width = \linewidth]{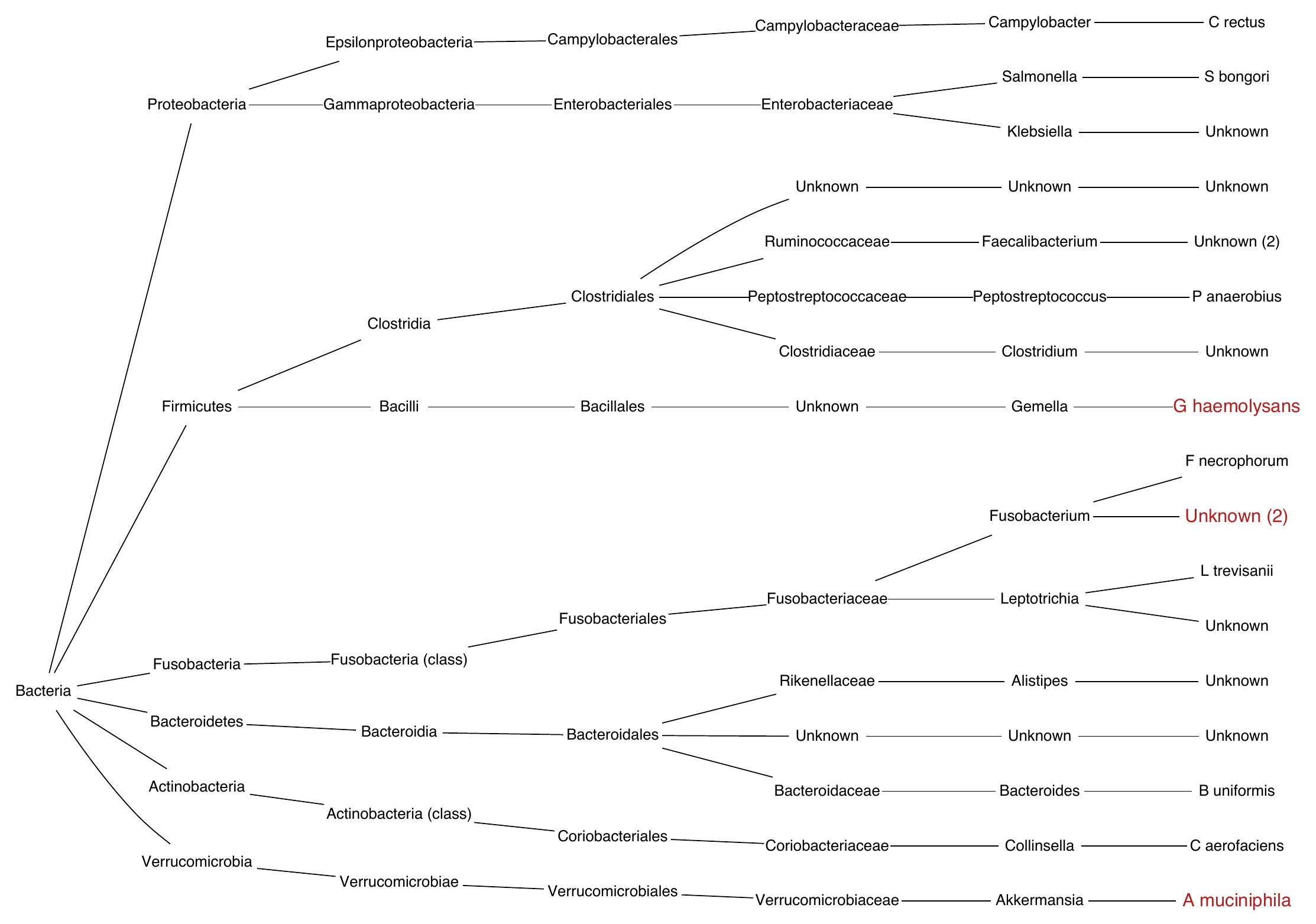}}
\put(0,332){\line(1,0){455}}
\put(0,340){\scriptsize Kingdom}
\put(60,340){\scriptsize Phylum}
\put(130,340){\scriptsize Class}
\put(200,340){\scriptsize Order}
\put(280,340){\scriptsize Family}
\put(360,340){\scriptsize Genus}
\put(430,340){\scriptsize Species}
\end{picture}
\caption{\em Taxonomic tree of selections using TreeBH procedure. Additional discoveries from TreeBH not found using BH are bolded and marked in red.}
\label{fig:microbiome1}
\end{figure}

\section*{Conclusions}
We have introduced a novel hierarchical procedure that allows the selective testing of families of hypotheses that hold more promise for discoveries, potentially increasing power without inflating error rates.
The TreeBH 
procedure
guarantees control of selective level-specific error rates under
certain assumptions on the dependence of
$p$-values within the tree. Our simulations show how
controlling these newly defined selective error rates often implies
control of other level-specific error measures. Furthermore, we have
shown that in scenarios designed to mimic the dependence and
sparsity structure of multi-trait GWAS, the procedure continues to
control (at least approximately) the target error rates.
We also showed  that our procedure can be
modified to control other selective error rates that are not based
on FDR. 
This makes the
hierarchical testing approach we present quite flexible and
adaptable to multiple settings, where the number of hypotheses in
families at some levels, for example, is small enough to make
controlling the FDR an unsatisfactory goal.

The TreeBH procedure  has many points of contact with the
p-filter \cite{FBR15}: both procedures control some form of level-specific FDR and, as implemented in our R package, in both cases one assumes that $p$-values for the finer scale hypotheses are available,
 and these are summarized with Simes method to obtain $p$-values for the group-level hypotheses.
When adapted to our
hierarchical organization of hypotheses, p-filter controls the
level-restricted FDR, but has no consideration for the distinct
families that make up the collection of hypotheses at a given
level.
Therefore, it does not control our selective
error rates, and cannot gain power by adapting to the different sparsity levels
across families.  Finally, the computational
time required by p-filter is substantially higher than that of
TreeBH, making its application in genomic problems difficult.
In a recent preprint \cite{RetJ17}, the p-filter framework has been extended to control FDR for any identified family of hypotheses: the set of families for which one desires this control has, however, to be specified in advance.
Technically it is possible within this framework to control FDR for each family, however this strategy does not guarantee control of our selective error rates because it does not adjust for the data-based selection of families. This is demonstrated in \cite{BB14} for a tree with two levels. Similarly to p-filter (\cite{RetJ17}), the TreeBH procedure can be extended to address weighted FDP (\cite{BH97}), prior knowledge weights on the hypotheses (\cite{GRW06}), adaptivity to the proportion of nulls by incorporating the estimator of \cite{SetS04}, and arbitrary dependence. However, the p-filter algorithm requires the $p$-values for all the hypotheses to be available at the same time, while our framework allows the data to be gathered sequentially based on the results of previous testing: the data for the child hypotheses can be collected only if their parent hypothesis is rejected. Thus our framework can be used in an evolving scientific investigation.

Our procedures are extensions of the proposals in \cite{BB14} and therefore have some of the same merits and limitations.
In particular, we want to point out that a recent work \cite{HetS16} underscores how selective error rates might be controlled with higher power when a conditional approach to testing is possible. They demonstrate the feasibility for two layers structure and under independence at the second level.
In principle, as long as exact conditional distributions can be evaluated, it might be possible to adopt the conditional testing approach in \cite{HetS16}  to also control the selective error rates that we introduced here. Currently it is not obvious to us how it can be done in a general hierarchical structure.

The study of eQTLs across multiple tissues presents such a formidable multiplicity challenge that it has motivated not only us, but also others to investigate novel procedures.
Among the most recent developments we would like to single out \cite{UWS16}, which takes an empirical Bayes approach, 
and whose analysis of GTEx
data has informed some of the displays we presented. Differently from the view-point we adopt, however, \cite{UWS16} does not offer control of the FDR for global discoveries.

\section*{Supplementary Material}
The Supplementary Material contains all proofs, the methodology and theoretical results for general classes of selective error rates and selection rules, and  details of GWAS simulation.
\section*{Acknowledgements}

The authors would like to thank the donors that made possible the
GTEx data collection and all the scientists that participated in the
consortium. Y.B. and C.S. acknowledge  support by NIH grant HG006695; C.P. and C.S. by NIH grant MH101782; C.P.  by NIH/NCI grant P30CA016672; Y.B. by ERC grant PSARPS 294519; and M.B.
by the Israel Science Foundation grant no. 1112/14. We acknowledge
use of the dataset ``STAMPEED: Northern Finland Birth Cohort 1966''
(phs000276.v2.p1).

\bibliography{Mybib}
\label{refs}
\bibliographystyle{ieeetrnew}

\renewcommand{\algorithmcfname}{TreeC}

\newpage
\beginsupplement
\pagenumbering{roman}
\setcounter{page}{1}
\setcounter{equation}{0}
\section*{Supplementary Material}

Note: for citations please refer to the main paper.

\subsection*{Proof of Proposition 1}
Assume that for a given Level $\ell,$ the procedure
guarantees that for each rejected hypothesis at Level $\ell-1$, at
least one hypothesis within the family it indexes is rejected. We
will show that the error measure assigned to each hypothesis at
Level $\ell-1$ by $\text{sFDR}^{\ell-1}$ is smaller or equal than  the
error measure assigned to the same hypothesis by $\text{sFDR}^{\ell},$ when we use the recursive definition
of $\text{sFDR}^{\ell}$s. Let
$H_i$ be a rejected hypothesis at Level $\ell-1,$ and let
$\overline{\mathcal{C}}_i^{\ell-1}$ and $\overline{\mathcal{C}}_i^{\ell}$ be
its error measure according to  $\text{sFDR}^{\ell-1}$ and $\text{sFDR}^{\ell}$ respectively.
For $\text{sFDR}^{\ell-1},$ $\overline{\mathcal{C}}_i=\textbf{I}_\text{\{$H_i$ is
null\}},$ while for $\text{sFDR}^{\ell},$ $\overline{\mathcal{C}}_i=\sum_{j\in
\mathcal{S}_i^\ell}\textbf{I}_\text{\{$H_j$ is
null\}}/|\mathcal{S}_i^\ell|.$ If $H_i$ is a true null hypothesis,
than its error measure according to $\text{sFDR}^{\ell-1}$ is 1. On the
other hand, in this case all the hypotheses in the family indexed by
$H_i$ are also true, and at least one hypothesis in this family is
rejected. This yields that the FDP in $\mathcal{F}_i^{\ell}$ is also 1, i.e.\ $\sum_{j\in
\mathcal{S}_i^\ell}\textbf{I}_\text{\{$H_j$ is
null\}}/|\mathcal{S}_i^\ell|=1$. Thus in the case where
$H_i$ is a true null hypothesis, the error measure assigned to this
hypothesis is 1 both for $\text{sFDR}^\ell$ and $\text{sFDR}^{\ell-1}.$ If $H_i$
is a false null hypothesis, the error measure assigned to it by
$\text{sFDR}^{\ell-1}$ is 0, while the error measure assigned to it by
$\text{sFDR}^{\ell-1}$ is the FDP in $\mathcal{F}_i^{\ell},$ which is 0 if no
false null hypothesis within this family is rejected, and is
positive otherwise. Thus we have proved that for each rejected hypothesis at Level $\ell-1,$ its error measure corresponding to $\text{sFDR}^{\ell-1}$
is smaller or equal to that
corresponding to $\text{sFDR}^{\ell}.$ Since the error measure for each
parent hypothesis is the average across the error measures of all
its child hypotheses, and the final error measure corresponds to the
error measure assigned to the root node of the tree, we obtain
immediately that $\text{sFDR}^{\ell-1}\leq \text{sFDR}^{\ell}.$

\subsection*{Proofs of Theorems 1, 3, and 4}
In all the proofs, we identify the root node hypothesis $H_0$ with $H_{i_0}$ and we identify each family $\mathcal{F}_i^{\ell}$ with the set of indices of hypotheses that belong to $\mathcal{F}_i^{\ell}.$
Before proving Theorems 1, 3, and 4, we will develop an expression for $\text{sFDR}^{\ell},$ $\ell\in\{1, \ldots, L\},$ for TreeBH procedure.

Consider TreeBH procedure with input parameters $(q^{(1)}, \ldots, q^{(L)}).$
For any $\ell\in\{1, \ldots, L\},$ any sequence of indices $i_1, i_2, \ldots, i_{\ell}$ such that
$i_1\in \mathcal{F}^1,$ $i_2\in
\mathcal{F}_{i_1}^2,\ldots,i_{\ell}\in
\mathcal{F}_{i_{\ell-1}}^{\ell},$ and any sequence $(r_1, \ldots, r_{\ell})$ such that $r_k\in\{1, \ldots, |\mathcal{F}_{i_{k-1}}^k|\}$ for each $k\in\{1, \ldots, \ell\},$ let us define  
 the following event: 
\begin{align}C_{r_{1},\ldots, r_{\ell}}^{(i_1,\ldots, i_{\ell})}=\bigcap_{k=1}^{\ell}C_{r_k}^{(i_k)}[q_{i_{k-1}}],\label{Cr}\end{align}
where $C_{r_k}^{(i_k)}[\alpha]$ is defined below, $q_{i_0}=q^{(1)},$ and
\begin{align}q_{i_{k-1}}=\frac{\prod_{j=1}^{k-1}r_{j}}{\prod_{j=1}^{k-1}|\mathcal{F}_{i_{j-1}}^{j}|}q^{(k)}\label{q_i}\end{align} for $k>1,$ i.e.\ $q_{i_{k-1}}$ is the testing level of the BH procedure applied on the $p$-values of family $\mathcal{F}_{i_{k-1}}^{k}$ when the numbers of rejections in  $\mathcal{F}^1,\ldots, \mathcal{F}^{k-1}_{i_{k-2}}$ are
$r_{1}, \ldots, r_{k-1},$ respectively, according to the TreeBH procedure. The event $C_{r_k}^{(i_k)}[\alpha]$ for $0<\alpha<1$ is defined 
as follows:
$$\{\bm p[\mathcal{F}_{i_{k-1}}^{k}\setminus\{i_{k}\}]: p^{(i_k)}_{(r_k-1)}\leq r_k\alpha/|\mathcal{F}^k_{i_{k-1}}|, p^{(i_k)}_{(r_k)}> (r_k+1)\alpha/|\mathcal{F}^k_{i_{k-1}}|,\ldots, p^{(i_k)}_{(|\mathcal{F}^k_{i_{k-1}}|-1 )}>\alpha \},$$
where $\bm p[\mathcal{F}_{i_{k-1}}^{k}\setminus\{i_{k}\}]$ is the vector of $p$-values corresponding to the hypotheses in $\mathcal{F}_{i_{k-1}}^{k}\setminus\{i_{k}\},$ and $p^{(i_k)}_{(1)}\leq p^{(i_k)}_{(2)}\leq \ldots \leq p^{(i_k)}_{(|\mathcal{F}_{i_{k-1}}|-1)}$ is the ordered sequence of these $p$-values. In other words, $C_{r_k}^{(i_k)}[\alpha]$ is the event in which if $H_{i_k}$ is rejected using BH procedure at level $\alpha$ applied on the $p$-values of family $\mathcal{F}_{i_{k-1}}^{k},$ then $r_k-1$ other hypotheses in this family are rejected alongside with it. The event $C_{r_k}^{(i_k)}[\alpha]$ was defined in \cite{BY01} for a single family of hypotheses. If only the $p$-values for Level-$L$ $p$-values are available, and the $p$-values for the hypotheses at the coarser resolution levels are combinations of Level-$L$ $p$-values, then the vector $\bm p[\mathcal{F}_{i_{k-1}}^{k}\setminus\{i_{k}\}]$ is a function of $p$-values for the hypotheses at level $L,$ so the events $C_{r_k}^{(i_k)}[\alpha]$ and $C_{r_1, \ldots, r_{\ell}}^{(i_1, \ldots, i_{\ell})}$ are defined on the space of Level-$L$ $p$-values.
Note that $C_{r_{1},\ldots, r_{\ell}}^{(i_1,\ldots, i_{\ell})}$ is the event in which if $H_{i_1},$ $H_{i_2},\ldots, H_{i_{\ell}}$ are rejected by the TreeBH procedure with input parameters $q^{(1)}, \ldots, q^{(\ell)}$, then the numbers of rejected hypotheses by TreeBH in families $\mathcal{F}^1, \mathcal{F}_{i_1}^2,\ldots, \mathcal{F}_{i_{\ell-1}}^{\ell}$ are $r_{1}, r_{2},\ldots, r_{\ell},$ respectively.

Let us now define \begin{align}B_s^{(i_1, \ldots, i_{\ell})}=\bigcup_{(r_1, \ldots, r_{\ell})\in I_s^{(i_1, \ldots, i_{\ell})}}C_{r_1, r_2,\ldots, r_{\ell}}^{(i_1, \ldots, i_{\ell})}\label{Bevent}\end{align} for $$I_s^{(i_1, \ldots, i_{\ell})}=\{(r_1, \ldots, r_{\ell}): r_k\in\{1,\ldots,|\mathcal{F}_{i_{k-1}}^k|\} \,\,\text{for}\,\,k=1,\ldots, \ell,\,\, \prod_{k=1}^{\ell}r_k=s\}.$$
\begin{lem}\label{lemtreeBH}
For TreeBH procedure with input parameters $(q^{(1)},\ldots, q^{(L)}),$ for each $\ell\in\{1, \ldots, L\}:$
\begin{enumerate}
\item For any sequence of indices $i_1, i_2, \ldots, i_{\ell}$ such that
$i_1\in \mathcal{F}^1,$ $i_2\in
\mathcal{F}_{i_1}^2,\ldots,i_{\ell}\in
\mathcal{F}_{i_{\ell-1}}^{\ell},$ the events $C_{r_1, \ldots, r_{\ell}}^{(i_1, \ldots, i_{\ell})}$ and $C_{r'_1, \ldots, r'_{\ell}}^{(i_1, \ldots, i_{\ell})}$ are disjoint for $(r_1, \ldots, r_{\ell})\neq(r'_1, \ldots, r'_{\ell}).$
\item For any sequence of indices $i_1, i_2, \ldots, i_{\ell}$ such that
$i_1\in \mathcal{F}^1,$ $i_2\in
\mathcal{F}_{i_1}^2,\ldots,i_{\ell}\in
\mathcal{F}_{i_{\ell-1}}^{\ell},$ the events $B_s^{(i_1, \ldots, i_{\ell})}$ and $B_{s'}^{(i_1, \ldots, i_{\ell})}$ are disjoint for $s\neq s',$ and the event $\cup_{s=1}^{\prod_{k=1}^{\ell}|\mathcal{F}_{i_{k-1}}^{k}|}B_s^{(i_1, \ldots, i_{\ell})}$ is the entire sample space.
\item The following inequality holds:
\begin{align*}\text{sFDR}^{\ell}\leq \sum_{i_1\in
\mathcal{F}^1}\sum_{i_2\in
\mathcal{F}_{i_1}^2}\cdots\!\!\!\!\!\!
\sum_{i_{\ell-1}\in
\mathcal{F}_{i_{\ell-2}}^{\ell-1}}\sum_{i_{\ell}\in
\mathcal{F}_{i_{\ell-1}}^{\ell}\cap \mathcal{H}_0}\sum_{s=1}^{\prod_{k=1}^{\ell}|\mathcal{F}_{i_{k-1}}^k|}\frac{1}{s}\mathbb{P}\left( P_{i_{\ell}}\leq \frac{sq^{(\ell)}}{\prod_{k=1}^{\ell}|\mathcal{F}_{i_{k-1}}^k|}, B_s^{(i_1, \ldots, i_{\ell})}\right).\end{align*}
\end{enumerate}
\end{lem}
\subsubsection*{Proof of Theorem 1}
The proof uses the techniques of the proof of Theorem 1.2 in \cite{BY01} and the proof of Theorem 3 in \cite{BB14}.
Let $\ell\in\{1, \ldots, L\}$ be an arbitrary fixed level. Let $\mathcal{H}_0$ be the set of indices of true null hypotheses. Using Lemma \ref{lemtreeBH} we have the following inequality for the TreeBH procedure:
\begin{align}
\text{sFDR}^{\ell}\leq\sum_{i_1\in
\mathcal{F}^1}\sum_{i_2\in
\mathcal{F}_{i_1}^2}\cdots\!\!\!\!\!\!
\sum_{i_{\ell-1}\in
\mathcal{F}_{i_{\ell-2}}^{\ell-1}}\sum_{i_{\ell}\in
\mathcal{F}_{i_{\ell-1}}^{\ell}\cap \mathcal{H}_0}\sum_{s=1}^{\prod_{k=1}^{\ell}|\mathcal{F}_{i_{k-1}}^k|}\frac{1}{s}\mathbb{P}\left( P_{i_{\ell}}\leq \frac{sq^{(\ell)}}{\prod_{k=1}^{\ell}|\mathcal{F}_{i_{k-1}}^k|}, B_s^{(i_1, \ldots, i_{\ell})}\right)\label{mainprds}
\end{align}
Note that for any sequence of indices $(i_1,\ldots, i_{\ell})$ such that $i_1\in \mathcal{F}_{1},$ $i_2\in \mathcal{F}_{i_1}^2,$ $\ldots,$ $i_{\ell-1}\in\mathcal{F}_{i_{\ell-2}}^{\ell-1},$ $i_{\ell}\in \mathcal{F}_{i_{\ell-1}}^{\ell}\cap \mathcal{H}_0,$  and $s\in\{1, \ldots, \prod_{k=1}^{\ell}|\mathcal{F}_{i_{k-1}}^k|\},$
it holds:
\begin{align}
\mathbb{P}\left( P_{i_{\ell}}\leq \frac{sq^{(\ell)}}{\prod_{k=1}^{\ell}|\mathcal{F}_{i_{k-1}}^k|}, B_s^{(i_1, \ldots, i_{\ell})}\right)=&\mathbb{P}\left(B_s^{(i_1, \ldots, i_{\ell})}\mid  P_{i_{\ell}}\leq \frac{sq^{(\ell)}}{\prod_{k=1}^{\ell}|\mathcal{F}_{i_{k-1}}^k|}\right)\mathbb{P}\left(P_{i_{\ell}}\leq
\frac{sq^{(\ell)}}{\prod_{k=1}^{\ell}|\mathcal{F}_{i_{k-1}}^k|} \right)\notag\\\leq&
\mathbb{P}\left(B_s^{(i_1, \ldots, i_{\ell})}\mid  P_{i_{\ell}}\leq \frac{sq^{(\ell)}}{\prod_{k=1}^{\ell}|\mathcal{F}_{i_{k-1}}^k|}\right)
\frac{sq^{(\ell)}}{\prod_{k=1}^{\ell}|\mathcal{F}_{i_{k-1}}^k|},\label{unif}
\end{align}
where the inequality in (\ref{unif}) follows from the fact that $P_{i_{\ell}}$ is a $p$-value of a true null hypothesis, therefore it satisfies $\mathbb{P}(P_{i_{\ell}}\leq x)\leq x$ for all $x\in[0,1].$
Combining (\ref{unif}) with (\ref{mainprds}) we obtain the following inequality:
\begin{align}
&\text{sFDR}^{\ell}\leq\notag \\&q^{(\ell)}\sum_{i_1\in
\mathcal{F}^1}\sum_{i_2\in
\mathcal{F}_{i_1}^2}\cdots\!\!\!\!\!\!
\sum_{i_{\ell-1}\in
\mathcal{F}_{i_{\ell-2}}^{\ell-1}}\sum_{i_{\ell}\in
\mathcal{F}_{i_{\ell-1}}^{\ell}\cap \mathcal{H}_0}\frac{1}{\prod_{k=1}^{\ell}|\mathcal{F}_{i_{k-1}}^k|}\sum_{s=1}^{\prod_{k=1}^{\ell}|\mathcal{F}_{i_{k-1}}^k|}\mathbb{P}\left(B_s^{(i_1, \ldots, i_{\ell})}\mid P_{i_{\ell}}\leq \frac{sq^{(\ell)}}{\prod_{k=1}^{\ell}|\mathcal{F}_{i_{k-1}}^k|}\right)\label{int-fdr}
\end{align}
It is enough to prove that for each sequence of indices $(i_1, \ldots, i_{\ell})$  such that $i_1\in \mathcal{F}_{1},$ $i_2\in \mathcal{F}_{i_1}^2,$ $\ldots,$ $i_{\ell-1}\in\mathcal{F}_{i_{\ell-2}}^{\ell-1},$ $i_{\ell}\in \mathcal{F}_{i_{\ell-1}}^{\ell}\cap \mathcal{H}_0,$  and $s\in\{1, \ldots, \prod_{k=1}^{\ell}|\mathcal{F}_{i_{k-1}}^k|\},$
\begin{align}
\sum_{s=1}^{\prod_{k=1}^{\ell}|\mathcal{F}_{i_{k-1}}^k|}\mathbb{P}\left(B_s^{(i_1, \ldots, i_{\ell})}\mid P_{i_{\ell}}\leq \frac{sq^{(\ell)}}{\prod_{k=1}^{\ell}|\mathcal{F}_{i_{k-1}}^k|}\right)\leq 1,\label{enough}
\end{align}
because combining (\ref{enough}) with (\ref{int-fdr}) we obtain
\begin{align*}
&\text{sFDR}^{\ell}\leq q^{(\ell)}\sum_{i_1\in
\mathcal{F}^1}\sum_{i_2\in
\mathcal{F}_{i_1}^2}\cdots\!\!\!\!\!\!
\sum_{i_{\ell-1}\in
\mathcal{F}_{i_{\ell-2}}^{\ell-1}}\sum_{i_{\ell}\in
\mathcal{F}_{i_{\ell-1}}^{\ell}\cap \mathcal{H}_0}\frac{1}{\prod_{k=1}^{\ell}|\mathcal{F}_{i_{k-1}}^k|}\leq
\\&q^{(\ell)}\sum_{i_1\in
\mathcal{F}^1}\sum_{i_2\in
\mathcal{F}_{i_1}^2}\cdots\!\!\!\!\!\!
\sum_{i_{\ell-1}\in
\mathcal{F}_{i_{\ell-2}}^{\ell-1}}\sum_{i_{\ell}\in
\mathcal{F}_{i_{\ell-1}}^{\ell}}\frac{1}{\prod_{k=1}^{\ell}|\mathcal{F}_{i_{k-1}}^k|}=q^{(\ell)}.
\end{align*}
For each sequence of indices $(i_1, \ldots, i_{\ell})$ such that $i_1\in \mathcal{F}^1,$ $i_2\in \mathcal{F}_{i_1}^2,$ $\ldots,$ $i_{\ell}\in \mathcal{F}_{i_{\ell-1}}^{\ell},$  and $t\in\{1, \ldots, \prod_{k=1}^{\ell}|\mathcal{F}_{i_{k-1}}^k|\},$
we define the following event:
\begin{align}
A_t^{(i_1, \ldots, i_{\ell})}=\bigcup_{s=1}^{t}B_s^{(i_1, \ldots, i_{\ell})}.\label{at}
\end{align}
In order to prove inequality (\ref{enough}) we will prove the following lemma.
\begin{lem}\label{incr}
For each sequence of indices $(i_1, \ldots, i_{\ell})$ such that $i_1\in \mathcal{F}^1,$ $i_2\in \mathcal{F}_{i_1}^2,$ $\ldots,$ $i_{\ell}\in \mathcal{F}_{i_{\ell-1}}^{\ell},$  and $t\in\{1, \ldots, \prod_{k=1}^{\ell}|\mathcal{F}_{i_{k-1}}^k|\},$  $A_t^{(i_1, \ldots, i_{\ell})}$
is an increasing set.
\end{lem}
Let  $(i_1,\ldots, i_{\ell})$ be an arbitrary sequence of indices such that $i_1\in \mathcal{F}_{1},$ $i_2\in \mathcal{F}_{i_1}^2,$ $\ldots,$ $i_{\ell-1}\in\mathcal{F}_{i_{\ell-2}}^{\ell-1},$ $i_{\ell}\in \mathcal{F}_{i_{\ell-1}}^{\ell}\cap \mathcal{H}_0.$ Using the result of Lemma \ref{incr} and the assumption that
the $p$-values corresponding to the hypotheses in the set $\mathcal{F}_{i_{\ell-1}}^{\ell}\cup\left[\cup_{k=1}^{\ell-1}\left(\mathcal{F}_{i_{k-1}}^{k}\setminus\{i_k\}\right)\right]$
 satisfy the PRDS property on the subset of $p$-values corresponding to hypotheses with indices in $\mathcal{F}_{i_{\ell-1}}^{\ell}\cap \mathcal{H}_0,$ we obtain for $s=1,\ldots, \prod_{k=1}^{\ell}|\mathcal{F}_{i_{k-1}}^k|-1,$
\begin{align}\label{prds}
\mathbb{P}\left(A_s^{(i_1, \ldots, i_{\ell})}|P_{i_{\ell}}\leq \frac{sq^{(\ell)}}{\prod_{k=1}^{\ell}|\mathcal{F}_{i_{k-1}}^k|}\right)\leq \mathbb{P}\left(A_s^{(i_1, \ldots, i_{\ell})}|P_{i_{\ell}}\leq \frac{(s+1)q^{(\ell)}}{\prod_{k=1}^{\ell}|\mathcal{F}_{i_{k-1}}^k|}\right).
\end{align}
The events $B_s^{(i_1, \ldots, i_{\ell})}$ and $B_{s'}^{(i_1, \ldots, i_{\ell})}$ are disjoint for $s\neq s',$ due to item 2 of Lemma \ref{lemtreeBH}. Therefore we obtain for any $t=1,\ldots, \prod_{k=1}^{\ell}|\mathcal{F}_{i_{k-1}}^k|$ and $\alpha\in[0,1]$
\begin{align}
\mathbb{P}\left(A_t^{(i_1, \ldots, i_{\ell})}\mid P_{i_{\ell}}\leq \alpha\right)=\sum_{s=1}^{t}\mathbb{P}\left(B_s^{(i_1, \ldots, i_{\ell})}\mid P_{i_{\ell}}\leq \alpha\right).\label{disjoint}
\end{align}
Using (\ref{prds}) and (\ref{disjoint}), we obtain for $s=1, \ldots, \prod_{k=1}^{\ell}|\mathcal{F}_{i_{k-1}}^k|-1,$
\begin{align}
&\mathbb{P}\left(A_s^{(i_1, \ldots, i_{\ell})}\mid P_{i_{\ell}}\leq \frac{sq^{(\ell)}}{\prod_{k=1}^{\ell}|\mathcal{F}_{i_{k-1}}^k|}\right)+\mathbb{P}\left(B_{s+1}^{(i_1, \ldots, i_{\ell})}\mid P_{i_{\ell}}\leq \frac{(s+1)q^{(\ell)}}{\prod_{k=1}^{\ell}|\mathcal{F}_{i_{k-1}}^k|}\right)\leq\label{mainineq}\\&\mathbb{P}\left(A_s^{(i_1, \ldots, i_{\ell})}\mid P_{i_{\ell}}\leq \frac{(s+1)q^{(\ell)}}{\prod_{k=1}^{\ell}|\mathcal{F}_{i_{k-1}}^k|}\right)+\mathbb{P}\left(B_{s+1}^{(i_1, \ldots, i_{\ell})}\mid P_{i_{\ell}}\leq \frac{(s+1)q^{(\ell)}}{\prod_{k=1}^{\ell}|\mathcal{F}_{i_{k-1}}^k|}\right)=\notag\\&\mathbb{P}\left(A_{s+1}^{(i_1, \ldots, i_{\ell})}\mid P_{i_{\ell}}\leq \frac{(s+1)q^{(\ell)}}{\prod_{k=1}^{\ell}|\mathcal{F}_{i_{k-1}}^k|}\right).\notag
\end{align}
Note that \begin{align}\mathbb{P}\left(A_1^{(i_1, \ldots, i_{\ell})}\mid P_{i_{\ell}}\leq \frac{q^{(\ell)}}{\prod_{k=1}^{\ell}|\mathcal{F}_{i_{k-1}}^k|}\right)=\mathbb{P}\left(B_1^{(i_1, \ldots, i_{\ell})}\mid P_{i_{\ell}}\leq \frac{q^{(\ell)}}{\prod_{k=1}^{\ell}|\mathcal{F}_{i_{k-1}}^k|}\right).\label{simple}\end{align}
Using (\ref{simple}) along with (\ref{disjoint}) and applying the inequality in (\ref{mainineq}) repeatedly, we obtain
\begin{align*}
&\sum_{s=1}^{\prod_{k=1}^{\ell}|\mathcal{F}_{i_{k-1}}^k|}\mathbb{P}\left(B_s^{(i_1, \ldots, i_{\ell})}\mid P_{i_{\ell}}\leq \frac{sq^{(\ell)}}{\prod_{k=1}^{\ell}|\mathcal{F}_{i_{k-1}}^k|}\right)\leq
\mathbb{P}\left(A_{\prod_{k=1}^{\ell}|\mathcal{F}_{i_{k-1}}^k|}^{(i_1, \ldots, i_{\ell})}\mid P_{i_{\ell}}\leq q^{(\ell)}\right)=1,
\end{align*}
where the last equality follows from the fact that $A_{\prod_{k=1}^{\ell}|\mathcal{F}_{i_{k-1}}^k|}^{(i_1, \ldots, i_{\ell})}$ is the entire sample space, due to item 2 of Lemma \ref{lemtreeBH}. This gives inequality (\ref{enough}), which completes the proof.
\begin{proof}[Proof of Lemma \ref{incr}]
Let $i_1\in \mathcal{F}^1,$ $i_2\in \mathcal{F}_{i_1}^2,$ $\ldots,$ $i_{\ell}\in \mathcal{F}_{i_{\ell-1}}^{\ell},$  and $s\in\{1, \ldots, \prod_{k=1}^{\ell}|\mathcal{F}_{i_{k-1}}^k|\}$ be arbitrary fixed. For simplicity, throughout this proof we denote the set $I_t^{(i_1, \ldots, i_{\ell})}$ by $I_t.$ Let us first show that
\begin{align*}
A_s^{(i_1, \ldots, i_{\ell})}=\bigcup_{t=1}^{s}\bigcup_{(r_1, \ldots, r_{\ell})\in I_t}\bigcap_{k=1}^{\ell}D_{r_k}^{(i_k)}[q_{i_{k-1}}],
\end{align*}
where $q_{i_{k-1}}$ is defined in (\ref{q_i}),
and
$$D_{r_k}^{(i_k)}[q_{i_{k-1}}]=\bigcup_{j=1}^{r_k}C_{j}^{(i_k)}[q_{i_{k-1}}],$$
i.e.\
$$D_{r_k}^{(i_k)}[q_{i_{k-1}}]=\left\{\bm p[\mathcal{F}_{i_{k-1}}^k\setminus\{i_k\}]: p_{(r_k)}^{(i_k)}>\frac{(r_k+1)q_{i_{k-1}}}{|\mathcal{F}_{i_{k-1}}^k|},  p_{(r_k+1)}^{(i_k)}>\frac{(r_k+2)q_{i_{k-1}}}{|\mathcal{F}_{i_{k-1}}^k|}, \ldots,  p_{(|\mathcal{F}_{i_{k-1}}^k|-1)}^{(i_k)}>q_{i_{k-1}}\right\}.$$
Using the definition of $B_s^{(i_1, \ldots, i_{\ell})}$ and $C_{r_1, \ldots, r_{\ell}}^{(i_1, \ldots, i_{\ell})},$ we obtain that
\begin{align*}
A_s^{(i_1, \ldots, i_{\ell})}=\bigcup_{t=1}^{s}\bigcup_{(r_1, \ldots, r_{\ell})\in I_t}\bigcap_{k=1}^{\ell} C_{r_k}^{(i_k)}[q_{i_{k-1}}],
\end{align*}
therefore we need to show that
\begin{align}\bigcup_{t=1}^{s}\bigcup_{(r_1, \ldots, r_{\ell})\in I_t}\bigcap_{k=1}^{\ell} C_{r_k}^{(i_k)}[q_{i_{k-1}}]=\bigcup_{t=1}^{s}\bigcup_{(r_1, \ldots, r_{\ell})\in I_t}\bigcap_{k=1}^{\ell}D_{r_k}^{(i_k)}[q_{i_{k-1}}].\label{equiv} \end{align}
Obviously, for each $t\leq s$ and sequence $(r_1, \ldots, r_{\ell})\in I_t,$ $C_{r_k}^{(i_k)}[q_{i_{k-1}}]\subseteq D_{r_k}^{(i_k)}[q_{i_{k-1}}]$ for $k=1, \ldots, \ell,$ therefore
$$ \bigcup_{t=1}^{s}\bigcup_{(r_1, \ldots, r_{\ell})\in I_t}\bigcap_{k=1}^{\ell} C_{r_k}^{(i_k)}[q_{i_{k-1}}]\subseteq\bigcup_{t=1}^{s}\bigcup_{(r_1, \ldots, r_{\ell})\in I_t}\bigcap_{k=1}^{\ell}D_{r_k}^{(i_k)}[q_{i_{k-1}}].$$
Let us now show that
 \begin{align}\bigcup_{t=1}^{s}\bigcup_{(r_1, \ldots, r_{\ell})\in I_t}\bigcap_{k=1}^{\ell} C_{r_k}^{(i_k)}[q_{i_{k-1}}]\supseteq\bigcup_{t=1}^{s}\bigcup_{(r_1, \ldots, r_{\ell})\in I_t}\bigcap_{k=1}^{\ell}D_{r_k}^{(i_k)}[q_{i_{k-1}}].\label{subset} \end{align}
Assume that the event $\bigcup_{t=1}^{s}\bigcup_{(r_1, \ldots, r_{\ell})\in I_t}\bigcap_{k=1}^{\ell}D_{r_k}^{(i_k)}[q_{i_{k-1}}]$ occurs. Then there exists $t_0\leq s$ and a sequence of integers $(r_1, \ldots, r_{\ell})\in I_{t_0}$ such that $\bigcap_{k=1}^{\ell}D_{r_k}^{(i_k)}[q_{i_{k-1}}]$ occurs. It follows that there exists a sequence of integers $(\tilde{r}_1, \ldots, \tilde{r}_{\ell})$ such that $1\leq \tilde{r}_k\leq r_k $ for each $k\leq \ell,$ and the event
$\bigcap_{k=1}^{\ell}C_{\tilde{r}_k}^{(i_k)}[q_{i_{k-1}}]$ occurs. Let $\tilde{t}=\prod_{k=1}^{\ell}\tilde{r}_k.$ Obviously, $\tilde{t}\leq \prod_{k=1}^{\ell}r_k=t_0\leq s,$ therefore the event $\bigcup_{t=1}^{s}\bigcup_{(r_1, \ldots, r_{\ell})\in I_t}\bigcap_{k=1}^{\ell} C_{r_k}^{(i_k)}[q_{i_{k-1}}]$ occurs. Thus we have proved (\ref{subset}), and therefore (\ref{equiv}) holds.

It is easy to see that for each $t\leq s$ and sequence $(r_1, \ldots, r_{\ell})\in I_t,$ $D_{r_k}^{(i_k)}[q_{i_{k-1}}]$ is an increasing set for each $k\in\{1, \ldots, \ell\}.$ Both intersection of increasing sets and union of  increasing sets are also increasing sets, therefore $A_s^{(i_1, \ldots, i_{\ell})}$ is an increasing set. This completes the proof.
\end{proof}
\subsubsection*{Proof of Theorem 3}
\paragraph{Proof of item 1}
 Let $\bm p^{L}$ denote the vector of $p$-values for the Level-$L$ hypotheses. In the setting of Theorem 3, the events $C_{r_1, \ldots, r_{L}}^{(i_1, \ldots, i_L)}$ and $B_s^{(i_1, \ldots, i_L)}$ are defined on the space of $\bm p^{L}.$ According to the proof of Theorem 1, it is enough to prove inequality (\ref{enough}) for $\ell=L.$ Let us first show that the set $A_s^{(i_1, \ldots, i_{L})}$ defined in (\ref{at}) is an increasing set of $\bm p^L$ for each sequence of indices $(i_1, \ldots, i_{L})$ such that $i_1\in \mathcal{F}^1,$ $i_2\in \mathcal{F}_{i_1}^2,$ $\ldots,$ $i_{L}\in \mathcal{F}_{i_{L-1}}^{L}$  and $s\in\{1, \ldots, \prod_{k=1}^{L}|\mathcal{F}_{i_{k-1}}^k|\}.$ Let us fix such $(i_1, \ldots, i_{L})$ and $s.$ It is shown in the proof of Lemma \ref{incr} that
\begin{align*}
A_s^{(i_1, \ldots, i_{L})}=\bigcup_{t=1}^{s}\bigcup_{(r_1, \ldots, r_{L})\in I_t}\bigcap_{k=1}^{L}D_{r_k}^{(i_k)}[q_{i_{k-1}}],
\end{align*}
where
$$D_{r_k}^{(i_k)}[q_{i_{k-1}}]=\left\{\bm p[\mathcal{F}_{i_{k-1}}^k\setminus\{i_k\}]: p_{(r_k)}^{(i_k)}>\frac{(r_k+1)q_{i_{k-1}}}{|\mathcal{F}_{i_{k-1}}^k|},  p_{(r_k+1)}^{(i_k)}>\frac{(r_k+2)q_{i_{k-1}}}{|\mathcal{F}_{i_{k-1}}^k|}, \ldots,  p_{(|\mathcal{F}_{i_{k-1}}^k|-1)}^{(i_k)}>q_{i_{k-1}}\right\}.$$
As shown in the proof of Lemma \ref{incr}, in order to show that $A_s^{(i_1, \ldots, i_{L})}$ is an increasing set, it is enough to show that $D_{r_k}^{(i_k)}[q_{i_{k-1}}]$ is an increasing set for each $t\leq s,$ $(r_1, \ldots, r_L)\in I_t,$ and $k=1, \ldots, L.$ Let us fix such $t,$ $(r_1, \ldots, r_L),$ and $k.$
Let $\bm p_1^L\leq \bm p_2^L$ be two vectors of Level-$L$ $p$-values, where the inequality is understood coordinate-wise. Let
$\bm p_1[\mathcal{F}_{i_{k-1}}^k\setminus\{i_k\}]$ and $\bm p_2[\mathcal{F}_{i_{k-1}}^k\setminus\{i_k\}] $ be the vectors of $p$-values for the hypotheses with indices in $\mathcal{F}_{i_{k-1}}^k\setminus\{i_k\},$
based on $\bm p_1^L$ and $\bm p_2^L$ respectively. Assume that the event $D_{r_k}^{(i_k)}[q_{i_{k-1}}]$
occurs for $\bm p_1[\mathcal{F}_{i_{k-1}}^k\setminus\{i_k\}].$
   Since each of the combination rules used for computing $p$-values at levels $1,\ldots,L-1$ is a
monotone non-decreasing function of each of its input $p$-values, it follows that
   $\bm p_1[\mathcal{F}_{i_{k-1}}^k\setminus\{i_k\}]\leq \bm p_2[\mathcal{F}_{i_{k-1}}^k\setminus\{i_k\}].$ Therefore, using the definition of $D_{r_k}^{(i_k)}[q_{i_{k-1}}],$ we obtain that $D_{r_k}^{(i_k)}[q_{i_{k-1}}]$ occurs for $\bm p_2[\mathcal{F}_{i_{k-1}}^k\setminus\{i_k\}]$ as well. Thus we have shown that
   $D_{r_k}^{(i_k)}[q_{i_{k-1}}]$ is an increasing set of $\bm p^L,$ which yields that  $A_s^{(i_1, \ldots, i_{L})}$ is an increasing set of $\bm p^L.$ Using this fact and the assumption that $\bm p^L$ satisfies the PRDS property, we obtain that
\begin{align*}
\mathbb{P}\left(A_s^{(i_1, \ldots, i_{L})}|P_{i_{L}}\leq \frac{sq^{(L)}}{\prod_{k=1}^{L}|\mathcal{F}_{i_{k-1}}^k|}\right)\leq \mathbb{P}\left(A_s^{(i_1, \ldots, i_{L})}|P_{i_{L}}\leq \frac{(s+1)q^{(L)}}{\prod_{k=1}^{L}|\mathcal{F}_{i_{k-1}}^k|}\right)
\end{align*}
for each sequence of indices $(i_1,\ldots, i_{L})$  such that $i_1\in \mathcal{F}_{1}, i_2\in \mathcal{F}_{i_1}^2,\ldots, i_{L-1}\in\mathcal{F}_{i_{L-2}}^{L-1},$ $i_{L}\in \mathcal{F}_{i_{L-1}}^{L}\cap \mathcal{H}_0,$ and $s\in\{1, \ldots, \prod_{k=1}^{L}|\mathcal{F}_{i_{k-1}}^k|-1\}.$ We have shown that inequality (\ref{prds}) holds for $\ell=L.$
The inequality (\ref{enough}) for $\ell=L$ follows using the derivations following inequality (\ref{prds}) in the proof of Theorem 1.
Thus we have proved that $\text{sFDR}^{L}\leq q^{(L)}.$

\paragraph{Proof of item 2} When $q^{(1)}= q^{(2)}= \ldots = q^{(L)}=q$ and Simes' method is used for obtaining $p$-values for the intersection hypotheses, the TreeBH procedure satisfies the consonance property:
for all the levels $\ell= 1,\ldots,L-1,$ it is guaranteed that each
rejected parent hypothesis has at least one rejected child hypothesis, as explained in Section 3 of the main manuscript. Therefore, we obtain from Proposition 1 that $\text{sFDR}^{\ell}\leq \text{sFDR}^{L}$ for $\ell=1, \ldots, L-1.$ Since Simes' combination rule is non-decreasing in each of its input $p$-values, it follows from item 1 that $\text{sFDR}^{L}\leq q^{(L)}=q,$ therefore we obtain $\text{sFDR}^{\ell}\leq q$ for $\ell=1, \ldots, L.$
\subsubsection*{Proof of Theorem 4}
Similarly to the proof of Lemma \ref{lemtreeBH}, it can be shown that for the modified variant of TreeBH for arbitrary dependence, for each $\ell\in\{1, \ldots, L\},$ the following inequality holds:
\begin{align}&\text{sFDR}^{\ell}\leq\notag\\&\sum_{i_1\in
\mathcal{F}^1}\sum_{i_2\in
\mathcal{F}_{i_1}^2}\cdots\!\!\!\!\!\!
\sum_{i_{\ell-1}\in
\mathcal{F}_{i_{\ell-2}}^{\ell-1}}\sum_{i_{\ell}\in
\mathcal{F}_{i_{\ell-1}}^{\ell}\cap \mathcal{H}_0}\sum_{s=1}^{\prod_{k=1}^{\ell}|\mathcal{F}_{i_{k-1}}^k|}\frac{1}{s}\mathbb{P}\left( P_{i_{\ell}}\leq \frac{sq^{(\ell)}}{\prod_{k=1}^{\ell}|\mathcal{F}_{i_{k-1}}^k|g\left(\prod_{k=1}^{\ell}|\mathcal{F}_{i_{k-1}}^k|\right)}, \tilde{B}_s^{(i_1, \ldots, i_{\ell})}\right),\label{mainmod}\end{align}
where $g(m)=\sum_{i=1}^m1/i$ for each natural number $m,$ as defined in the main manuscript, and for each sequence of indices $(i_1, \ldots, i_{\ell})$ such that $i_1\in \mathcal{F}^1,$ $i_2\in \mathcal{F}_{i_1}^2,$ $\ldots,$ $i_{\ell}\in \mathcal{F}_{i_{\ell-1}}^{\ell},$  and $s\in\{1, \ldots, \prod_{k=1}^{\ell}|\mathcal{F}_{i_{k-1}}^k|\},$
$$\tilde{B}_s^{(i_1, \ldots, i_{\ell})}=\bigcup_{(r_1, \ldots, r_{\ell})\in I_s^{(i_1, \ldots, i_{\ell})}}\tilde{C}_{r_1, r_2,\ldots, r_{\ell}}^{(i_1, \ldots, i_{\ell})}$$ for
\begin{align*}&I_s^{(i_1, \ldots, i_{\ell})}=\{(r_1, \ldots, r_{\ell}): r_k\in\{1,\ldots,|\mathcal{F}_{i_{k-1}}^k|\} \,\,\text{for}\,\,k=1,\ldots, \ell,\,\, \prod_{k=1}^{\ell}r_k=s\} \\
&\tilde{C}_{r_{1},\ldots, r_{\ell}}^{(i_1,\ldots, i_{\ell})}=\bigcap_{k=1}^{\ell}C_{r_k}^{(i_k)}[\tilde{q}_{i_{k-1}}],\end{align*}
where  $\tilde{q}_{i_0}=q^{(1)}/g(|\mathcal{F}^1|)$ and for $k>1,$ $$\tilde{q}_{i_{k-1}}=\frac{\prod_{j=1}^{k-1}r_{j}}{{\prod_{j=1}^{k-1}|\mathcal{F}_{i_{j-1}}^j|}g\left(\prod_{j=1}^{k}|\mathcal{F}_{i_{j-1}}^j|\right)}q^{(k)}$$ i.e.\ $\tilde{q}_{i_{k-1}}$ is the testing level of the BH procedure applied on the $p$-values of family $\mathcal{F}_{i_{k-1}}^{k}$ when the numbers of rejections in the families of ancestors of $\mathcal{F}_{i_{k-1}}^{k},$ i.e.\ $\mathcal{F}^1,\ldots, \mathcal{F}^{k-1}_{i_{k-2}},$ are $r_{1}, \ldots, r_{k-1},$ respectively, according to the modified TreeBH for arbitrary dependence with input parameters $(q^{(1)}, \ldots, q^{(L)}).$ Using similar arguments to those in the proof of Lemma \ref{lemtreeBH}, it can be shown that the events $\tilde{B}_s^{(i_1, \ldots, i_{\ell})},$ $s=1, \ldots, \prod_{k=1}^{\ell}|\mathcal{F}_{i_{k-1}}^k|,$ are disjoint and their union is the entire space.

Now the proof is similar to the proof of Theorem 3.1 in \cite{BY01}. Let $\ell\in\{1, \ldots, L\}$ be arbitrary fixed. For each sequence of indices $(i_1,\ldots, i_{\ell})$ such that $i_1\in \mathcal{F}_{1},$ $i_2\in \mathcal{F}_{i_1}^2,$ $\ldots,$ $i_{\ell-1}\in\mathcal{F}_{i_{\ell-2}}^{\ell-1}, i_{\ell}\in \mathcal{F}_{i_{\ell-1}}^{\ell}\cap \mathcal{H}_0,$  $s\in\{1, \ldots, \prod_{j=1}^{k}|\mathcal{F}_{i_{j-1}}^j|\},$ and
$j\in\{1, \ldots, s\},$ we define
$$p_{js}^{(i_1, \ldots, i_{\ell})}=\mathbb{P}\left(P_{i_{\ell}}\in \Bigg(\frac{(j-1)q^{(\ell)}}{\prod_{k=1}^{\ell}|\mathcal{F}_{i_{k-1}}^k|g\left(\prod_{k=1}^{\ell}|\mathcal{F}_{i_{k-1}}^k|\right)}, \frac{jq^{(\ell)}}{\prod_{k=1}^{\ell}|\mathcal{F}_{i_{k-1}}^k|g\left(\prod_{k=1}^{\ell}|\mathcal{F}_{i_{k-1}}^k|\right)}\Bigg], \tilde{B}_s^{(i_1,\ldots, i_{\ell})} \right).$$
Since $P_{i_{\ell}}$ is a $p$-value of a true null hypothesis, it satisfies $\mathbb{P}(P_{i_{\ell}}\leq x)\leq x$ for each $x\in[0,1],$ in particular, $\mathbb{P}(P_{i_{\ell}}\leq 0)=0.$ Therefore,
\begin{align}
\sum_{j=1}^{s}p_{js}^{(i_1, \ldots, i_{\ell})}=\mathbb{P}\left(P_{i_{\ell}}\leq \frac{sq^{(\ell)}}{\prod_{k=1}^{\ell}|\mathcal{F}_{i_{k-1}}^k|g\left(\prod_{k=1}^{\ell}|\mathcal{F}_{i_{k-1}}^k|\right)}, \tilde{B}_s^{(i_1,\ldots, i_{\ell})} \right).\label{sum1}
\end{align}
Using (\ref{mainmod}) and (\ref{sum1}) we obtain
\begin{align}&\text{sFDR}^{\ell}\leq \notag\\&\sum_{i_1\in
\mathcal{F}^1}\sum_{i_2\in
\mathcal{F}_{i_1}^2}\cdots\!\!\!\!\!\!
\sum_{i_{\ell-1}\in
\mathcal{F}_{i_{\ell-2}}^{\ell-1}}\sum_{i_{\ell}\in
\mathcal{F}_{i_{\ell-1}}^{\ell}\cap \mathcal{H}_0}\sum_{s=1}^{\prod_{k=1}^{\ell}|\mathcal{F}_{i_{k-1}}^k|}\frac{1}{s}\mathbb{P}\left( P_{i_{\ell}}\leq \frac{sq^{(\ell)}}{\prod_{k=1}^{\ell}|\mathcal{F}_{i_{k-1}}^k|g\left(\prod_{k=1}^{\ell}|\mathcal{F}_{i_{k-1}}^k|\right)}, \tilde{B}_s^{(i_1, \ldots, i_{\ell})}\right)=
\notag\\&\sum_{i_1\in
\mathcal{F}^1}\sum_{i_2\in
\mathcal{F}_{i_1}^2}\cdots\!\!\!\!\!\!
\sum_{i_{\ell-1}\in
\mathcal{F}_{i_{\ell-2}}^{\ell-1}}\sum_{i_{\ell}\in
\mathcal{F}_{i_{\ell-1}}^{\ell}\cap \mathcal{H}_0}\sum_{s=1}^{\prod_{k=1}^{\ell}|\mathcal{F}_{i_{k-1}}^k|}\frac{1}{s}\sum_{j=1}^{s}p_{js}^{(i_1, \ldots, i_{\ell})}=\notag\\&
\sum_{i_1\in
\mathcal{F}^1}\sum_{i_2\in
\mathcal{F}_{i_1}^2}\cdots\!\!\!\!\!\!
\sum_{i_{\ell-1}\in
\mathcal{F}_{i_{\ell-2}}^{\ell-1}}\sum_{i_{\ell}\in
\mathcal{F}_{i_{\ell-1}}^{\ell}\cap \mathcal{H}_0}\sum_{j=1}^{\prod_{k=1}^{\ell}|\mathcal{F}_{i_{k-1}}^k|}\sum_{s=j}^{\prod_{k=1}^{\ell}|\mathcal{F}_{i_{k-1}}^k|}\frac{1}{s}p_{js}^{(i_1, \ldots, i_{\ell})}\leq\notag\\&
\sum_{i_1\in
\mathcal{F}^1}\sum_{i_2\in
\mathcal{F}_{i_1}^2}\cdots\!\!\!\!\!\!
\sum_{i_{\ell-1}\in
\mathcal{F}_{i_{\ell-2}}^{\ell-1}}\sum_{i_{\ell}\in
\mathcal{F}_{i_{\ell-1}}^{\ell}\cap \mathcal{H}_0}\sum_{j=1}^{\prod_{k=1}^{\ell}|\mathcal{F}_{i_{k-1}}^k|}\sum_{s=j}^{\prod_{k=1}^{\ell}|\mathcal{F}_{i_{k-1}}^k|}\frac{1}{j}p_{js}^{(i_1, \ldots, i_{\ell})}\leq\notag\\&
\sum_{i_1\in
\mathcal{F}^1}\sum_{i_2\in
\mathcal{F}_{i_1}^2}\cdots\!\!\!\!\!\!
\sum_{i_{\ell-1}\in
\mathcal{F}_{i_{\ell-2}}^{\ell-1}}\sum_{i_{\ell}\in
\mathcal{F}_{i_{\ell-1}}^{\ell}\cap \mathcal{H}_0}\sum_{j=1}^{\prod_{k=1}^{\ell}|\mathcal{F}_{i_{k-1}}^k|}\frac{1}{j}\sum_{s=1}^{\prod_{k=1}^{\ell}|\mathcal{F}_{i_{k-1}}^k|}p_{js}^{(i_1, \ldots, i_{\ell})}\label{mainstop}
\end{align}
Since $\bigcup_{s=1}^{\prod_{k=1}^{\ell}|\mathcal{F}_{i_{k-1}}^k|}\tilde{B}_s^{(i_1, \ldots, i_{\ell})}$ is the whole sample space represented as a union of disjoint events, we obtain
\begin{align}
&\sum_{s=1}^{\prod_{k=1}^{\ell}|\mathcal{F}_{i_{k-1}}^k|}p_{js}^{(i_1, \ldots, i_{\ell})}=\notag\\&
\mathbb{P}\left(P_{i_{\ell}}\in \Bigg(\frac{(j-1)q^{(\ell)}}{\prod_{k=1}^{\ell}|\mathcal{F}_{i_{k-1}}^k|g\left(\prod_{k=1}^{\ell}|\mathcal{F}_{i_{k-1}}^k|\right)}, \frac{jq^{(\ell)}}{\prod_{k=1}^{\ell}|\mathcal{F}_{i_{k-1}}^k|g\left(\prod_{k=1}^{\ell}|\mathcal{F}_{i_{k-1}}^k|\right)}\Bigg],     \bigcup_{s=1}^{\prod_{k=1}^{\ell}|\mathcal{F}_{i_{k-1}}^k|}\tilde{B}_s^{(i_1, \ldots, i_{\ell})}
\right)=\notag\\&\mathbb{P}\left(P_{i_{\ell}}\in \Bigg(\frac{(j-1)q^{(\ell)}}{\prod_{k=1}^{\ell}|\mathcal{F}_{i_{k-1}}^k|g\left(\prod_{k=1}^{\ell}|\mathcal{F}_{i_{k-1}}^k|\right)}, \frac{jq^{(\ell)}}{\prod_{k=1}^{\ell}|\mathcal{F}_{i_{k-1}}^k|g\left(\prod_{k=1}^{\ell}|\mathcal{F}_{i_{k-1}}^k|\right)}\Bigg]\right)=\notag\\&
\mathbb{P}\left(P_{i_{\ell}}\leq \frac{jq^{(\ell)}}{\prod_{k=1}^{\ell}|\mathcal{F}_{i_{k-1}}^k|g\left(\prod_{k=1}^{\ell}|\mathcal{F}_{i_{k-1}}^k|\right)}\right)-\mathbb{P}\left(\frac{(j-1)q^{(\ell)}}{\prod_{k=1}^{\ell}|\mathcal{F}_{i_{k-1}}^k|g\left(\prod_{k=1}^{\ell}|\mathcal{F}_{i_{k-1}}^k|\right)}\right)\label{full}
\end{align}
Combining (\ref{full}) with (\ref{mainstop}) we obtain
\begin{align}&\text{sFDR}^{\ell}\leq
\sum_{i_1\in
\mathcal{F}^1}\sum_{i_2\in
\mathcal{F}_{i_1}^2}\cdots\!\!\!\!\!\!
\sum_{i_{\ell-1}\in
\mathcal{F}_{i_{\ell-2}}^{\ell-1}}\sum_{i_{\ell}\in
\mathcal{F}_{i_{\ell-1}}^{\ell}\cap \mathcal{H}_0}a_{(i_1, \ldots, i_{\ell})}\label{main-2}
\end{align}
where for each sequence of indices $(i_1,\ldots, i_{\ell})$ such that $i_1\in \mathcal{F}_{1},$ $i_2\in \mathcal{F}_{i_1}^2,$ $\ldots,$ $i_{\ell-1}\in\mathcal{F}_{i_{\ell-2}}^{\ell-1}, i_{\ell}\in \mathcal{F}_{i_{\ell-1}}^{\ell}\cap \mathcal{H}_0,$ $a_{(i_1, \ldots, i_{\ell})}$ is defined as follows.
\begin{align*}
&a_{(i_1, \ldots, i_{\ell})}=\\&\sum_{j=1}^{\prod_{k=1}^{\ell}|\mathcal{F}_{i_{k-1}}^k|}\frac{1}{j}\left\{\mathbb{P}\left(P_{i_{\ell}}\leq \frac{jq^{(\ell)}}{\prod_{k=1}^{\ell}|\mathcal{F}_{i_{k-1}}^k|g\left(\prod_{k=1}^{\ell}|\mathcal{F}_{i_{k-1}}^k|\right)}\right)-\mathbb{P}\left(\frac{(j-1)q^{(\ell)}}{\prod_{k=1}^{\ell}|\mathcal{F}_{i_{k-1}}^k|g\left(\prod_{k=1}^{\ell}|\mathcal{F}_{i_{k-1}}^k|\right)}\right)\right\}
\end{align*}
Let us obtain an equivalent expression for $a_{(i_1, \ldots, i_{\ell})}.$
\begin{align}
&a_{(i_1, \ldots, i_{\ell})}=\notag\\&
\sum_{j=1}^{\prod_{k=1}^{\ell}|\mathcal{F}_{i_{k-1}}^k|}\frac{1}{j}\mathbb{P}\left(P_{i_{\ell}}\leq \frac{jq^{(\ell)}}{\prod_{k=1}^{\ell}|\mathcal{F}_{i_{k-1}}^k|g\left(\prod_{k=1}^{\ell}|\mathcal{F}_{i_{k-1}}^k|\right)}\right)-
\notag\\&\sum_{j=0}^{\prod_{k=1}^{\ell}|\mathcal{F}_{i_{k-1}}^k|-1}\frac{1}{j+1}\mathbb{P}\left(P_{i_{\ell}}\leq \frac{jq^{(\ell)}}{\prod_{k=1}^{\ell}|\mathcal{F}_{i_{k-1}}^k|g\left(\prod_{k=1}^{\ell}|\mathcal{F}_{i_{k-1}}^k|\right)}\right)=\notag\\&
\sum_{j=1}^{\prod_{k=1}^{\ell}|\mathcal{F}_{i_{k-1}}^k|-1}\left(\frac{1}{j}-\frac{1}{j+1}\right)\mathbb{P}\left(P_{i_{\ell}}\leq \frac{jq^{(\ell)}}{\prod_{k=1}^{\ell}|\mathcal{F}_{i_{k-1}}^k|g\left(\prod_{k=1}^{\ell}|\mathcal{F}_{i_{k-1}}^k|\right)}\right)+\notag\\&\frac{1}{\prod_{k=1}^{\ell}|\mathcal{F}_{i_{k-1}}^k|}
\mathbb{P}\left(P_{i_{\ell}}\leq \frac{q^{(\ell)}}{g\left(\prod_{k=1}^{\ell}|\mathcal{F}_{i_{k-1}}^k|\right)}\right)\label{pre-last}
\end{align}
For $i_{\ell}\in \mathcal{H}_0,$  $\mathbb{P}(P_{i_{\ell}}\leq x)\leq x$ for all $x\in[0,1],$ therefore
\begin{align}
&\sum_{j=1}^{\prod_{k=1}^{\ell}|\mathcal{F}_{i_{k-1}}^k|-1}\left(\frac{1}{j}-\frac{1}{j+1}\right)\mathbb{P}\left(P_{i_{\ell}}\leq \frac{jq^{(\ell)}}{\prod_{k=1}^{\ell}|\mathcal{F}_{i_{k-1}}^k|g\left(\prod_{k=1}^{\ell}|\mathcal{F}_{i_{k-1}}^k|\right)}\right)+\notag\\&\frac{1}{\prod_{k=1}^{\ell}|\mathcal{F}_{i_{k-1}}^k|}
\mathbb{P}\left(P_{i_{\ell}}\leq \frac{q^{(\ell)}}{g\left(\prod_{k=1}^{\ell}|\mathcal{F}_{i_{k-1}}^k|\right)}\right)\leq\notag\\
&\sum_{j=1}^{\prod_{k=1}^{\ell}|\mathcal{F}_{i_{k-1}}^k|-1} \frac{q^{(\ell)}}{(j+1)\prod_{k=1}^{\ell}|\mathcal{F}_{i_{k-1}}^k|g\left(\prod_{k=1}^{\ell}|\mathcal{F}_{i_{k-1}}^k|\right)}+
\frac{q^{(\ell)}}{\prod_{k=1}^{\ell}|\mathcal{F}_{i_{k-1}}^k|g\left(\prod_{k=1}^{\ell}|\mathcal{F}_{i_{k-1}}^k|\right)}=\notag\\&
\frac{q^{(\ell)}}{\prod_{k=1}^{\ell}|\mathcal{F}_{i_{k-1}}^k|g\left(\prod_{k=1}^{\ell}|\mathcal{F}_{i_{k-1}}^k|\right)}\sum_{j=1}^{\prod_{k=1}^{\ell}|\mathcal{F}_{i_{k-1}}^k|}\frac{1}{j}=
\frac{q^{(\ell)}}{\prod_{k=1}^{\ell}|\mathcal{F}_{i_{k-1}}^k|}\label{last3}
\end{align}
Combining (\ref{main-2}), (\ref{pre-last}), and (\ref{last3}), we obtain
\begin{align*}
&\text{sFDR}^{\ell}\leq \sum_{i_1\in
\mathcal{F}^1}\sum_{i_2\in
\mathcal{F}_{i_1}^2}\cdots\!\!\!\!\!\!
\sum_{i_{\ell-1}\in
\mathcal{F}_{i_{\ell-2}}^{\ell-1}}\sum_{i_{\ell}\in
\mathcal{F}_{i_{\ell-1}}^{\ell}\cap \mathcal{H}_0}\frac{q^{(\ell)}}{\prod_{k=1}^{\ell}|\mathcal{F}_{i_{k-1}}^k|}\leq
\\&\sum_{i_1\in
\mathcal{F}^1}\sum_{i_2\in
\mathcal{F}_{i_1}^2}\cdots\!\!\!\!\!\!
\sum_{i_{\ell-1}\in
\mathcal{F}_{i_{\ell-2}}^{\ell-1}}\sum_{i_{\ell}\in
\mathcal{F}_{i_{\ell-1}}^{\ell}}\frac{q^{(\ell)}}{\prod_{k=1}^{\ell}|\mathcal{F}_{i_{k-1}}^k|}=q^{(\ell)}
\end{align*}
Thus the proof is complete.

\subsubsection*{Proof of Lemma \ref{lemtreeBH}}
In the proofs of items, let $\ell\in\{1, \ldots, L\}$ be an arbitrary fixed level.
\paragraph{Proof of item 1} Let us fix a sequence $(i_1, \ldots, i_{\ell})$ such that $i_1\in \mathcal{F}^1, i_2\in \mathcal{F}_{i_1}^2, \ldots, i_{\ell}\in \mathcal{F}_{i_{\ell-1}}^{\ell},$ and sequences $(r_1,\ldots r_{\ell}), (r'_1,\ldots r'_{\ell})$ such that
$(r_1,\ldots r_{\ell})\neq(r'_1,\ldots r'_{\ell}),$ and
$r_k, r'_k\in \{1,\ldots, |\mathcal{F}_{i_{k-1}}^{k}|\}$ for $k=1, \ldots, \ell.$ It is easy to see that for any $0<\alpha<1,$ the events $C_{r_k}^{(i_k)}[\alpha]$ and $C_{r'_k}^{(i_k)}[\alpha]$ are disjoint for each $k=1, \ldots, \ell.$ Assume that $r_1\neq r'_1.$ Then the events $C_{r_1}[q^{(1)}]$ and $C_{r'_1}[q^{(1)}]$ are disjoint, which yields that the events $C_{r_1, \ldots, r_{\ell}}^{(i_1, \ldots, i_{\ell})}$ and $C_{r'_1, \ldots, r'_{\ell}}^{(i_1, \ldots, i_{\ell})}$ are disjoint. Assume now that $r_1=r'_1.$ Let $k_0=\min\{j: r_j\neq r'_j\}.$ According to our assumptions, $1<k_0\leq \ell.$ Let $q_{i_{k-1}}$ and $q'_{i_{k-1}}$ be the testing levels of BH for $\mathcal{F}_{i_{k-1}}^{k}$ according to TreeBH, when the numbers of rejections in $\mathcal{F}_{i_{j-1}}^{j},$ $j=1, \ldots, k-1$ are $r_1, \ldots, r_{k-1}$ and $r'_1, \ldots, r'_{k-1},$ respectively. Since $r_j=r'_j$ for each $j=1, \ldots, k_0-1,$  we obtain $q_{i_{k_0-1}}=q'_{i_{k_0-1}}.$
Hence the events $C_{r_{k_0}}^{(i_{k_0})}[q_{i_{k_0-1}}] $ and  $C_{r'_{k_0}}^{(i_{k_0})}[q'_{i_{k_0-1}}] $ are disjoint, because $r_{k_0}\neq r'_{k_0}.$ This yields that the events  $C_{r_1, \ldots, r_{\ell}}^{(i_1, \ldots, i_{\ell})}$ and $C_{r'_1, \ldots, r'_{\ell}}^{(i_1, \ldots, i_{\ell})}$ are disjoint.
\paragraph{Proof of item 2} Let us fix a sequence $(i_1, \ldots, i_{\ell})$ such that $i_1\in \mathcal{F}^1, i_2\in \mathcal{F}_{i_1}^2, \ldots, i_{\ell}\in \mathcal{F}_{i_{\ell-1}}^{\ell}.$ By the contrary, assume that there exist $s, s'\in\{1, \ldots, \prod_{k=1}^{\ell}\mathcal{F}_{i_{k-1}}^k\}$ such that $s\neq s'$ and both $B_s^{(i_1, \ldots, i_{\ell})}$ and $B_{s'}^{(i_1, \ldots, i_{\ell})}$ occur.
Then there exists $(r_1,\ldots r_{\ell})\in I_s^{(i_1, \ldots, i_{\ell})}$ and $(r'_1,\ldots r'_{\ell})\in I_{s'}^{(i_1, \ldots, i_{\ell})}$
such that both $C_{r_1, \ldots, r_{\ell}}^{(i_1, \ldots, i_{\ell})}$ and $C_{r'_1, \ldots, r'_{\ell}}^{(i_1, \ldots, i_{\ell})}$ occur. However, $\prod_{k=1}^{\ell}r_k=s,$ while  $\prod_{k=1}^{\ell}r'_k=s',$ therefore $(r_1,\ldots r_{\ell})\neq(r'_1,\ldots r'_{\ell}).$ Thus we obtain a contradiction to item 1 of Lemma \ref{lemtreeBH}. We have proved that the events $B_s^{(i_1, \ldots, i_{\ell})}$ and $B_{s'}^{(i_1, \ldots, i_{\ell})}$ are disjoint for $s\neq s'.$ Let us now prove that the event $\cup_{s=1}^{\prod_{k=1}^{\ell}|\mathcal{F}_{i_{k-1}}^{k}|}B_s^{(i_1, \ldots, i_{\ell})}$ occurs with probability one.
Note that for each $0<\alpha<1$ and $k\in\{1, \ldots, \ell\},$ the event $\cup_{r_k=1}^{|\mathcal{F}_{i_{k-1}}^k|}C_{r_k}^{(i_k)}[\alpha]$ occurs with probability one. 
 Let $\bm p[\cup_{k=1}^{\ell}\mathcal{F}_{i_{k-1}}^k\setminus\{i_k\}]$ be an arbitrary fixed vector of $p$-values corresponding to the hypotheses with indices in $\cup_{k=1}^{\ell}\mathcal{F}_{i_{k-1}}^k\setminus\{i_k\}.$ If only the $p$-values for the finest level $L$ hypotheses are available, and the $p$-values for coarser resolution hypotheses are combinations of Level-$L$ $p$-values, then let $\bm p[\cup_{k=1}^{\ell}\mathcal{F}_{i_{k-1}}^k\setminus\{i_k\}]$ be a vector of $p$-values corresponding to the hypotheses with indices in $\cup_{k=1}^{\ell}\mathcal{F}_{i_{k-1}}^k\setminus\{i_k\},$ which is based on an arbitrary fixed vector of Level-$L$ $p$-values. Consider the subvector of $\bm p[\cup_{k=1}^{\ell}\mathcal{F}_{i_{k-1}}^k\setminus\{i_k\}]$ corresponding to the hypotheses in $\mathcal{F}^1\setminus\{i_1\}.$ For this subvector, there exists $r_1^0$ such that $C_{r_1^0}[q^{(1)}]$ occurs. Consider now the subvector of $\bm p[\cup_{k=1}^{\ell}\mathcal{F}_{i_{k-1}}^k\setminus\{i_k\}]$ corresponding to the hypotheses in $\mathcal{F}^2_{i_1}\setminus\{i_2\}.$
For this subvector, there exists $r_2^0$ such that $C_{r_2^0}\left[r_1^0q^{(1)}/|\mathcal{F}^1|\right]$ occurs. Continuing in this fashion, we obtain that there exists a vector $(r_1^0, \ldots, r_{\ell}^0)$ such that the event $\cap_{k=1}^{\ell}C_{r_k^0}^{(i_k)}[q_{i_{k-1}}]=C_{r_1^0, \ldots, r_{\ell}^0}^{(i_1, \ldots, i_{\ell})} $ occurs for $\bm p[\cup_{k=1}^{\ell}\mathcal{F}_{i_{k-1}}^k\setminus\{i_k\}].$ Letting $s_0=\prod_{k=1}^{\ell}r_k^0,$ we obtain that the event $B_{s_0}^{(i_1, \ldots, i_{\ell})}$ occurs for $\bm p[\cup_{k=1}^{\ell}\mathcal{F}_{i_{k-1}}^k\setminus\{i_k\}].$ Thus we have proved that the event $\cup_{s=1}^{\prod_{k=1}^{\ell}|\mathcal{F}_{i_{k-1}}^{k}|}B_s^{(i_1, \ldots, i_{\ell})}$ occurs with probability one. 
\paragraph{Proof of item 3} The error rate $\text{sFDR}^{\ell}$ has the following equivalent definition.
Let $\overline{\text{sFDP}}^{1}=\text{FDP}^1,$ and
\begin{eqnarray*}\overline{\text{sFDP}}^{\ell}&\!\!\!\!=\!\!\!\!&
\sum_{i_1\in
\mathcal{S}^1_{0}}\frac{1}{\max\{|\mathcal{S}^1_{0}|,1\}}\sum_{i_2\in
\mathcal{S}^2_{i_1}}\frac{1}{\max\{|\mathcal{S}^2_{i_1}|,1\}}\cdots\!\!\!\!\sum_{i_{\ell-1}\in
\mathcal{S}^{\ell-1}_{i_{\ell-2}}}\frac{1}{\max\{|\mathcal{S}^{\ell-1}_{i_{\ell-2}}|,1\}}\sum_{i_{\ell}\in
\mathcal{S}^{\ell}_{i_{\ell-1}}\cap \mathcal{H}_0}\frac{1}{\max\{|\mathcal{S}^{\ell}_{i_{\ell-1}}|,1\}}
\\
\text{sFDR}^{\ell}&\!\!\!\!=\!\!\!\!&\mathbb{E}(\overline{\text{sFDP}}^{\ell})
\end{eqnarray*}
Therefore, $\text{sFDR}^{\ell}$ is equal to
\begin{align}
&\sum_{i_1\in
\mathcal{F}^1}\sum_{i_2\in
\mathcal{F}_{i_1}^2}\cdots\!\!\!\!\!\!
\sum_{i_{\ell-1}\in
\mathcal{F}_{i_{\ell-2}}^{\ell-1}}\sum_{i_{\ell}\in
\mathcal{F}_{i_{\ell-1}}^{\ell}\cap \mathcal{H}_0}\sum_{j=1}^{\ell}\sum_{r_j=1}^{|\mathcal{F}_{i_{j-1}}^{j}|}\frac{\mathbb{P}\left(i_k\in \mathcal{S}_{i_{k-1}}^{k}, \,\,|\mathcal{S}_{i_{k-1}}^{k}|=r_k \text{ for } k=1, \ldots, \ell\right)}{\prod_{k=1}^{\ell}r_k}\label{lemma-main-1}
\end{align}
Let us fix a sequence $(i_1, \ldots, i_{\ell})$ such that $i_1\in \mathcal{F}^1, i_2\in \mathcal{F}_{i_1}^2, \ldots, i_{\ell}\in \mathcal{F}_{i_{\ell-1}}^{\ell},$ $i_{\ell}\in \mathcal{H}_0,$ and a sequence $(r_1,\ldots r_{\ell})$ such that $r_k\in \{1,\ldots |\mathcal{F}_{i_{k-1}}^{k}|\}$ for $k=1, \ldots, \ell.$
Note that
\begin{align*}
&\mathbb{P}\left(i_k\in \mathcal{S}_{i_{k-1}}^{k}, \,\,|\mathcal{S}_{i_{k-1}}^{k}|=r_k \text{ for } k=1, \ldots, \ell\right)=
\mathbb{P}\left(\cap_{k=1}^{\ell}\left[\{P_{i_k}\leq r_kq_{i_{k-1}}/|\mathcal{F}_{i_{k-1}}^{k}|\}\cap\{|\mathcal{S}_{i_{k-1}}^{k}|=r_k\}\right]\right)=\\&
\mathbb{P}\left(\cap_{k=1}^{\ell}\left[\{P_{i_k}\leq r_kq_{i_{k-1}}/|\mathcal{F}_{i_{k-1}}^{k}|\}\cap C_{r_k}^{(i_k)}[q_{i_{k-1}}]\right]\right)=
\mathbb{P}\left(\cap_{k=1}^{\ell}\{P_{i_k}\leq r_kq_{i_{k-1}}/|\mathcal{F}_{i_{k-1}}^{k}|\}\cap\left[ \cap_{k=1}^{\ell}C_{r_k}^{(i_k)}[q_{i_{k-1}}]\right]\right)
\end{align*}
Recall that $\cap_{k=1}^{\ell}C_{r_k}^{i_k}[q_{i_{k-1}}]=C_{r_1, \ldots, r_{\ell}}^{(i_1, \ldots, i_{\ell})}$ by definition, thus we obtain
\begin{align*}
&\mathbb{P}\left(i_k\in \mathcal{S}_{i_{k-1}}^{k}, \,\,|\mathcal{S}_{i_{k-1}}^{k}|=r_k \text{ for } k=1, \ldots, \ell\right)=
\mathbb{P}\left(\cap_{k=1}^{\ell}\{P_{i_k}\leq r_kq_{i_{k-1}}/|\mathcal{F}_{i_{k-1}}^{k}|\}\cap C_{r_1, \ldots, r_{\ell}}^{(i_1, \ldots, i_{\ell})}\right)\leq\\&\mathbb{P}\left(P_{i_{\ell}}\leq r_{\ell}q_{i_{\ell-1}}/|\mathcal{F}_{i_{\ell-1}}^{\ell}|, C_{r_1, \ldots, r_{\ell}}^{(i_1, \ldots, i_{\ell})}\right)=
\mathbb{P}\left(P_{i_{\ell}}\leq \frac{\left(\prod_{k=1}^{\ell}r_k\right)q^{(\ell)}}{\prod_{k=1}^{\ell}|\mathcal{F}_{i_{k-1}}^{k}|}, C_{r_1, \ldots, r_{\ell}}^{(i_1, \ldots, i_{\ell})}\right),
\end{align*}
where the last equality follows from the definition of $q_{i_{\ell-1}}$ in (\ref{q_i}).
Thus we have proved that
\begin{align}
&\mathbb{P}\left(i_k\in \mathcal{S}_{i_{k-1}}^{k}, \,\,|\mathcal{S}_{i_{k-1}}^{k}|=r_k \text{ for } k=1, \ldots, \ell\right)\leq\mathbb{P}\left(P_{i_{\ell}}\leq \frac{\left(\prod_{k=1}^{\ell}r_k\right)q^{(\ell)}}{\prod_{k=1}^{\ell}|\mathcal{F}_{i_{k-1}}^{k}|}, C_{r_1, \ldots, r_{\ell}}^{(i_1, \ldots, i_{\ell})}\right)\label{main-lem-2}
\end{align}
for any fixed sequences $(i_1, \ldots, i_{\ell})$ and $(r_1,\ldots r_{\ell})$ such that $i_1\in \mathcal{F}^1, i_2\in \mathcal{F}_{i_1}^2, \ldots, i_{\ell}\in \mathcal{F}_{i_{\ell-1}}^{\ell},$ $i_{\ell}\in \mathcal{H}_0,$ and $r_k\in \{1, \ldots, \mathcal{F}_{i_{k-1}}^{k}\}$ for $k=1, \ldots, \ell.$
Combining (\ref{main-lem-2}) with (\ref{lemma-main-1}), we obtain
\begin{align*}
&\text{sFDR}^{\ell}\leq \\&\sum_{i_1\in
\mathcal{F}^1}\sum_{i_2\in
\mathcal{F}_{i_1}^2}\cdots\!\!\!\!\!\!
\sum_{i_{\ell-1}\in
\mathcal{F}_{i_{\ell-2}}^{\ell-1}}\sum_{i_{\ell}\in
\mathcal{F}_{i_{\ell-1}}^{\ell}\cap \mathcal{H}_0}\sum_{j=1}^{\ell}\sum_{r_j=1}^{|\mathcal{F}_{i_{j-1}}^{j}|}\frac{1}{\prod_{k=1}^{\ell}r_k}\mathbb{P}\left(P_{i_{\ell}}\leq \frac{\left(\prod_{k=1}^{\ell}r_k\right)q^{(\ell)}}{\prod_{k=1}^{\ell}|\mathcal{F}_{i_{k-1}}^{k}|}, C_{r_1, \ldots, r_{\ell}}^{(i_1, \ldots, i_{\ell})}\right)=\\&\sum_{i_1\in
\mathcal{F}^1}\sum_{i_2\in
\mathcal{F}_{i_1}^2}\sum_{r_{2}=1}^{|\mathcal{F}_{i_1}^2|}\cdots\!\!\!\!\!\!
\sum_{i_{\ell-1}\in
\mathcal{F}_{i_{\ell-2}}^{\ell-1}}\sum_{i_{\ell}\in
\mathcal{F}_{i_{\ell-1}}^{\ell}\cap \mathcal{H}_0}\sum_{s=1}^{\prod_{k=1}^{\ell}|\mathcal{F}_{i_{k-1}}^{k}|}\sum_{(r_1, \ldots, r_{\ell})\in I_s^{(i_1, \ldots, i_{\ell})}}\frac{1}{s}\mathbb{P}\left(P_{i_{\ell}}\leq \frac{sq^{(\ell)}}{\prod_{k=1}^{\ell}|\mathcal{F}_{i_{k-1}}^{k}|}, C_{r_1, \ldots, r_{\ell}}^{(i_1, \ldots, i_{\ell})}\right)=\\&
\sum_{i_1\in
\mathcal{F}^1}\sum_{i_2\in
\mathcal{F}_{i_1}^2}\sum_{r_{2}=1}^{|\mathcal{F}_{i_1}^2|}\cdots\!\!\!\!\!\!
\sum_{i_{\ell-1}\in
\mathcal{F}_{i_{\ell-2}}^{\ell-1}}\sum_{i_{\ell}\in
\mathcal{F}_{i_{\ell-1}}^{\ell}\cap \mathcal{H}_0}\sum_{s=1}^{\prod_{k=1}^{\ell}|\mathcal{F}_{i_{k-1}}^{k}|}\frac{1}{s}\mathbb{P}\left(P_{i_{\ell}}\leq \frac{sq^{(\ell)}}{\prod_{k=1}^{\ell}|\mathcal{F}_{i_{k-1}}^{k}|}, \cup_{(r_1, \ldots, r_{\ell})\in I_s^{(i_1, \ldots, i_{\ell})}}C_{r_1, \ldots, r_{\ell}}^{(i_1, \ldots, i_{\ell})}\right)=\\&
\sum_{i_1\in
\mathcal{F}^1}\sum_{i_2\in
\mathcal{F}_{i_1}^2}\sum_{r_{2}=1}^{|\mathcal{F}_{i_1}^2|}\cdots\!\!\!\!\!\!
\sum_{i_{\ell-1}\in
\mathcal{F}_{i_{\ell-2}}^{\ell-1}}\sum_{i_{\ell}\in
\mathcal{F}_{i_{\ell-1}}^{\ell}\cap \mathcal{H}_0}\sum_{s=1}^{\prod_{k=1}^{\ell}|\mathcal{F}_{i_{k-1}}^{k}|}\frac{1}{s}\mathbb{P}\left(P_{i_{\ell}}\leq \frac{sq^{(\ell)}}{\prod_{k=1}^{\ell}|\mathcal{F}_{i_{k-1}}^{k}|}, B_s^{(i_1,\ldots, i_{\ell})}\right),
 \end{align*}
 where the second equality follows from the fact that the events in the set $\{C_{r_1, \ldots, r_{\ell}}^{(i_1, \ldots, i_{\ell})}, (r_1, \ldots, r_{\ell})\in I_s^{(i_1, \ldots, i_{\ell})}\}$ are disjoint for every $s$ (due to item 1 of Lemma \ref{lemtreeBH}), and the last equality holds since $B_s^{(i_1,\ldots, i_{\ell})}$ is the union of the events in this set.

\subsection*{Methodology and theoretical results for a general class of error rates}\label{general-error}
Similarly to \cite{BB14}, we consider a general class of error rates that can be written in the form
$\mathbb{E}(\mathcal{C})$ for some measure of errors $\mathcal{C}.$ These error rates include the family-wise error rate (FWER),
$\mathbb{P}(V\geq 1)=\mathbb{E}(\textbf{I}_{\{V\geq 1\}}),$ where $V$ is the number of type I errors, the FDR, the weighted FDR, and others (see \cite{BB14} for additional examples).
Assume that each family $\mathcal{F}_i^{\ell}$ is assigned with an error measure $\mathcal{C}_i^{\ell}.$ The general error rate we address is defined as follows.
\begin{align*}&\mathbb{E}(s\mathcal{C}^{\ell})=\mathbb{E}\left(\sum_{i_1\in
\mathcal{S}^1}\frac{1}{\max\{|\mathcal{S}^1|, 1\}}\!\!\sum_{i_2\in
\mathcal{S}^2_{i_1}}\frac{1}{\max\{|\mathcal{S}^2_{i_1}|, 1\}}\cdots\!\!\!\!\!\!\!\!\sum_{i_{\ell-1}\in
\mathcal{S}^{\ell-1}_{i_{\ell-2}}}\!\!\frac{1}{\max\{|\mathcal{S}^{\ell-1}_{i_{\ell-2}}|, 1\}}\mathcal{C}^{\ell}_{i_{\ell-1}}\!\!\right)
\end{align*}
for $\ell>1,$ and $\mathbb{E}(s\mathcal{C}^{1})=\mathbb{E}\left(\mathcal{C}^{1}_0\right),$
or equivalently
\begin{align*}
&\text{For each selected hypothesis $H_i$ at level $\ell-1,$}\,\,\,
\overline{\mathcal{C}}_{i}=\mathcal{C}_i^{\ell}.\\&\text{For each
selected hypothesis $H_j$ at level $k\in\{0,\ldots,\ell-2\},$}\,\,\,
\overline{\mathcal{C}}_{j}=\frac{\sum_{r\in
\mathcal{S}_{j}^{k+1}}\overline{\mathcal{C}}_{r}}{\max\{|\mathcal{S}_{j}^{k+1}|, 1\}}.
\end{align*}
Then
\begin{align*}
\mathbb{E}(s\mathcal{C}^{\ell})=\mathbb{E}(\overline{\mathcal{C}}_{0}),
\end{align*}
where $H_{0}$ (which is identified with $H_{i_0}$) is the root node hypothesis residing at Level 0, i.e. the parent of $\mathcal{F}^1$ that is always selected.
We next present a general hierarchical testing strategy
for the case
where the investigator is interested only in the discoveries in
specific levels $\ell\in \mathcal{L},$ where
$\mathcal{L}\subseteq\{1,\ldots,L\},$ targeting control of
$\mathbb{E}(s\mathcal{C}^{\ell})$ for all $\ell\in \mathcal{L}.$ In
this case one may use any selection rules satisfying Definition \ref{simple1}
for selecting the
hypotheses within the families belonging to levels $\ell\notin
\mathcal{L},$ since they only serve for selection of families
residing at the levels of interest.
\newpage\makeatletter
\def\BState{\State\hskip-\ALG@thistlm}
\makeatother
\begin{algorithm}[h!]
\DontPrintSemicolon \caption{a general selection-adjusted
hierarchical procedure\label{TreeProc3}}
\SetKwInOut{Input}{Input}\SetKwInOut{Output}{Output}
\Input{The target levels for  $\mathbb{E}(s\mathcal{C}^{\ell}),$ $\ell\in \mathcal{L}: \{q^{(i)}, i\in \mathcal{L} \}$\;
For each $\ell\in \mathcal{L}$ and family $\mathcal{F}_{i}^{\ell},$
the $\mathbb{E}(\mathcal{C}_i^{\ell})$-controlling multiple testing
procedure for testing family $\mathcal{F}_{i}^{\ell}$ \; For each
$\ell\notin \mathcal{L},$ $\ell\leq K,$ where $K=\max\{\ell: \ell\in
\mathcal{L} \},$ and family $\mathcal{F}_{i}^{\ell},$ the selection
rule defining the set $\mathcal{S}_i^{\ell}$\; The $p$-values for
all the hypotheses at levels $1,\ldots, K$ \;}
\Begin{ $S^0 \longleftarrow {0}$\; $q_{0}=q^{(1)}$\;
\For{$\ell = 1,\ldots,K$}{$\mathcal{S}^{\ell} \longleftarrow
\emptyset$\; \If{$\ell\notin \mathcal{L}$} {$q^{(\ell)}\longleftarrow 1$}  }}
\For{$\ell=1,\ldots,K$}{\If{$\mathcal{S}^{\ell-1} \neq
\emptyset$}{\For{$i\in \mathcal{S}^{\ell-1}$}{If $\ell\in
\mathcal{L},$ apply the
$\mathbb{E}(\mathcal{C}_i^{\ell})$-controlling multiple testing
procedure at level $q_i$ on the $p$-values of family
$\mathcal{F}_i^{\ell},$ defining the set $\mathcal{S}_i^{\ell}$\; If $\ell\notin \mathcal{L},$ apply the
assigned selection rule for selecting the hypotheses in family
$\mathcal{F}_i^{\ell},$ defining the set $\mathcal{S}_i^{\ell}$ \;
$\mathcal{S}^{\ell}\longleftarrow \mathcal{S}^{\ell}\bigcup
\mathcal{S}_i^{\ell}$\; \For {$j\in \mathcal{S}_i^{\ell}$}{
$q_{j}\longleftarrow q^{(\ell+1)}\times q_i/q^{(\ell)}\times
|\mathcal{S}_i^{\ell}|/|\mathcal{F}_i^{\ell}|$\;} } } } \Output{The
set of all the selected hypotheses, $\bigcup_{\ell=1}^K
\mathcal{S}^{\ell}.$}
\end{algorithm}
The input for the general procedure above can be given in a sequential manner, i.e. the $p$-values for family $\mathcal{F}_i^{\ell}$ can be obtained only after $H_i$ is selected, and only at this stage the decision regarding the selection rule (or the multiple testing procedure) for this family has to be made. This general procedure reduces to the TreeBH if $\mathcal{L}=\{1, \ldots, L\},$ and for each family $\mathcal{F}_{i}^{\ell},$ $\mathcal{C}_i^{\ell}=\text{FDP}$ and the assigned multiple testing procedure is the BH procedure.
\noindent In order to prove that the general selection-adjusted hierarchical procedure (TreeC) controls
$\mathbb{E}(s\mathcal{C}^{\ell})$ for each $\ell\in \mathcal{L},$ we make the following assumptions:
\begin{enumerate}
\item [A1] For each family $\mathcal{F}_i^{\ell}$ where $\ell\in \mathcal{L},$ the error rate $\mathbb{E}(\mathcal{C}_i^{\ell})$ is such that $\mathcal{C}_i^{\ell}$ takes values in a countable set, and the procedure used for testing this family can control  $\mathbb{E}(\mathcal{C}_i^{\ell})$  at any desired level under the dependence among the $p$-values in $\mathcal{F}_i^{\ell}.$
\item [A2] For each $\ell\geq 2$ such that $\ell\in \mathcal{L},$ and a family $\mathcal{F}_i^{\ell},$  the $p$-values for the hypotheses in $\mathcal{F}_i^{\ell}$ are independent of the $p$-values for the hypotheses in $\cup_{k=1}^{\ell-1}\left(\mathcal{F}_{P^k(i)}^{\ell-k}\setminus \{H_{P^{k-1}(i)}\}\right),$ where $P^0(i)=i.$
\item [A3] For each family $\mathcal{F}_i^{k}$ where $k\leq \max\{\ell: \ell\in \mathcal{L}\}$ the selection rule defining the set $\mathcal{S}_i^k$ (which is defined by a multiple testing procedure for $k\in \mathcal{L}$) is a simple selection rule (see definition below).
\end{enumerate}
As in \cite{BB14} we note that for all practical purposes
$\mathcal{C}_i^{\ell}$ is a count or a ratio of counts, therefore it
takes values in a countable set, as required in A1. Below we
accommodate to our setting the definition of a simple selection
rule, given in \cite{BB14}. In our context, selection of a parent
hypothesis is equivalent to the selection of the family it indexes.
\begin{definition}\label{simple1} (simple selection rule). A selection rule is called simple if for each selected hypothesis $H_j\in \mathcal{F}_i^\ell,$ when the p-values belonging to $\mathcal{F}_i^\ell \setminus \{H_j\}$ are fixed,
and the $p$-value for $H_j$ can change as long as $H_j$ is selected,
the number of selected hypotheses in $\mathcal{F}_i^\ell,$ i.e.
$|\mathcal{S}_i^\ell |,$ remains unchanged.
\end{definition}
The requirement that the selection rule is simple is a very lenient
requirement. It is shown in \cite{BB14} that any step-up or
step-down non-adaptive multiple testing procedure defines a simple
selection rule. Therefore, the BH procedure, as well as its variants incorporating penalty and/or prior weights on the hypotheses (the weighted FDR-controlling procedure of \cite{BH97}, the FDR-controlling procedure of \cite{GRW06} incorporating prior weights, and the procedure of \cite{blanchard2008two} which incorporates both penalty and prior weights), define simple selection rules. Moreover, it is easy to see that the adaptive BH procedure that incorporates the estimator of \cite{SetS04} for the proportion of true null hypotheses defines a simple selection rule, as well as a generalized procedure of \cite{RetJ17}, incorporating  the estimator of \cite{SetS04}, and both prior and penalty weights.

\begin{thm}\label{thm-alg-3}
If assumptions A1--A3 hold, the TreeC procedure with input
parameters $q^{(i)}, i\in \mathcal{L},$ guarantees for each $\ell\in
\mathcal{L}$ that
$$\mathbb{E}(s\mathcal{C}^{\ell})\leq q^{(\ell)}.$$
\end{thm}
The result of Theorem 2 in the main manuscript follows from Theorem \ref{thm-alg-3}. Theorem 2 in the main manuscript considers the TreeC procedure with $\mathcal{L}=\{1, \ldots, L\},$ where each selected family is tested by a multiple testing procedure, and its assumptions are equivalent to the assumptions of Theorem \ref{thm-alg-3} for this case.


Assumption A2 does not imply any
special dependency structure between the $p$-values of hypotheses in
the same branch of
the tree, i.e.\ 
ancestor hypotheses of each $H_i$ residing at level $L:$  these can be arbitrarily dependent. 
For example, it holds when each parent hypothesis is the intersection of all its children,  the $p$-values within each family at level $L$ are PRDS, they are independent
of the $p$-values in any other family at Level $L,$ and the $p$-value for each parent hypothesis
is obtained by Simes' method applied on the $p$-values in the family it indexes. In this case all the $p$-values in the tree are valid, and the $p$-values for the hypotheses residing at the same level $\ell$ are independent, for $\ell=1, \ldots, L-1.$ If one has prior and/or penalty weights for all the hypotheses in the tree, one can replace the BH procedure in the TreeBH algorithm by the doubly-weighted BH procedure of \cite{blanchard2008two} for testing the selected Level-$L$ families, since it guarantees weighted FDR control under PRDS dependency (\cite{blanchard2008two}). For testing coarser resolution families, one can in addition incorporate the estimator of \cite{SetS04} for the proportion of true null hypotheses within a family, by using the procedure of \cite{RetJ17}, which is valid under independence (\cite{RetJ17}). Due to Theorem \ref{thm-alg-3}, the corresponding algorithm guarantees control of $\mathbb{E}(s\mathcal{C}^{\ell})$ for $\ell=1, \ldots, L,$ where for each family $\mathcal{F}_i^{\ell},$ $\mathcal{C}_i^{\ell}$ is the weighted FDP. Similarly, one can consider other selective error rates, possibly assigning different error measures to different families, and corresponding multiple testing procedures satisfying assumptions A1--A3.

The following theorem is a generalization of Theorem 3 in the main manuscript. 
 According to the definition of a simple selection rule, for each such selection rule, family $\mathcal{F}_{i}^{\ell},$ $H_j\in \mathcal{F}_{i}^{\ell}$ and $k\in\
\{1, \ldots, |\mathcal{F}_{i}^{\ell}|\},$ we can define the event $C_k^{(j)}$ on the space of $p$-values in $\mathcal{F}_{i}^{\ell}\setminus\{H_j\},$  in which if $j\in \mathcal{S}_i^{\ell},$ then the number of selected hypotheses in $\mathcal{F}_{i}^{\ell}$ is $k.$ If the simple selection rule is defined by a multiple testing procedure, then
we can define the event $C_k^{(j)}[\alpha]$ on the space of $p$-values in $\mathcal{F}_{i}^{\ell}\setminus\{H_j\},$  in which if $H_j$ is rejected by the given procedure at level $\alpha,$ then the number of rejected hypotheses in $\mathcal{F}_{i}^{\ell}$ is $k.$
Let us now consider for each simple selection rule and family $\mathcal{F}_{i}^{\ell}$ the event \begin{align}D_s^{(j)}=\cup_{k=1}^{s}C_k^{(j)}\label{d1}\end{align} for $s=1,\ldots, |\mathcal{F}_{i}^{\ell}|$ and $H_j\in \mathcal{F}_{i}^{\ell}.$ This is the event in which if $j\in \mathcal{S}_i^{\ell},$ then the number of selected hypotheses in $\mathcal{F}_{i}^{\ell}$ is at most $s.$ Similarly, if the selection within the family $\mathcal{F}_{i}^{\ell}$ is made by a multiple testing procedure, then we define \begin{align}D_s^{(j)}[\alpha]=\cup_{k=1}^{s}C_k^{(j)}[\alpha]\label{d2}\end{align} for $s=1,\ldots, |\mathcal{F}_{i}^{\ell}|$ and $H_j\in \mathcal{F}_{i}^{\ell}.$ When the set $D_s^{(j)}$ is an increasing set for each $s$ and $j,$ the selection rule is called \textit{concordant}. When the selection rule is defined by a multiple testing procedure, it is called concordant if $D_s^{(j)}[\alpha]$ is an increasing set for every $0<\alpha<1.$   This definition is closely related to Definition 5 of \cite{BY05j} and to Definition 3 of \cite{BB14}. We consider here the concordance property only for simple selection rules, while in \cite{BB14} the definition of concordance is given for general selection rules.

\begin{thm}\label{thm2-gen}
Assume that each parent hypothesis is the intersection of the hypotheses within the family it indexes. Assume that the $p$-values for the hypotheses at the finest Level $L$ satisfy the PRDS property, and each of the $p$-values for the hypotheses at Level $\ell-1$ is calculated using a certain combination rule on the $p$-values of the family of hypotheses it indexes at Level $\ell,$ for $\ell=2, \ldots, L,$
where each such combination rule is a monotone non-decreasing function of each of its input $p$-values. Consider the TreeC procedure where $L\in\mathcal{L},$  each selected Level-$L$ family  is tested using the BH procedure, and the selection rules for families at levels $\ell=1, \ldots, L-1$ are simple and concordant. Under these conditions we have the following results:
\begin{enumerate}
\item The TreeC procedure guarantees  $\text{sFDR}^L\leq q^{(L)}$ where $q^{(L)}$ is the targeted $\text{sFDR}^{L}$ input parameter.
\item If the TreeC  procedure guarantees for each $\ell=2, \ldots, L$ that if a parent hypothesis at level $\ell-1$ is selected, at least one of the hypotheses within the family it indexes is selected, then the TreeC procedure guarantees $\text{sFDR}^{\ell}\leq q^{(L)}$ for each $\ell\in\{1, \ldots, L\}.$
\end{enumerate}
\end{thm}
It follows from the proof of Theorem 1 in the main manuscript that the BH procedure defines a simple and a concordant selection rule. Therefore, item 1 of Theorem 3 in the main manuscript follows from item 1 of Theorem \ref{thm2-gen}. When the target levels of TreeBH satisfy $q^{(1)}= q^{(2)}= \ldots= q^{(L)}$ and Simes' combination rule is used for calculating the $p$-values for the parent hypotheses, the TreeBH procedure guarantees for each $\ell=2, \ldots, L$ that if a parent hypothesis at level $\ell-1$ is selected, at least one of the hypotheses within the family it indexes is selected, as explained in Section 3 in the main manuscript. In addition, Simes' combination rule is
a monotone non-decreasing function of each of its input $p$-values. Therefore, item 2 of Theorem 3 in the main manuscript follows from item 2 of Theorem \ref{thm2-gen}. Thus we have shown that indeed Theorem \ref{thm2-gen} implies the result of Theorem 3 in the main manuscript, and gives its generalization for selection rules other than BH and $p$-values for parent hypotheses based on methods other than Simes' method.

\begin{proof}[Proof of Theorem \ref{thm-alg-3}]
The proof is similar to the proof of Theorem 1 in \cite{BB14}. It is based on the following lemma. We identify each family $\mathcal{F}_i^{\ell}$ with the set of indices of hypotheses that belong to $\mathcal{F}_i^{\ell}.$
\begin{lem}\label{lem1}
If Assumption A3 holds, then for each $\ell\in\mathcal{L}$ such that $\ell>1$ and any sequence of indices $i_1, i_2, \ldots, i_{\ell-1}$ such that
$i_1\in \mathcal{F}^1,$ $i_2\in
\mathcal{F}_{i_1}^2,\ldots,i_{\ell-1}\in
\mathcal{F}_{i_{\ell-2}}^{\ell-1},$ the following property holds: if the $p$-values of $H_{i_1},$ $H_{i_2},\ldots, H_{i_{\ell-1}}$ can change as long as $H_{i_1},$ $H_{i_2},\ldots, H_{i_{\ell-1}}$ are selected, while the $p$-values of hypotheses in the sets $\mathcal{F}^1\setminus\{i_1\},$ $\mathcal{F}_{i_1}^2\setminus\{i_2\},\ldots, \mathcal{F}_{i_{\ell-2}}^{\ell-1}\setminus\{i_{\ell-1}\}$ are fixed, the numbers of selected hypotheses in families $\mathcal{F}^1, \mathcal{F}_{i_1}^2,\ldots, \mathcal{F}_{i_{\ell-2}}^{\ell-1}$ remain unchanged.
\end{lem}
 Consider first the case where $\ell=1$ and $\ell\in \mathcal{L}.$ Then $\mathbb{E}(s\mathcal{C}^{\ell})=\mathbb{E}(\mathcal{C}_0^1),$ and since a valid $\mathbb{E}(\mathcal{C}_0^1)$-controlling procedure at level $q^{(1)}$ is applied on the $p$-values of $\mathcal{F}^1,$ we have $\mathbb{E}(s\mathcal{C}^{1})\leq q^{(1)}.$ It remains to prove the result of Theorem \ref{thm-alg-3} for $\ell>1$ such that $\ell\in \mathcal{L}.$ Let us fix such $\ell.$
Let $\mathcal{C}_i^{+}$ be the
countable support of $\mathcal{C}_i^\ell$ for each family
$\mathcal{F}_i^{\ell}.$  Based on Lemma \ref{lem1}, we obtain that for any sequence of indices $i_1, i_2, \ldots, i_{\ell-1}$ such that
$i_1\in \mathcal{F}^1,$ $i_2\in
\mathcal{F}_{i_1}^2,\ldots,i_{\ell-1}\in
\mathcal{F}_{i_{\ell-2}}^{\ell-1},$ the following event on the space of $p$-values for hypotheses in the union of sets $\mathcal{F}^1\setminus\{i_1\},$ $\mathcal{F}_{i_1}^2\setminus\{i_2\},\ldots, \mathcal{F}_{i_{\ell-2}}^{\ell-1}\setminus\{i_{\ell-1}\}$ can be defined: if $H_{i_1},$ $H_{i_2},\ldots, H_{i_{\ell-1}}$ are selected, then the numbers of selected hypotheses in families $\mathcal{F}^1, \mathcal{F}_{i_1}^2,\ldots, \mathcal{F}_{i_{\ell-2}}^{\ell-1}$ are $r_{i_1}, r_{i_2},\ldots, r_{i_{\ell-1}},$ respectively. We denote this event by $C_{r_{i_1},\ldots, r_{i_{\ell-1}}}^{(i_1,\ldots, i_{\ell-1})}.$ Then
\begin{align*}
&\mathbb{E}(s\mathcal{C}^{\ell})=\\&\sum_{i_1\in
\mathcal{F}^1}\sum_{r_{i_1}=1}^{|\mathcal{F}^1|}\sum_{i_2\in
\mathcal{F}_{i_1}^2}\sum_{r_{i_2}=1}^{|\mathcal{F}_{i_1}^2|}\cdots\!\!\!\!\!\!
\sum_{i_{\ell-1}\in
\mathcal{F}_{i_{\ell-2}}^{\ell-1}}\sum_{r_{i_{\ell-1}}=1}^{|\mathcal{F}_{i_{\ell-2}}^{\ell-1}|}\frac{\displaystyle\sum_{c\in
\mathcal{C}_{i_{\ell-1}}^{+}}\!\!\!\!c\mathbb{P}(\mathcal{C}_{i_{\ell-1}}^{\ell}\!\!=c,
i_k\in \mathcal{S}_{i_{k-1}}^{k} \text{ for } k=1,\ldots, \ell-1, C_{r_{i_1}, \ldots, r_{i_{\ell-1}}}^{(i_1,\ldots, i_{\ell-1})})}{\prod_{k=1}^{\ell-1}r_{i_k}}.
\end{align*}
The TreeC procedure does not make any rejections in families that
are not selected, therefore if the family
$\mathcal{F}_{i_{\ell-1}}^{\ell}$ has an ancestor hypothesis that is
not selected, $\mathcal{C}_{i_{\ell-1}}^{\ell}=0.$  Therefore, for
this procedure we obtain
\begin{align}
&\mathbb{E}(s\mathcal{C}^{\ell})=\notag\\& \sum_{i_1\in
\mathcal{F}^1}\sum_{r_{i_1}=1}^{|\mathcal{F}^1|}\sum_{i_2\in
\mathcal{F}_{i_1}^2}\sum_{r_{i_2}=1}^{|\mathcal{F}_{i_1}^2|}\cdots\!\!\!\!\!\!
\sum_{i_{\ell-1}\in
\mathcal{F}_{i_{\ell-2}}^{\ell-1}}\sum_{r_{i_{\ell-1}}=1}^{|\mathcal{F}_{i_{\ell-2}}^{\ell-1}|}\frac{\displaystyle\sum_{c\in
\mathcal{C}_{i_{\ell-1}}^{+}}c\mathbb{P}(\mathcal{C}_{i_{\ell-1}}^{\ell}=c,
 C_{r_{i_1}, \ldots, r_{i_{\ell-1}}}^{(i_1,\ldots, i_{\ell-1})})}{\prod_{k=1}^{\ell-1}r_{i_k}}\label{indep}=\\
&\sum_{i_1\in
\mathcal{F}^1}\sum_{r_{i_1}=1}^{|\mathcal{F}^1|}\sum_{i_2\in
\mathcal{F}_{i_1}^2}\sum_{r_{i_2}=1}^{|\mathcal{F}_{i_1}^2|}\cdots\!\!\!\!\!\!
\sum_{i_{\ell-1}\in
\mathcal{F}_{i_{\ell-2}}^{\ell-1}}\sum_{r_{i_{\ell-1}}=1}^{|\mathcal{F}_{i_{\ell-2}}^{\ell-1}|}\frac{\displaystyle\sum_{c\in
\mathcal{C}_{i_{\ell-1}}^{+}}c\mathbb{P}(\mathcal{C}_{i_{\ell-1}}^{\ell}=c)\mathbb{P}( C_{r_{i_1}, \ldots, r_{i_{\ell-1}}}^{(i_1,\ldots, i_{\ell-1})})}{\prod_{k=1}^{\ell-1}r_{i_k}}=
\notag\\&\sum_{i_1\in
\mathcal{F}^1}\sum_{r_{i_1}=1}^{|\mathcal{F}^1|}\sum_{i_2\in
\mathcal{F}_{i_1}^2}\sum_{r_{i_2}=1}^{|\mathcal{F}_{i_1}^2|}\cdots\!\!\!\!\!\!
\sum_{i_{\ell-1}\in
\mathcal{F}_{i_{\ell-2}}^{\ell-1}}\sum_{r_{i_{\ell-1}}=1}^{|\mathcal{F}_{i_{\ell-2}}^{\ell-1}|}\frac{\mathbb{E}(\mathcal{C}_{i_{\ell-1}}^{\ell})\mathbb{P}(C_{r_{i_1}, \ldots, r_{i_{\ell-1}}}^{(i_1,\ldots, i_{\ell-1})})}{\prod_{k=1}^{\ell-1}r_{i_k}}\label{last}
\end{align}
The equality in (\ref{indep}) follows from the independence between
the $p$-values of hypotheses in $\mathcal{F}_i^{\ell}$ and the $p$-values of
hypotheses in $\cup_{k=1}^{\ell-1}\left(\mathcal{F}_{P^k(i)}^{\ell-k}\setminus \{P^{k-1}(i)\}\right),$ where $P^0(i)=i,$
for each family $\mathcal{F}_i^{\ell}.$
In expression (\ref{last}), $\mathcal{C}_{i_{\ell-1}}^{\ell}$ is the
value of the random measure of errors assigned to family
$\mathcal{F}_{i_{\ell-1}}^{\ell}$ (where
$\mathcal{F}_{i_{0}}^{1}=\mathcal{F}^1$) when a valid
$\mathbb{E}(\mathcal{C}_{i_{\ell-1}}^{\ell})$-controlling procedure at level
$\frac{\prod_{k=1}^{\ell-1}r_{i_k}}{\prod_{k=1}^{\ell-1}|\mathcal{F}_{i_{k-1}}^k|}q^{(\ell)}
$ is applied in each selected family residing at Level $\ell.$ Since
$\mathcal{C}_{i}^{\ell}=0$ for families that are not selected, we
obtain that $\mathbb{E}(\mathcal{C}_{i}^{\ell})\leq
\frac{\prod_{k=1}^{\ell-1}r_{i_k}}{\prod_{k=1}^{\ell-1}|\mathcal{F}_{i_{k-1}}^k|}
q^{(\ell)} $ for \textit{each} family $\mathcal{F}_i^{\ell}.$ Using
this inequality and the fact that for each sequence $i_1,
i_2,\ldots, i_{\ell-1}$ such that $i_1\in \mathcal{F}^1,$ $i_2\in
\mathcal{F}_{i_1}^2,$ \ldots, $i_{\ell-1}\in
\mathcal{F}_{i_{\ell-2}}^{\ell-1}$
$$\sum_{r_{i_1}=1}^{|\mathcal{F}^1|}\sum_{r_{i_2}=1}^{|\mathcal{F}_{i_1}^2|}\ldots \sum_{r_{i_{\ell-1}}=1}^{|\mathcal{F}_{i_{\ell-2}}^{\ell-1}|}\mathbb{P}(C_{r_{i_1}, \ldots, r_{i_{\ell-1}}}^{(i_1,\ldots, i_{\ell-1})})=1,$$ we obtain
\begin{align*}&\sum_{i_1\in
\mathcal{F}^1}\sum_{r_{i_1}=1}^{|\mathcal{F}^1|}\sum_{i_2\in
\mathcal{F}_{i_1}^2}\sum_{r_{i_2}=1}^{|\mathcal{F}_{i_1}^2|}\ldots
\sum_{i_{\ell-1}\in
\mathcal{F}_{i_{\ell-2}}^{\ell-1}}\sum_{r_{i_{\ell-1}}=1}^{|\mathcal{F}_{i_{\ell-2}}^{\ell-1}|}\frac{\mathbb{E}(\mathcal{C}_{i_{\ell-1}}^{\ell})\mathbb{P}(C_{r_{i_1}, \ldots, r_{i_{\ell-1}}}^{(i_1,\ldots, i_{\ell-1})})}{\prod_{k=1}^{\ell-1}r_{i_k}}\leq\\&\sum_{i_1\in
\mathcal{F}^1}\sum_{i_2\in \mathcal{F}_{i_1}^2}\cdots\!\!\!\!\!\!
\sum_{i_{\ell-1}\in
\mathcal{F}_{i_{\ell-2}}^{\ell-1}}\frac{\prod_{k=1}^{\ell-1}r_{i_k}
}{\prod_{k=1}^{\ell-1}|\mathcal{F}_{i_{k-1}}^k|}q^{(\ell)}\frac{1}{\prod_{k=1}^{\ell-1}r_{i_k}}\sum_{r_{i_1}=1}^{|\mathcal{F}^1|}\sum_{r_{i_2}=1}^{|\mathcal{F}_{i_1}^2|}\cdots\!\!\!\!\sum_{r_{i_{\ell-1}}=1}^{|\mathcal{F}_{i_{\ell-2}}^{\ell-1}|}\!\!\mathbb{P}(C_{r_{i_1}, \ldots, r_{i_{\ell-1}}}^{(i_1,\ldots, i_{\ell-1})})=\\
&q^{(\ell)}\sum_{i_1\in
\mathcal{F}^1}\frac{1}{|\mathcal{F}^1|}\sum_{i_2\in
\mathcal{F}_{i_1}^2}\frac{1}{|\mathcal{F}_{i_1}^2|}\ldots
\sum_{i_{\ell-1}\in
\mathcal{F}_{i_{\ell-2}}^{\ell-1}}\frac{1}{|\mathcal{F}_{i_{\ell-2}}^{\ell-1}|}=q^{(\ell)}.
\end{align*}
Combining this result with the result in (\ref{last}) we obtain
$\mathbb{E}(s\mathcal{C}^{\ell})\leq q^{(\ell)}.$ 
\end{proof}
\begin{proof}[Proof of Lemma \ref{lem1}]
Let $\ell_1<\ell_2<\ldots<\ell_{|\mathcal{L}|}$ be the ordered sequence of all values in $\mathcal{L}$ that are greater than 1.  For $\ell=\ell_1,$ the result of
Lemma \ref{lem1} follows immediately from the fact that the selection rules applied on families at levels $1,\ldots, \ell_1-1$ are simple selection rules. Let us prove the result of Lemma \ref{lem1} for $\ell=\ell_2.$ Let $i_1, i_2, \ldots, i_{\ell_2-1}$ be a sequence of indices  such that $i_1\in \mathcal{F}^1,$ $i_2\in
\mathcal{F}_{i_1}^2,\ldots,i_{\ell_2-1}\in
\mathcal{F}_{i_{\ell_2-2}}^{\ell_2-1}.$ Assume that the $p$-values for $H_{i_1},$ $H_{i_2},\ldots, H_{i_{\ell_2-1}}$ can change as long as $H_{i_1},$ $H_{i_2},\ldots, H_{i_{\ell_2-1}}$ are selected, while the $p$-values of hypotheses in the sets $\mathcal{F}^1\setminus\{i_1\},$ $\mathcal{F}_{i_1}^2\setminus\{i_2\},\ldots, \mathcal{F}_{i_{\ell_2-2}}^{\ell_2-1}\setminus\{i_{\ell_2-1}\}$ are fixed. Then the number of selected hypotheses in each family $\mathcal{F}_{i_{k-1}}^{k}$ where $k\leq \ell_1-1$ or $\ell_1+1\leq k\leq \ell_2-1$ remains unchanged, as shown above. Let us denote the number of selected hypotheses in $\mathcal{F}_{i_{k-1}}^{k}$ by $r_{i_k},$ for $k\leq \ell_1-1.$
For selecting hypotheses in $\mathcal{F}^{\ell_1}_{i_{\ell_1-1}},$ an $\mathbb{E}(\mathcal{C}_{i_{\ell_1-1}}^{\ell_1})$-controlling procedure at level
$\frac{\prod_{k=1}^{\ell_1-1}{r_{i_k}}}{\prod_{k=1}^{\ell_1-1}|\mathcal{F}^{k}_{i_{k-1}}|}q^{(\ell_1)}$ is applied on the $p$-values of
$\mathcal{F}^{\ell_1}_{i_{\ell_1-1}}.$ Since this procedure defines a simple selection rule, the number of rejections in  $\mathcal{F}^{\ell_1}_{i_{\ell_1-1}}$ remains unchanged. Thus we have proven the result of Lemma \ref{lem1} for $\ell=\ell_2.$ 
The proof for $\ell=\ell_3,\ldots,\ell_{|\mathcal{L}|}$ follows from using the above arguments repeatedly.
\end{proof}
\begin{proof}[Proof of Theorem \ref{thm2-gen}]
The results follow from the arguments similar to those used in the proof of Theorem 3 in the main manuscript.  For any sequence of indices $(i_1, i_2, \ldots, i_{L})$ and sequence of numbers $(r_1, \ldots, r_{L})$ such that
$i_1\in \mathcal{F}^1,$ $i_2\in
\mathcal{F}_{i_1}^2,\ldots,i_{L}\in
\mathcal{F}_{i_{L-1}}^{L},$ and $r_k\in\{1, \ldots, |\mathcal{F}_{i_{k-1}}^{k}| \},$ the event $C_{r_1, \ldots, r_{L}}^{(i_1, \ldots, i_{L})}$
defined in (\ref{Cr}) should be defined now as
$$C_{r_{1},\ldots, r_{L}}^{(i_1,\ldots, i_{L})}=\left(\cap_{k\in \mathcal{L}}C_{r_{k}}^{(i_k)}[q_{i_{k-1}}]\right)\cap \left(\cap_{k\notin \mathcal{L}}C_{r_{k}}^{(i_k)}\right),
$$
where the events $C_r^{(j)}$
and $C_r^{(j)}[\alpha]$ are defined before the statement of Theorem \ref{thm2-gen} and $q_{i_{k-1}}$ is defined in (\ref{q_i}). Using this definition of $C_{r_{1},\ldots, r_{L}}^{(i_1,\ldots, i_{L})},$ we define the events
$B_s^{(i_1, \ldots, i_{L})}$ and $A_s^{(i_1, \ldots, i_L)}$ as in (\ref{Bevent}) and (\ref{at}) respectively. It can be easily verified that the results of Lemma \ref{lemtreeBH} hold for $\ell=L,$ therefore it is enough to prove inequality (\ref{enough}) for $\ell=L,$ as shown in the proof of Theorem 1 in the main manuscript.
 Using the arguments used in the proof of Lemma \ref{incr}, we obtain that for each sequence $(i_1, \ldots, i_{L})$ such that $i_1\in \mathcal{F}^1,$ $i_2\in
\mathcal{F}_{i_1}^2,\ldots,i_{L}\in
\mathcal{F}_{i_{L-1}}^{L},$ and $s\in\{1, \ldots, \prod_{k=1}^{L}|\mathcal{F}_{i_{k-1}}^{k}|\},$
\begin{align}
A_s^{(i_1, \ldots, i_{L})}=\bigcup_{t=1}^{s}\bigcup_{(r_1, \ldots, r_{L})\in I_t}\left\{\left(\cap_{k\in\mathcal{L}}D_{r_k}^{(i_k)}[q_{i_{k-1}}]\right)\cap \left(\cap_{j\notin\mathcal{L}}D_{r_j}^{(i_j)}\right)\right\},\label{as}
\end{align}
where the events $D_k^{(j)}$ and $D_k^{(j)}[\alpha]$ were defined in (\ref{d1}) and (\ref{d2}) respectively. Since the BH procedure is concordant, and the selection rules for families at levels $\ell=1, \ldots, L-1$ are concordant as well, the sets $D_{r_k}^{(i_k)}[q_{i_{k-1}}]$ and $D_{r_j}^{(i_j)}$ in (\ref{as}) are increasing sets of their input $p$-values, for each $k\in \mathcal{L}$ and $j\notin\mathcal{L}.$ Using the fact that each of the combination rules used for computing the $p$-values for parent hypotheses at levels $\ell=1, \ldots, L-1$ is a monotone non-decreasing function of each of its input $p$-values, and following the arguments of the proof of item 1 of Theorem 3 in the main manuscript, we obtain that the sets $D_{r_k}^{(i_k)}[q_{i_{k-1}}]$ and $D_{r_j}^{(i_j)}$ in (\ref{as}) are increasing sets of $\bm p^{L},$ which yields that the set $A_s^{(i_1, \ldots, i_{L})}$ is an increasing set of  $\bm p^{L}.$ Now the result of item 1 follows using the same arguments as in the proof of item 1 of Theorem 3 in the main manuscript. The result of item 2 follows from Proposition 1 and from the result of item 1.

\end{proof}
\newpage
\beginsupplement
\section*{Details of GWAS simulation}

We adapt the simulation study proposed in Lewin et al. (2015) \cite{MTHESS} for eQTL studies.
 Specifically, we assume the same underlying model, with $\mathbf{X}_{n\times p}$ the matrix of additively coded values for each of $p$ genetic variants in $n$ subjects.
 The ${n\times q}$ matrix $\mathbf{Y}_l$ represents the simulated gene expression values for tissue $l$, $1 \leq l \leq L$, and is generated from the $\mathbf{X}$ matrix via the linear model
\begin{equation} \label{mthess_model}
\mathbf{Y}_l  = \mathbf{X}\mathbf{B}_l + \mathbf{E}_l + \mathbf{E}_{shared},
 \end{equation} where $\mathbf{B}$ is a matrix of regression coefficients which is shared across tissues, $\mathbf{E}_l$ is a matrix of  tissue-specific residuals with entries
 $e_{ikl} \sim \mathcal{N}(0, \sigma^2_l)$, and $\mathbf{E}_{shared}$ is a matrix of residuals shared across tissues with entries $e_{ik_{shared}} \sim \mathcal{N}(0, \sigma^2_{shared})$.
 Following \cite{MTHESS}, we assume the same pattern of association within each tissue (i.e.\ the coefficient matrix $\mathbf{B}_l = \mathbf{B}$ is common across tissues).
 Note that this implies that  the Level 3 families are homogenous. The total variance of the residuals for tissue $l$ is then $\sigma^2_{total} = \sigma^2_{shared} + \sigma^2_l$. The entries in $\mathbf{B}$ are generated following the equation $\beta_{kj} = \lambda_{kj} \times \gamma_{kj}$, where the magnitude of the coefficients $\lambda_{kj} \sim \mathcal{N}(\mu, 0.001^2)$ is controlled by the signal strength parameter $\mu$ and $\gamma_{kj}$ is a binary indicator (i.e.\ $\gamma_{kj}$ determines the pattern of association, while $\lambda_{kj}$ controls the signal strength).

To better mimic a real multi-tissue eQTL study, we increase the dimensionality vs.\ that of \cite{MTHESS} (who consider a rat dataset with $n$ = 29 and $p$ = 1304 as their predictors and simulate the expression of $q$ = 150 genes in $L$ = 3 tissues), and take $\mathbf{X}_{n\times p}$ to be the chromosome 17 genotype data from the North Finland Birth Cohort (NFBC) data set ($n$ = 5402 and $p$ = 8713), where missing SNP values are replaced by the column mean and the data have been standardized. Given this genotype matrix, we then simulate the expression of $q$ = 250 genes in each of $L$ = 5 tissues following equation \eqref{mthess_model}.

 Following the scenario in \cite{MTHESS} with correlated noise, balanced across tissues, we set $\sigma^2_l = 0.0084$ and $\sigma^2_{shared} = 0.0016$, resulting in total noise standard deviation $\sigma_{total}$ = 0.1. Varying signal-to-noise ratios are achieved by varying the signal strength parameter $\mu$ (i.e.\ the mean of the normal distribution from which the non-zero entries in $\mathbf{B}$ are drawn). We assume that there are 50 causal SNPs which are each associated to the expression of 5, 10, or 20 genes. The ``true'' SNPs are selected at random in each iteration, but the configuration of which phenotype the ``true'' SNPs affect is held constant. This configuration is illustrated in Figure \ref{gamma_config}.

All $p$ = 8713 SNPs are tested for association to the simulated gene
expression traits. Association $p$-values are obtained using Matrix
eQTL \cite{MatrixEQTL}, including the first 5 principal components
of the genotypes as covariates to account for population structure.
Discoveries are considered to be true positives as long as the
discovered SNP is within 1Mb and has correlation of magnitude at
least 0.2 with a SNP that is truly relevant to the trait. Given this
definition, we compare the performance of BH, BB, 
and TreeBH,
where BB is applied with SNPs in Level 1 and all traits (i.e.\
expression for all genes across all tissues) in Level 2, 
and TreeBH is applied with SNPs in Level 1, genes in Level 2, and
tissues in Level 3. The p-filter is not included here as it was not
able to run in a timely manner (i.e.\ < 24 hrs) given similar
groupings to those in the  simulation with independent hypotheses.
All methods are applied with target level 0.05. A performance
comparison for the 4 methods compared across 25 simulated data sets
is given in Figure \ref{sim2_results}: the TreeBH 
procedure appears to at least approximately control the target error
rates, unlike BH, and has a similar FDR, sFDR, and power to BB. In
particular, we see that BB is close to controlling the target error
rates, unlike in the simulations from Section 5.2. This is
likely due to the fact that the adjustment for step 1 selection is
substantially more stringent here because of the sparsity of signal
across SNPs (specifically, the fact that only 50 of 8713 Level 1
hypotheses are non-null).

 \begin{figure}[h!]
 \includegraphics[trim = {6cm 12cm 5cm 10cm}, clip, width = .9\linewidth]{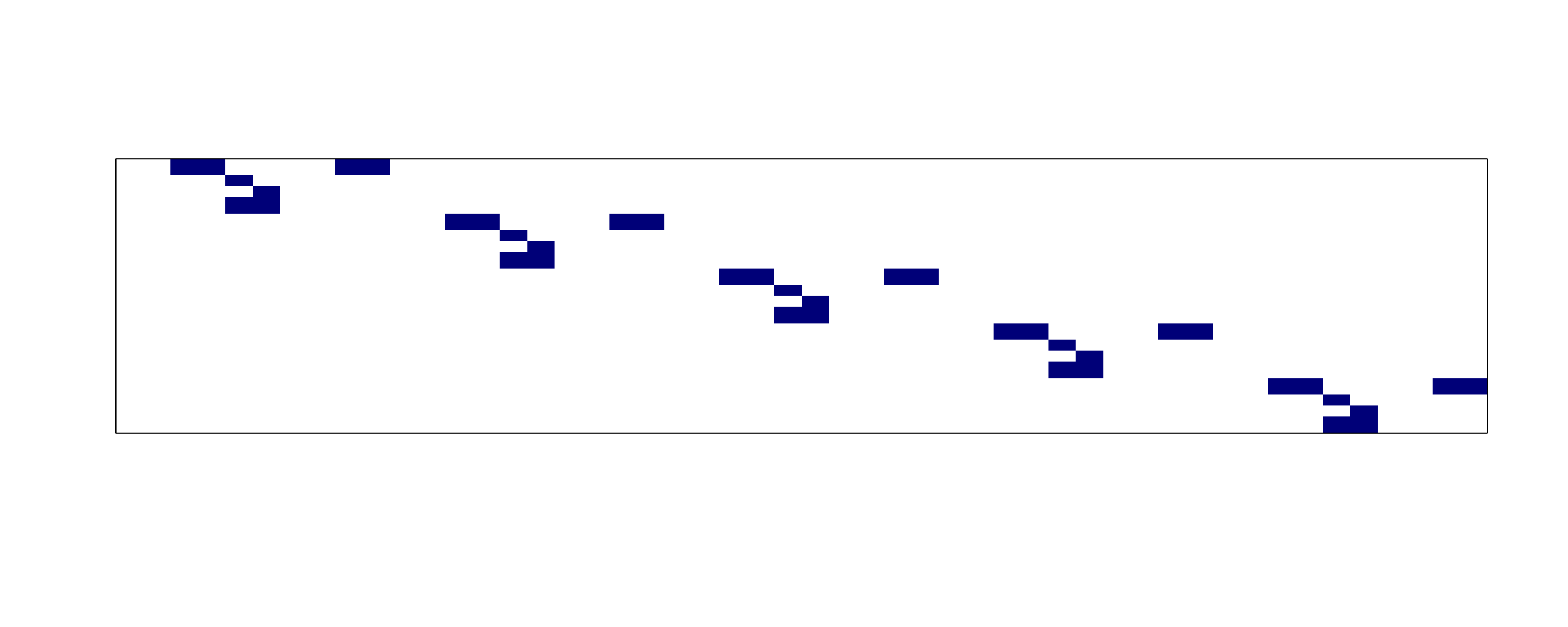}

\caption{\em Pattern of nonzero coefficients across 50 causal SNPs (rows) and 250 simulated traits (columns) for a single tissue for the simulation described in Section 5.3. Simulated traits with no genetic basis are included in the illustration (i.e.\ blank columns), but non-causal SNPs (i.e.\ blank rows, which are in fact much more numerous than the causal SNPs) are omitted. Note that in each  iteration of the simulation, the causal SNPs are selected at random, rather than being taken as a contiguous ordered block.  }
\label{gamma_config}
 \end{figure}

 \begin{figure}[h!]
\centering
\includegraphics[width = 0.9\linewidth, trim={3cm 0 0 0}, clip]{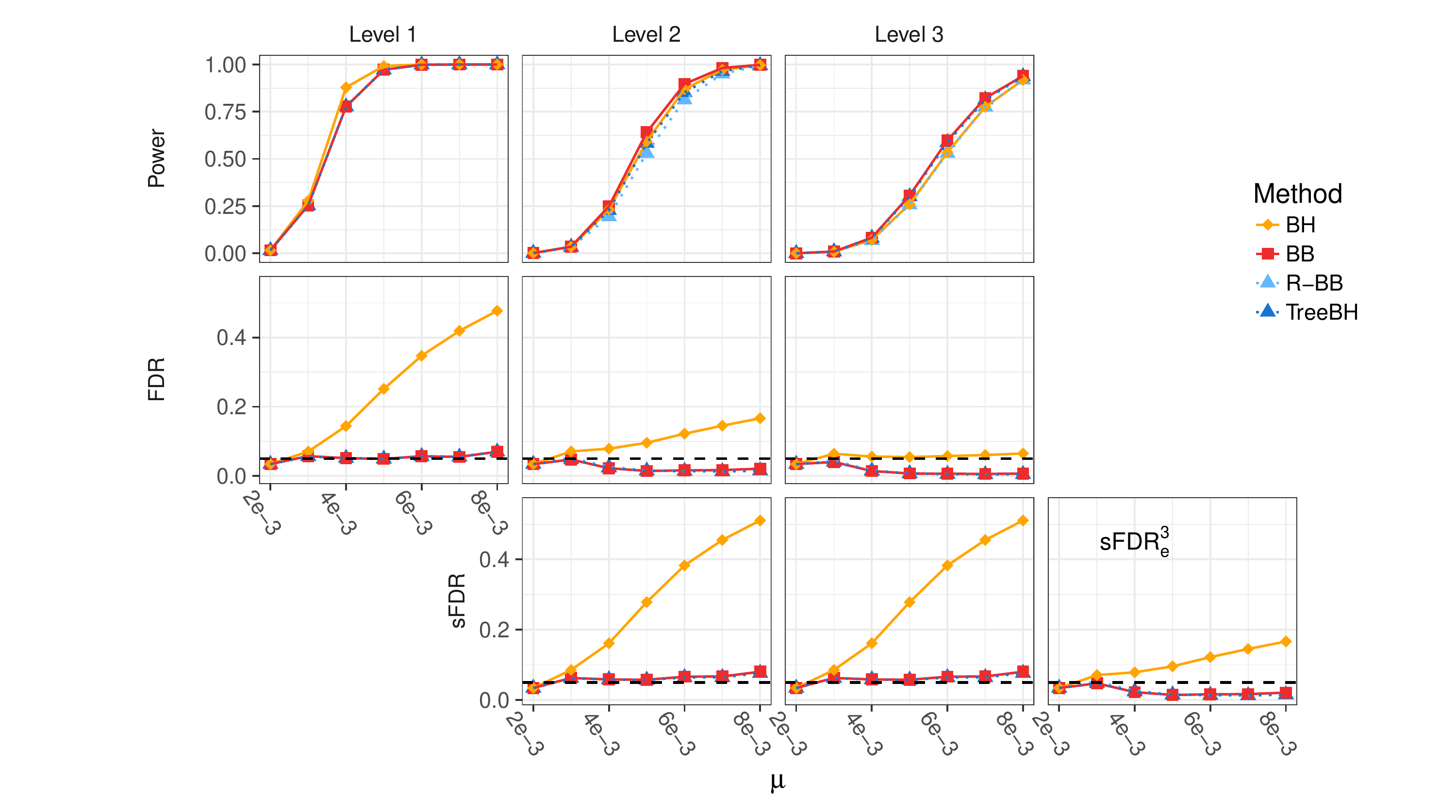}
\caption{\em Results for Example 5.3: multi-trait GWAS. Each point corresponds to the average of  25 simulated data sets using the full set of $p$ = 8713 SNPs on chromosome 17.}
\label{sim2_results}
\end{figure}

\begin{figure}[h!]
\centering
\includegraphics[width = \linewidth]{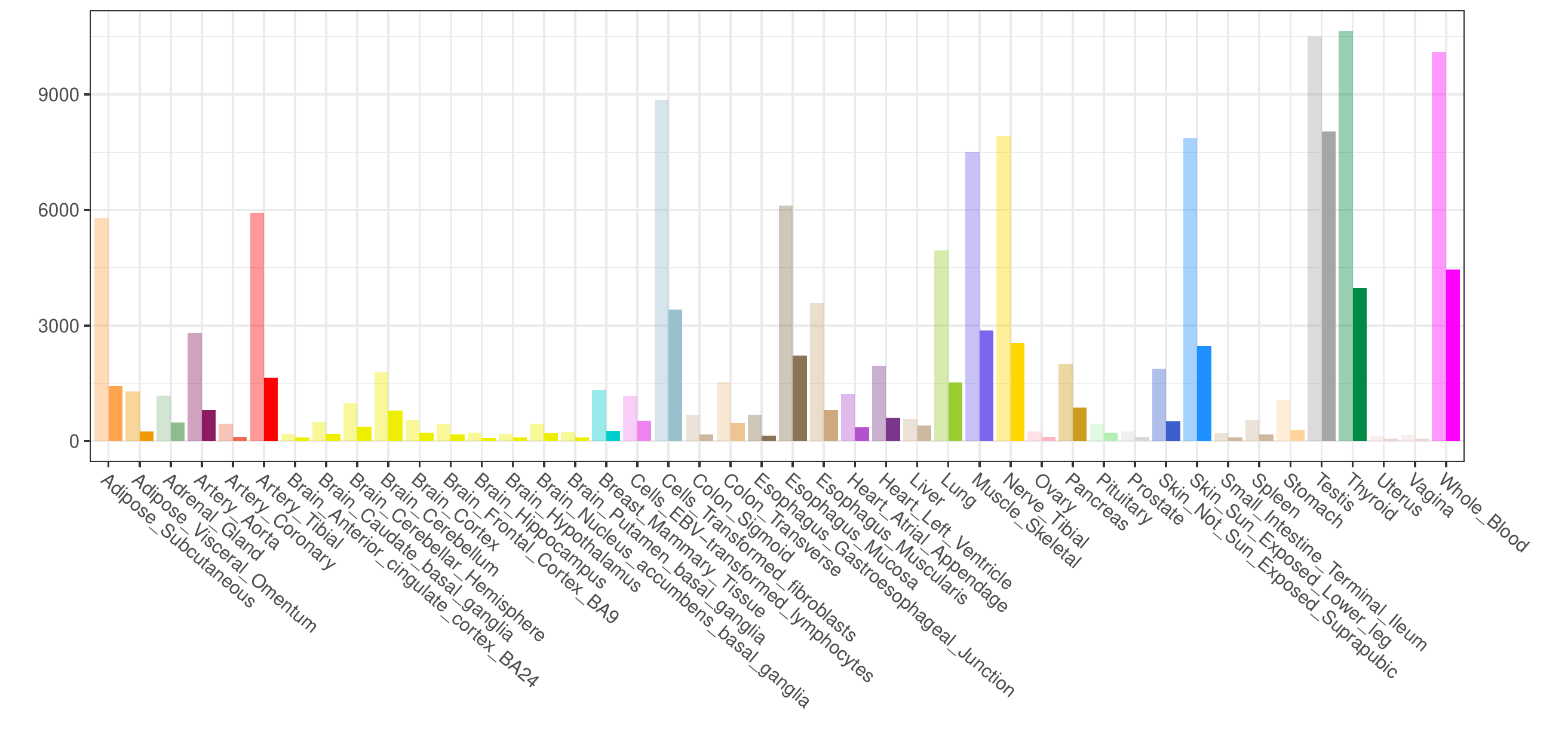} \\
\caption[]{\em Comparison of selection results for the BH procedure
applied separately by tissue (transparent) and the 3-level TreeBH procedure
(solid) in terms of the number of tissue-specific SNP-gene pairs per
tissue i.e.\ the number of SNP-gene pairs that were discovered only
in the given tissue. } \label{gtex_tissue_spec}
\end{figure}



\end{document}